\documentclass[preprint,3p,times]{elsarticle}

\usepackage{amssymb}
\usepackage{amsthm}
\usepackage{graphicx}
\usepackage{dcolumn}
\usepackage{bm}
\usepackage{subfigure}
\usepackage{amssymb}
\usepackage{verbatim}
\usepackage{amscd}
\usepackage{amssymb}
\usepackage{amsmath, amsfonts}
\usepackage{setspace}
\usepackage[]{graphicx}        
\usepackage{amsthm}
\usepackage{natbib}
\usepackage{enumerate}
\biboptions{sort&compress} 

\theoremstyle{plain}
\newtheorem{theorem}{Theorem}
\theoremstyle{plain}
\newtheorem{lemma}{Lemma}
\theoremstyle{plain}
\newtheorem{sublemma}{Sublemma} 
\theoremstyle{plain}

\theoremstyle{remark}

\theoremstyle{conjecture}

\theoremstyle{observation}

\theoremstyle{definition}

\theoremstyle{corollary}
\newtheorem{corollary}{Corollary}
\theoremstyle{definition}

\theoremstyle{definition}
\newtheorem{definition}{Definition}
\theoremstyle{assumption}

\theoremstyle{definition}
\newtheorem*{theorem_num}{Theorem}
\theoremstyle{problem}

\theoremstyle{fact}

\begin{document}


\begin{abstract}
Recently, it has become apparent that the thermal stability of topologically ordered systems at finite temperature, as discussed in condensed matter physics, can be studied by addressing the feasibility of self-correcting quantum memory, as discussed in quantum information science. Here, with this correspondence in mind, we propose a model of quantum codes that may cover a large class of physically realizable quantum memory. The model is supported by a certain class of gapped spin Hamiltonians, called stabilizer Hamiltonians, with translation symmetries and a small number of ground states that does not grow with the system size. We show that the model does not work as self-correcting quantum memory due to a certain topological constraint on geometric shapes of its logical operators. This quantum coding theoretical result implies that systems covered or approximated by the model cannot have thermally stable topological order, meaning that no system can be stable against both thermal fluctuations and local perturbations simultaneously in two and three spatial dimensions. 
\end{abstract}

\begin{frontmatter}

\title{Feasibility of self-correcting quantum memory and thermal stability of topological order}

\author{Beni Yoshida}
\address{Center for Theoretical Physics, Massachusetts Institute of Technology, Cambridge, Massachusetts 02139, USA}
\ead{rouge@mit.edu}


\end{frontmatter}

\tableofcontents

\section{Introduction}\label{sec:introduction}

In recent years, ideas from quantum information science have become increasingly useful in condensed matter physics where quantum information theoretical viewpoints give an important insight on studies of many-body entanglement arising in ground states or quasi-particle excitations of correlated spin systems~\cite{Bravyi06, Vidal03, Fradkin06, Osborne02, Li08, Kitaev06, Levin06, Kitaev03, Eisert10, Vidal07, Gu08}. 
In particular, it has been realized that many interesting physical systems in condensed matter physics may be described in the language of quantum codes such as stabilizer codes and its generalization~\cite{Kitaev97, Raussendorf03, Hein04, Bombin09b, Bacon06}.
The emerging closeness between two fields has made it possible to study several problems concerning many-body correlated spin systems through quantum coding theoretical tools~\cite{Beni10b}. 

In the present paper, we shall explore such a capacity of quantum coding theory for studying problems in condensed matter physics further, by making use of the correspondence between the following two open problems:
\begin{enumerate}[(a)]
\item Feasibility of self-correcting quantum memory.
\item Thermal stability of topological order.
\end{enumerate}

Feasibility of self-correcting quantum memory is an important open problem in quantum information science concerning reliable storage of qubits. Thermal stability of topological order at finite temperature is an open problem of fundamental importance in condensed matter physics which concerns whether topological order may survive at finite temperature or not. While these two problems may look very different from each other, they are fundamentally akin to each other, in a sense that by searching for self-correcting quantum memory, one can search for topological ordered spin systems which are stable at finite temperature~\cite{Dennis02, Nussinov08, Bravyi09, Alicki09, Alicki10}. This surprising correspondence enable us to address a condensed matter theoretical question, thermal stability of topological order, by analyzing a quantum information theoretical question, feasibility of self-correcting quantum memory.

\begin{figure}[htb!]
\centering
\includegraphics[width=0.65\linewidth]{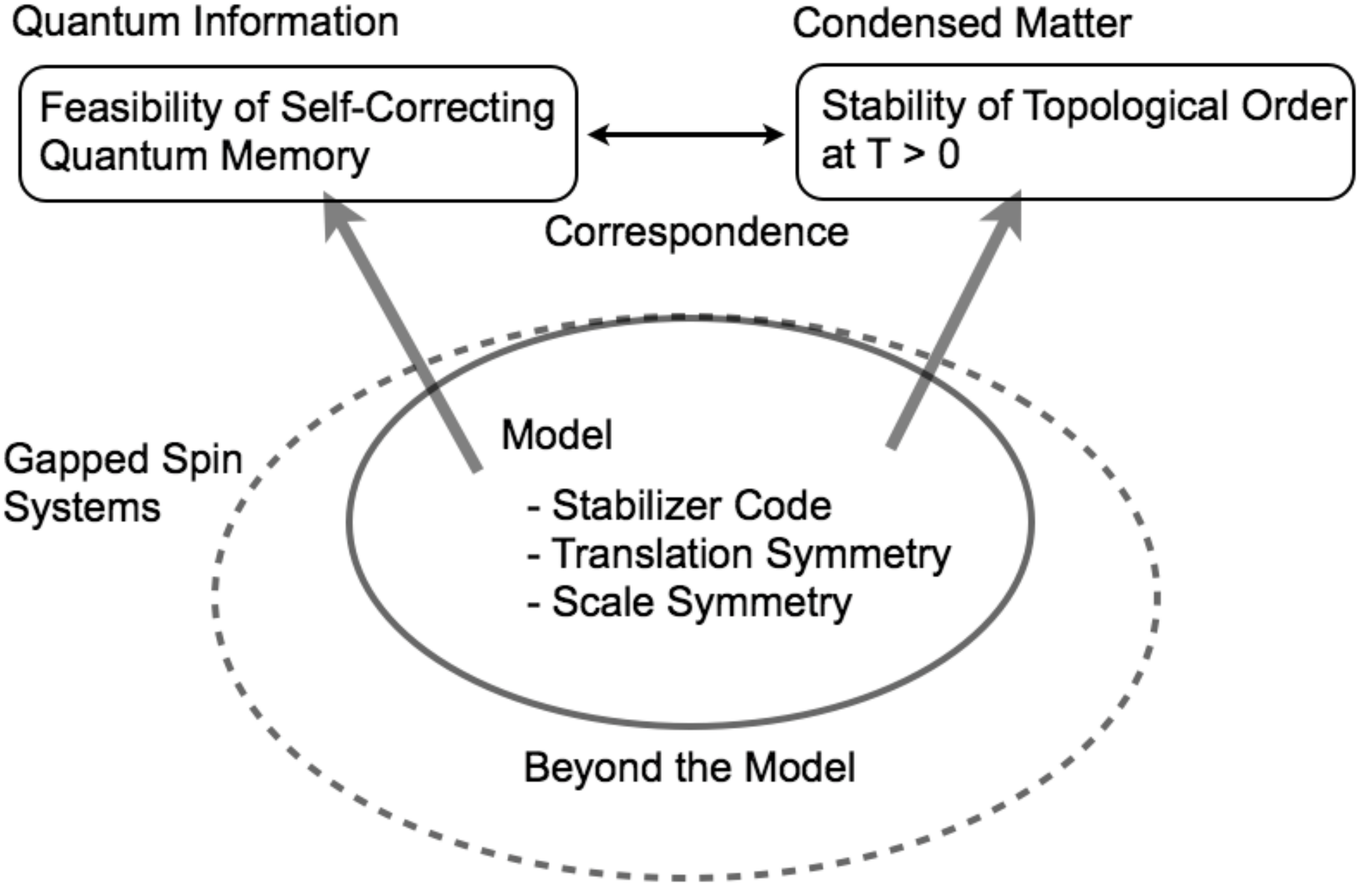}
\caption{The correspondence between the feasibility of self-correcting quantum memory and the thermal stability of topological order, and the plan of the paper. 
} 
\label{fig_plan}
\end{figure}

The main contributions of the present paper toward these two open problems are the followings;
\begin{itemize}
\item We propose a certain model of quantum codes which may cover a large class of physically realizable quantum codes, which is called a stabilizer code with translation and scale symmetries. 
\item We solve the $D$-dimensional model exactly by determining its coding properties completely and analyze the feasibility of self-correcting quantum memory for $D=1,2,3$. 
\item We establish the connection between self-correcting quantum memory and stable topological order at finite temperature, and analyze thermal stability of topological order arising in the model.  
\end{itemize}

The plan of the paper is schematically summarized in Fig.~\ref{fig_plan}. We begin the introduction of the paper by reminding why these two problems are important in each field. Then, we give brief summaries of our results in the rest of the introduction.

\textbf{Feasibility of self-correcting quantum memory:} Quantum entanglement decays easily. This underlying difficulty in quantum information science gave birth to the beautiful art of protecting qubits from decoherence; \emph{quantum coding theory}. The central idea of quantum error- correcting codes is to encode a qubit in many-body entangled states, and perform error-corrections so that encoded qubits are not lost. After discoveries of first examples of quantum codes~\cite{Shor95, Laflamme96, Bennett96, Calderbank96, Kitaev97} which culminated in stabilizer codes~\cite{Gottesman96}, a large number of quantum codes have been found. Now, quantum coding theory constitutes one of the most important building blocks for realizing fault-tolerant quantum computation~\cite{Shor96}.

Yet, there still remain important gaps between theoretical constructions of quantum codes and their physical realizations as quantum memory devices. First of all, most quantum coding schemes encode qubits \emph{dynamically} by applying a large number of logical gates, and encoding is discussed only in terms of the Hilbert space. However, since encoded qubits will be eventually lost in the presence of interactions with the external environment, it is desirable to store logical qubits \emph{statically} in some physical subspaces which are naturally protected from decoherence. 

One plausible approach may be to store logical qubits in the gapped ground space of some quantum many-body system. In this light, stabilizer codes constructed with geometrically local generators are promising candidates for physical realizations of quantum codes since such \emph{local stabilizer codes} can be realized as the ground space of gapped Hamiltonians by using their local generators as interaction terms. Indeed, from purely theoretical viewpoint, it is known that sufficiently accurate and frequent error-corrections can store logical qubits reliably~\cite{Dennis02}. However, one may still dream of having a quantum memory device which would work without active error-corrections, given the difficulties and inefficiencies of performing fast and accurate error-corrections in reality\footnote{To the best of our knowledge, there have been no convincing experimental demonstration of error-corrections. All the demonstrations are limited only to phase-type errors, and do not correct bit-type errors. Also, it seems very difficult to reach the frequency threshold for reliable storage of logical qubits~\cite{Dennis02}. Finally, it is extremely inefficient to keep performing error-corrections during the time one is storing a qubit if one stores a qubit for a long duration such as months or years.}. 

Self-correcting quantum memory is an ideal memory device which corrects errors by itself~\cite{Dennis02, Bacon06, Nussinov08, Bravyi09, Alicki09, Pastawski09}. Due to the large energy barrier separating degenerate ground states, natural thermal dissipation processes restore the system into the original encoded states by correcting errors automatically without any active error-correction. If such a memory device could exist, it will be a perfect quantum information storage device which may be used commercially in the future. Also, the reliable storage of qubits seems to be the starting point for building scalable quantum computers.

There has been significant progress toward construction of self-correcting quantum memory. It has been pointed out that the Toric code defined on a four-dimensional system ($D=4$) serves as self-correcting quantum memory~\cite{Dennis02, Takeda04}. Yet, such a four-dimensional code cannot be embedded in a three-dimensional system. While there have been several proposals for three-dimensional self-correcting memory~\cite{Bacon06, Hamma09, Haah11}, validities of none of these proposals have been verified yet\footnote{We shall discuss these proposals in detail in section~\ref{sec:result1}.}. In addition, it has been shown that a self-correcting local stabilizer code cannot exist in two-dimensional systems~\cite{Bravyi09, Kay08}. It seems that two-dimensional spin systems are not promising as resource for self-correcting quantum memory. Now, the question we want to address can be summarized as follows:
\begin{itemize}
\item Is it possible to have three-dimensional self-correcting quantum memory?
\end{itemize}

\textbf{Thermal stability of topological order:}
The feasibility of self-correcting quantum memory is closely related to another important open problem in condensed matter physics; the thermal stability of topological order. 

Studies of topologically ordered systems~\cite{Laughlin83, Wen90, Kitaev97, Misguich02, Kitaev03, Kane05, Levin06, Kitaev06b, Fu07} have been frontiers of researches in condensed matter physics community, as systems with topological order are beyond the description of the Landau's symmetry-breaking paradigm which was once considered almost as ``theory of everything'' for studies of many-body systems. Topologically ordered systems are also of practical importance in quantum information science since many-body entanglement arising in ground states and quasi-particle excitations of topologically ordered systems is a primary resource for realizing various quantum information processing tasks~\cite{Kitaev97, Kitaev03}. 

The notion of topological order was originally introduced in order to characterize the stability of ground states of many-body quantum systems against local perturbations~\cite{Wen90}. Loosely speaking, a system is said to have topological order when its ground state properties do not change significantly under any types of small, but finite local perturbations. This stability of ground states against local perturbations is also valuable for quantum information processing since topologically ordered spin systems can be used as good quantum codes with macroscopic code distances~\cite{Bravyi10}. 

However, the situation changes completely when one considers the effect of thermal fluctuations on topologically ordered systems. In fact, it is known that topological order in a two-dimensional Toric code is not stable at any finite temperature which may be quantitatively seen from the fact that topological order parameters such as topological entanglement entropy vanish at any non-zero temperature at the thermodynamic limit~\cite{Castelnovo07}. A similar result is obtained in a recent numerical work on topological entanglement entropy in a spin liquid model at finite temperature~\cite{Isakov11}. It seems that topological order in a two-dimensional system is not stable at finite temperature according to general studies on the ground state properties of two-dimensional frustration-free Hamiltonians~\cite{Bravyi09, Bravyi10}. Now, the question concerning the stability of topological order can be summarized as follows\footnote{The definition of topological order at finite temperature will be clearly stated in section~\ref{sec:topo}}:  
\begin{itemize}
\item Is there any system with thermally stable topological order, which is stable against both thermal fluctuations and local perturbations?
\end{itemize}

\textbf{Main results:} Interestingly, this condensed matter theoretical question on the stability of topological order can be addressed through quantum coding theory. In fact, it has been pointed out that when topological order in correlated spin systems is stable at finite temperature, such a system can be used as self-correcting quantum memory~\cite{Bravyi09, Castelnovo08, Nussinov09}. This \emph{correspondence} between self-correcting quantum memory and the thermal stability of topological order may be better understood by identifying thermal fluctuations as random errors acting on a quantum memory. Therefore, one may take a slight liberty and say that \emph{a search for self-correcting quantum memory is equivalent to a search for a novel quantum phase with topological order which is stable at finite temperature.} 

In this paper, with this correspondence in mind, we analyze coding and physical properties of a certain model of local stabilizer codes with physically reasonable constraints~\cite{Beni10b}. The model, which is called Stabilizer code with Translation and Scale symmetries (STS model), is constrained to the following physical conditions.
\begin{itemize}
\item Qubits are defined on a $D$-dimensional square (hypercubic) lattice with periodic boundary conditions. 
\item The Hamiltonian consists only of geometrically local interaction terms with translation symmetries.
\item The number of logical qubits does not grow with the system size (scale symmetries).
\end{itemize}

We show that the model has a certain dimensional duality on geometric shapes of logical operators, as summarized in the following informal theorem: 
\begin{theorem_num}[\textbf{Dimensional Duality}]
In a $D$-dimensional STS model ($D\leq 3$), $m$-dimensional and $(D-m)$-dimensional logical operators always form anti-commuting pairs where $m$ is an integer.
\end{theorem_num}

Based on this dimensional duality on logical operators, we give answers to open questions \textbf{(a)} and \textbf{(b)}, as summarized below:
\begin{enumerate}[(a)]
\item The three-dimensional STS model does not work as self-correcting quantum memory since the energy barrier is finite for the encoding with respect to a two-dimensional logical operator. 
\item  Systems covered or approximated by the model cannot have thermally stable topological order, meaning that systems cannot be stable against both thermal fluctuations and local perturbations simultaneously in two and three spatial dimensions. 
\end{enumerate}

\textbf{Code distance:}
Our results on STS models also provide a partial answer to a long-standing open problem in quantum coding theory concerning the upper bound on the code distance of \emph{local stabilizer codes}. The code distance $d$ is a measure of the robustness of quantum codes against errors, and one of the ultimate goals in quantum coding theory is to find a quantum code with a large code distance for a fixed system size $N$ (the total number of qubits). While an upper bound on the code distance of stabilizer codes is roughly known~\cite{Knill97}, the upper bound for local stabilizer codes is currently not known yet. In other words, despite the fact that the stabilizer formalism is a canonical framework in quantum coding theory, we still do not know how robust the best local stabilizer code can be !

For a long time, it had been believed that the code distance of local stabilizer codes with $N$ qubits is upper bounded by $O(\sqrt{N})$: $d \leq O(\sqrt{N})$ at $N \rightarrow \infty$ since all the examples of local stabilizer codes ever found satisfied this upper bound~\cite{Kitaev97, Dennis02}. Later, an example of a local stabilizer code whose code distance scales as $O(\sqrt{N}\log N)$ was found~\cite{Freedman02}. While the code distance of this local stabilizer code exceeds the previously believed upper bound $O(\sqrt{N})$ \emph{logarithmically}, an example of a local stabilizer code whose code distance exceeds $O(\sqrt{N})$ \emph{polynomially} has not been found yet. Therefore, the code distance of local stabilizer codes seemed to be upper bounded by $O(N^{\frac{1}{2}+\epsilon})$ at $N \rightarrow \infty$ where $\epsilon$ is an arbitrary small positive number, although no analytical result was known on the upper bound.

It was recently proven that the code distance of local stabilizer codes is upper bounded by $O(L^{D-1})$ where $L$ is a linear length of the system and $D$ is the spatial dimension ($N \sim O(L^{D})$)~\cite{Bravyi09}. While this work does not rule out the possibilities for the existence of a local stabilizer code whose code distance exceeds $O(N^{\frac{1}{2}})$ polynomially, this bound was proven to be tight only for $D = 1,2$. Thus, whether the tight upper bound is 
\begin{align}
d \ \leq \ O(N^{\frac{1}{2}+\epsilon}) \quad  \mbox{or} \quad d \ \leq \ O(L^{D-1}) \qquad \mbox{at} \quad N \ \rightarrow \ \infty
\end{align}
for $D >2$ seems to be one of the most important open questions concerning coding properties of local stabilizer codes.

Our analysis on STS models provides the following tight upper bound on the code distance of STS models:
\begin{itemize}
\item A three-dimensional STS model has a code distance which is tightly upper bounded by $O(L)$ where $L$ is the linear length of the system. Thus, in a three-dimensional system, the upper bound on code distances turns out to be more strict than  not only $O(L^{2})$, but also $O(\sqrt{N})$.
\end{itemize}

\textbf{Organization:} The paper is organized as follows. In section~\ref{sec:model}, we give a brief review of stabilizer codes and introduce STS models. In section~\ref{sec:result}, we discuss the feasibility of self-correcting quantum memory. In section~\ref{sec:topo}, we discuss the thermal stability of topological order at finite temperature. In section~\ref{sec:physics}, we present a certain universal property of logical operators shared among all the STS models. In section~\ref{sec:summary}, we discuss possible future problems. The definition of topological order at finite temperature is discussed further in~\ref{sec:topo_ap}. A possible relevance to topological quantum field theory is discussed in~\ref{sec:topology}. The proof of the dimensional duality of logical operators is presented in~\ref{sec:decomposition} and~\ref{sec:construction}. 

The main part of the paper is self-consistent and accessible to readers without previous knowledge on quantum coding theory and topological order. However, the proof part is rather technical, and relies heavily on theoretical tools developed previously in~\cite{Beni10b, Beni10} by the author. In particular, we owe a lot of arguments to~\cite{Beni10b} which introduced and solved the two-dimensional STS model originally. 

\section{Stabilizer code with physical constraints}\label{sec:model}

While local stabilizer codes are physically realizable as quantum memory devices in principle, realistic physical systems are often constrained to not only the locality of interaction terms, but also various physical symmetries. In this section, we give the definition of the STS model which are local stabilizer codes with translation and scale symmetries~\cite{Beni10b}. 

In section~\ref{sec:model1}, we give a brief review of stabilizer codes. In section~\ref{sec:model2}, we describe the definition of the STS model. 

\subsection{Stabilizer code}\label{sec:model1}

Here, we give a brief review of stabilizer codes which are quantum codes possessing Hamiltonians to support logical qubits in the ground space with a finite energy gap~\cite{Gottesman96}. Some notations which will be used throughout this paper are also fixed here. Note that we shall use the notations $\{\}$ for a \emph{set} and $\langle \rangle$ for a \emph{group}.

\textbf{Stabilizer formalism:} The main idea of stabilizer codes is to encode $k$ logical qubits into $N$ physical qubits ($N>k$) by using a subspace $V_{\mathcal{S}}$ spanned by states $|\psi\rangle$ that are invariant under the action of the \emph{stabilizer group} $\mathcal{S}$:
\begin{align}
V_{\mathcal{S}} \ = \ \Big\{ \ |\psi\rangle \ \in \ (\mathbb{C}^{2})^{\otimes N} \ : \ U|\psi\rangle \ = \ |\psi\rangle, \ \forall U \ \in \ \mathcal{S} \ \Big\}.
\end{align}
Here, the stabilizer group $\mathcal{S}$ is an arbitrary Abelian subgroup of the Pauli group 
\begin{align}
\mathcal{S} \ \subset \ \mathcal{P} \ = \ \Big\langle \ iI, X_{1},Z_{1},\dots , X_{N}, Z_{N} \ \Big\rangle
\end{align}
such that $-I \not\in \mathcal{S}$. The elements in $\mathcal{S}$ are called \emph{stabilizers}. The logical subspace $V_{\mathcal{S}}$ can be realized as the ground space of the following Hamiltonian (see Fig.~\ref{fig_stabilizer_summary})
\begin{align}
H \ = \ - \sum_{j} S_{j}, \qquad \mathcal{S} \ = \ \left\langle \ S_{1}, S_{2}, \cdots \ \right\rangle
\end{align}
since the energy eigenvalue is minimized for states satisfying $S_{j}|\psi\rangle = |\psi\rangle$ for all $j$. There are $k$ logical qubits encoded in $V_{\mathcal{S}}$ where $k \equiv N - G(\mathcal{S})$. Here, $G(\mathcal{S})$ represents the number of independent generators in $\mathcal{S}$. The ground space is separated from excited states by a finite energy gap since eigenstates are simultaneously diagonalized with respect to eigenvalues $\pm 1$ of $S_{j}$. 

\begin{figure}[htb!]
\centering
\includegraphics[width=0.35\linewidth]{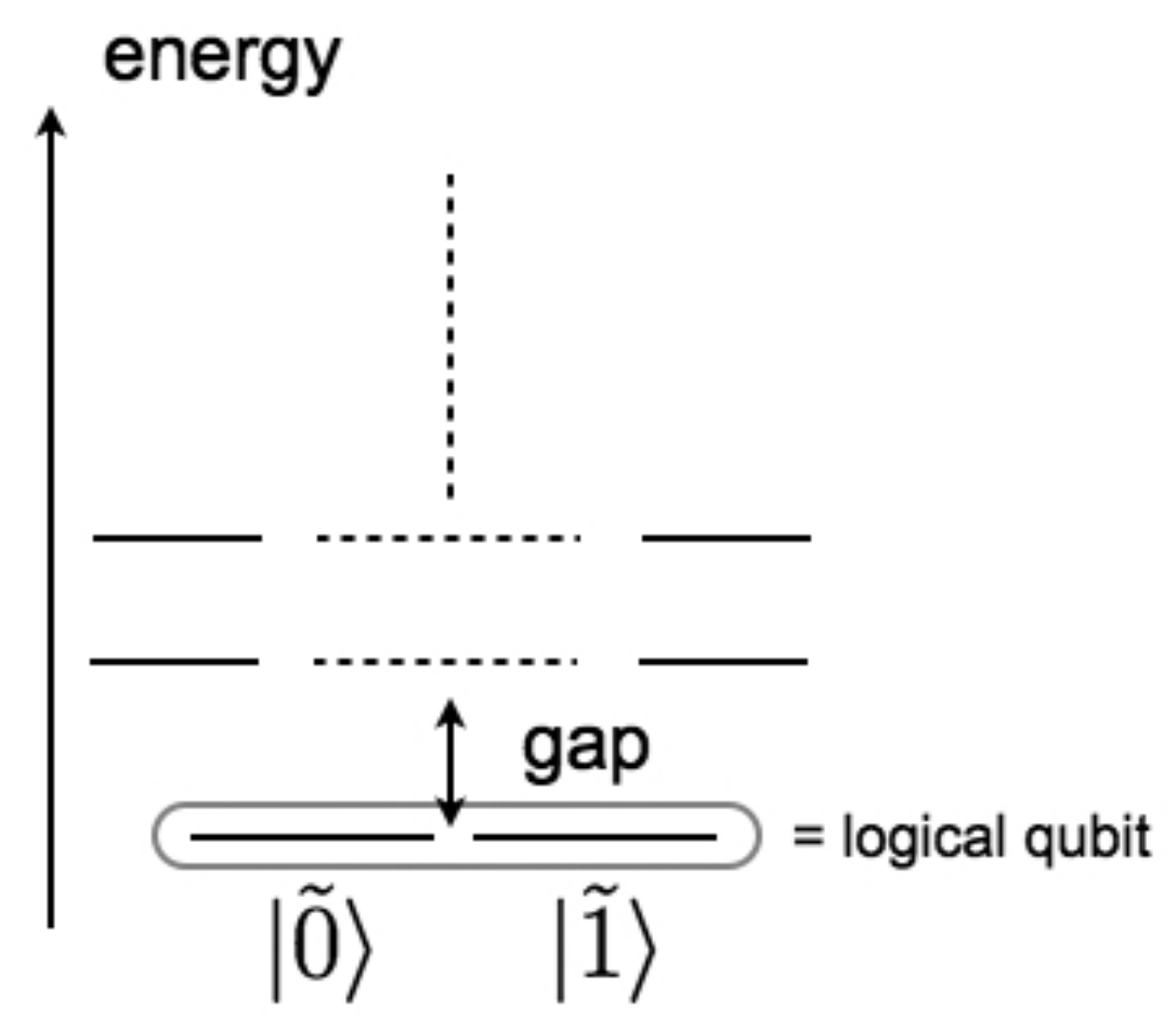}
\caption{Energy spectrum of a stabilizer Hamiltonian. Logical qubits are encoded in the degenerate ground space.
} 
\label{fig_stabilizer_summary}
\end{figure}

\textbf{Logical operators:}
In analyzing properties of logical qubits stored in the ground space, operators called \emph{logical operators} play central roles. Logical operators are Pauli operators which commute with the Hamiltonian, but not inside the stabilizer group $\mathcal{S}$. Logical operators can be found inside the centralizer group:
\begin{align}
\mathcal{C} \ = \ \Big\langle \ \Big\{ \ U \ \in \ \mathcal{P} \ : \ [U,S_{j}] \ = \ 0,\hspace{1ex} \mbox{for all}\ j \ \Big\} \ \Big\rangle
\end{align}
which is a group of Pauli operators commuting with all the stabilizers. Then, a set of logical operators is 
\begin{align}
\textbf{L} \ = \ \Big\{ \ U \ \in \ \mathcal{C} \ : \ U^{2} \ = \ I, \hspace{1ex} U \ \not \in \ \mathcal{S} \ \Big\}.
\end{align}
Logical operators may transform encoded qubits since they act non-trivially inside the ground space $V_{\mathcal{S}}$. 

\textbf{Equivalence relation:}
One may introduce an equivalence relation between logical operators by seeing how they act inside the ground space. Two logical operators $\ell$ and $\ell'$ are said to be \emph{equivalent} if and only if $\ell$ and $\ell'$ act in the same way inside the ground space:
\begin{align}
\ell \ \sim \ \ell' \  &\Leftrightarrow \ \ell|\psi\rangle \  = \  \ell' |\psi\rangle, \qquad \forall |\psi\rangle \  \in \  V_{\mathcal{S}}\\
&\Leftrightarrow \ \ell \ell' \ \in \ \mathcal{S}.
\end{align}
Therefore, logical operators remain equivalent under multiplications of stabilizers. 

\textbf{Canonical form:}
It is often convenient to represent a set of $2k$ independent logical operators in the following \emph{canonical} form~\cite{Beni10}:
\begin{align}
\left\{
\begin{array}{cccccc}
 \ell_{1}, & \cdots , &  \ell_{k}        \\
  r_{1}, & \cdots , &  r_{k} 
\end{array}
 \right\}.
\end{align}
Here, $\ell_{p}$ and $r_{p}$ are independent logical operators whose commutation relations are $\{ \ell_{p},r_{p}\}=0$, $[\ell_{p},r_{q}]=0$ for $p \not=q$, $[\ell_{p},\ell_{q}]=0$ and $[r_{p},r_{q}]=0$. Thus, only the operators in the same column anti-commute with each other. Note that choices of logical operators are not unique.

\textbf{Code distance:}
The code distance is a measure of the robustness of a quantum code, which is quantified by the minimal weight of logical operators:
\begin{align}
d \ = \ \min(w(U)) \qquad \mbox{where} \quad U \ \in \ \textbf{L}.
\end{align}
Here, $w(U)$ denotes the number of non-trivial Pauli operators constituting $U$. The code distance corresponds to a minimal number of single Pauli errors necessary to destroy an encoded qubit. Roughly speaking, a quantum code with a large code distance can securely protect logical qubits. 

\textbf{Bi-partition:} It is often convenient to split the entire system of qubits into two complementary subsets of qubits in studying coding properties of stabilizer codes~\cite{Fattal04, Wilde10}. Let us recall a useful formula to study geometric shapes of logical operators in stabilizer codes through a bi-partition. For a stabilizer code in a bi-partition, the following theorem is known to hold~\cite{Beni10} (Fig.~\ref{fig_stabilizer_summary_bipartition}). 

\begin{theorem}[Bi-partition]\label{theorem_partition}
For a stabilizer code with $k$ logical qubits, let the number of independent logical operators supported by a subset of qubits $R$ be $g_{R}$. Then, for an arbitrary bi-partition into two complementary subsets of qubits $R$ and $\bar{R}$, the numbers of logical operators supported by $R$ and $\bar{R}$ obey the following constraint:
\begin{align}
g_{R}+g_{\bar{R}} \ = \ 2k.
\end{align}
\end{theorem}

\begin{figure}[htb!]
\centering
\includegraphics[width=0.20\linewidth]{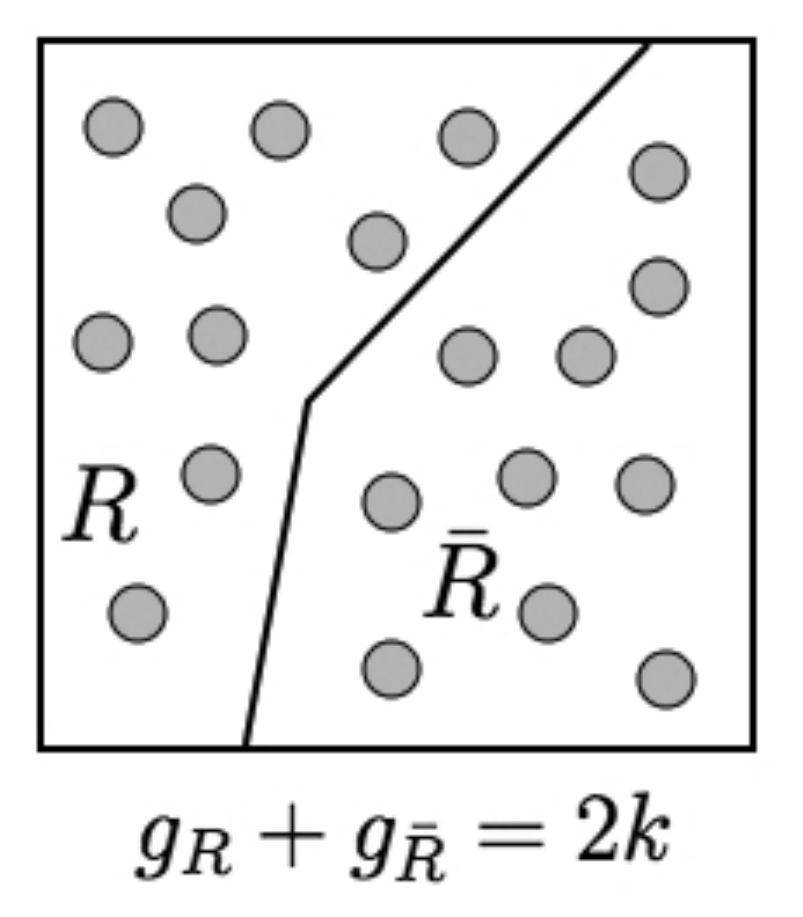}
\caption{A bi-partition of a stabilizer code. Each dot represents a qubit. 
} 
\label{fig_stabilizer_summary_bipartition}
\end{figure}

This bi-partition theorem is useful for analyzing geometric sizes and geometric shapes of logical operators. For example, if we find a region $R$ where there is no logical operator: $g_{R}=0$, we immediately know that all the logical operators can be supported inside $\bar{R}$ since $g_{\bar{R}}=2k$~\cite{Bravyi09}. Thus, one can restrict geometric regions of qubits where logical operators are supported. For an extension of the theorem to subsystem codes, see~\cite{Bravyi11, Haah10}.

\subsection{Stabilizer code with Translation and Scale symmetries}\label{sec:model2}

Here, we describe the definition of STS models, which are local stabilizer codes with translation and scale symmetries. 

\textbf{(1) Locality of interaction:}
Physically realistic systems must have only geometrically local interaction terms. To introduce the notion of locality to stabilizer codes, we consider a system of qubits defined on a $D$-dimensional square lattice (hypercubic lattice) which consists of  $N = L_{1} \times \cdots \times L_{D}$ qubits where $L_{m}$ is the total number of qubits in the $\hat{m}$ direction for $m = 1 ,\cdots ,D$. Therefore, qubits are distributed in the physical space with a \emph{metric}. 

Here, the entire system is separated into a collection of hypercubes which consists of $v = v_{1} \times \cdots \times v_{D}$ qubits without overlaps by assuming that $n_{m}\equiv L_{m}/ v_{m}$ are integer values (see Fig.~\ref{fig_STS}). We consider a block of $v = v_{1} \times \cdots \times v_{D}$ qubits as the single unit block which constitutes the entire system. In particular, we consider these unit blocks as single \emph{composite particles} with a larger Hilbert space $(\mathbb{C}^{2})^{\otimes v}$ (Fig.~\ref{fig_STS}). Thus, the entire system is viewed as a hypercubic lattice of $n_{1}\times \cdots \times n_{D}$ composite particles.

Now, we assume that interaction terms of the stabilizer Hamiltonian are defined \emph{locally}:
\begin{align}
H \ = \ - \sum_{j} S_{j} 
\end{align}
where $S_{j}$ are supported inside some regions with $2 \times \cdots \times 2$ composite particles (Fig.~\ref{fig_STS}). (Otherwise, we coarse-grain the system). In this paper, instead of qubits, we consider composite particles as the smallest building blocks of the system. 

\begin{figure}[htb!]
\centering
\includegraphics[width=0.70\linewidth]{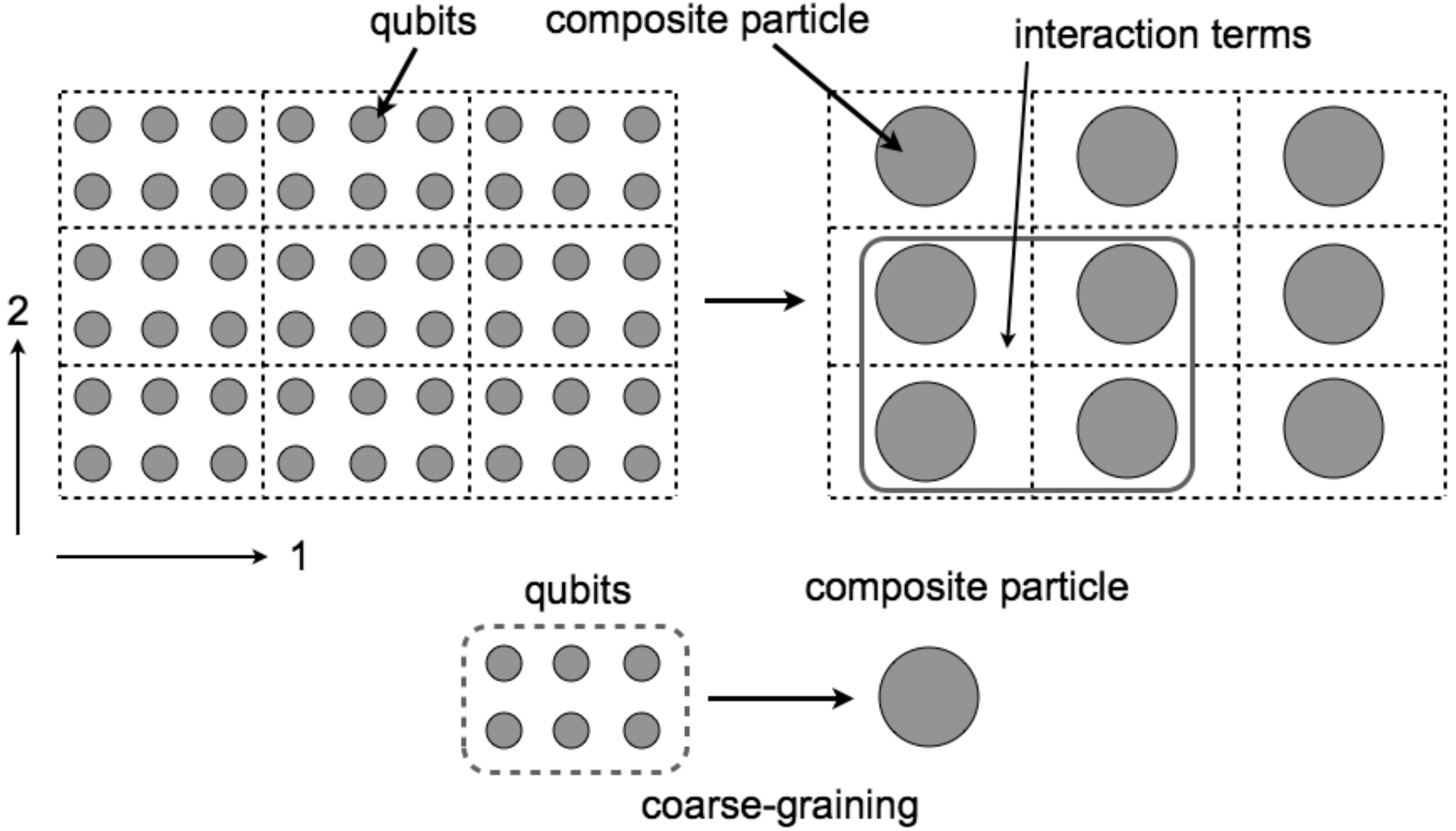}
\caption{An illustration of the STS model. A two-dimensional example is shown where a unit block of $3 \times 2$ qubits is considered as a composite particle with a larger Hilbert space. Interaction terms $S_{j}$ are defined locally inside a region of $2 \times 2$ composite particles. The Hamiltonian is invariant under unit translations of composite particles.
} 
\label{fig_STS}
\end{figure}

\textbf{(2) Translation symmetries:}
Physically realistic systems often have not only local interactions, but also some physical symmetries. Here, we assume that the stabilizer Hamiltonian possesses \emph{translation symmetries}:
\begin{align}
T_{m}(H)\ =\ H \qquad (m = 1, \cdots, D)
\end{align}
where $T_{m}$ represent \emph{unit translations of composite particles} in the $\hat{m}$ direction (Fig.~\ref{fig_STS}). 

For simplicity of discussion and in order to accommodate translation symmetries, we set the periodic boundary conditions. Then, the entire system may be viewed as a $D$-dimensional torus: $\textbf{T}^{D} = \textbf{S}^{1} \times \cdots \times \textbf{S}^{1}$
where $\textbf{S}^{1}$ is a circle. Thus, the entire system has a topologically non-trivial geometric shape \emph{a priori}.

\textbf{(3) Scale symmetries:}
In this paper, we are interested in coding properties at the limit where the system size goes to the infinity (in other words, at the \emph{thermodynamic limit}). So far, we have considered the cases where the system size $\vec{n} \equiv (n_{1}, \cdots, n_{D})$ is fixed. Here, we consider changes of the number of composite particles $n_{m}$ while keeping interaction terms $S_{j}$ the same. 

It is commonly believed that there is a tradeoff between the number of logical qubits $k$ and the code distance~\cite{Bravyi10, Haah10} where the code distance $d$ decreases as the number of logical qubits $k$ increases for a fixed $N$. Then, it may be legitimate to limit our considerations to the cases where the number of logical qubits $k$ remains small when the system size increases. 

We assume that stabilizer codes have \emph{scale symmetries} by requiring that the number of logical qubits $k_{\vec{n}}$ is independent of the system size $\vec{n}$:
\begin{align}
k_{\vec{n}} \ = \ k, \qquad \forall \vec{n}.
\end{align}
Here, we emphasize that, in a system with scale symmetries, the number of logical qubits $k$ remains constant under not only global scale transformations: $\vec{n} \rightarrow c \vec{n}$ where $c$ is some positive integer, but also arbitrary changes of $n_{m}$. 

One might think that scale symmetries are too strong as physical constraints. However, through appropriate coarse-graining, a fairly large class of local stabilizer codes with translation symmetries can be considered as the STS model. For example, let us consider the cases where the number of logical qubits $k_{\vec{n}}$ is small:
\begin{align}
k_{\vec{n}} \ \leq \ k_{0}, \qquad \forall \vec{n} \label{eq:small}
\end{align}
where $k_{\vec{n}}$ does not grow with the system size $\vec{n}$. Then, there always exists some finite coarse-graining such that a coarse-grained system has scale symmetries, as proven in~\cite{Beni10b}. As a result, by analyzing coding properties of stabilizer codes in the presence of scale symmetries, one can easily deduce coding properties of stabilizer codes which satisfy Eq.~(\ref{eq:small}). Therefore, solutions of STS models are sufficient to discuss coding properties of translation symmetric stabilizer codes with a small number of logical qubits. Note that STS models cover the Toric code and its generalizations to $D$-dimensional systems. 
 
\textbf{Translation equivalence of logical operators:}
There is a certain property of logical operators which emerges naturally as a result of translation and scale symmetries. For STS models, the following theorem holds~\cite{Beni10b}  (Fig.~\ref{fig_the_tools}).

\begin{theorem}[\textbf{Translation equivalence}]\label{theorem_TE}
For each and every logical operator $\ell$ in an STS model, a unit translation of $\ell$ with respect to composite particles in any direction is always equivalent to the original logical operator $\ell$:
\begin{align}
T_{m}(\ell) \ \sim \ \ell, \qquad \forall \ell \ \in \ \textbf{L}_{\vec{n}}  \qquad (m \ = \ 1, \cdots, D) 
\end{align}
where $\textbf{L}_{\vec{n}}$ is a set of all the logical operators for an STS model defined with the system size ${\vec{n}}$.
\end{theorem}

\begin{figure}[htb!]
\centering
\includegraphics[width=0.55\linewidth]{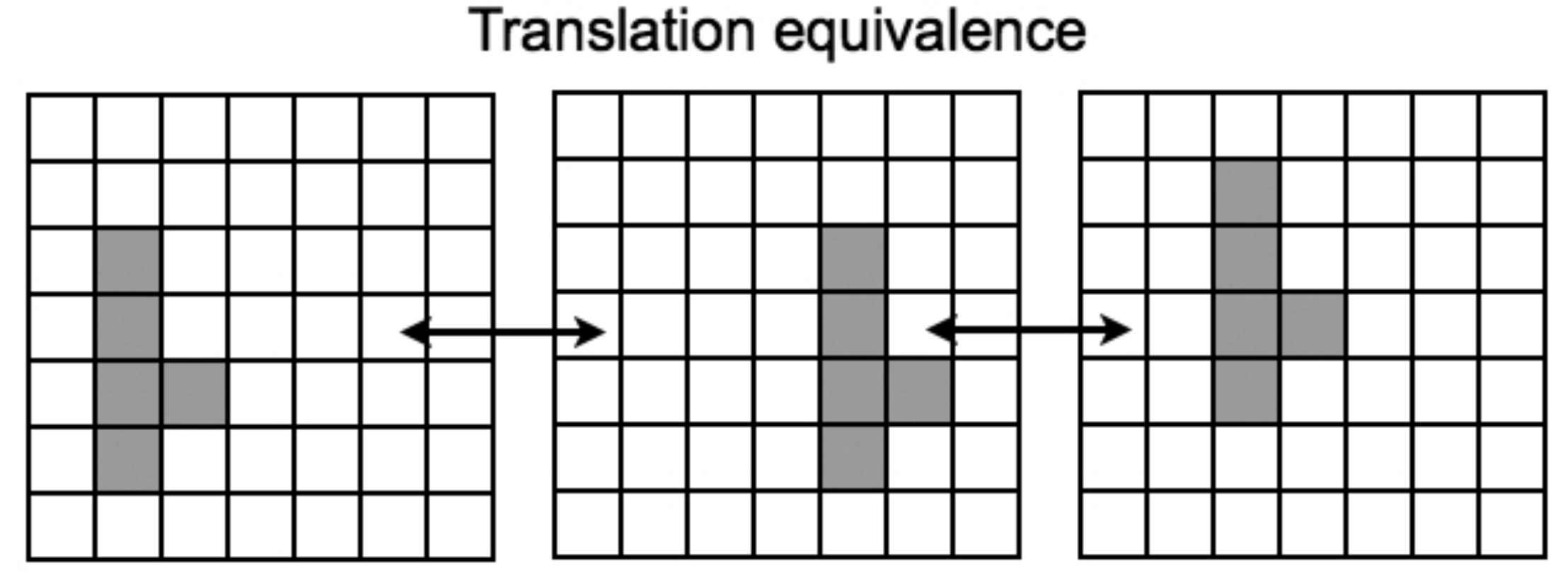}
\caption{The translation equivalence of logical operators. Each square represents a composite particle. Shaded regions represent translated logical operators which are equivalent to each other.
} 
\label{fig_the_tools}
\end{figure}

Here, we only give an intuition on why this theorem holds. Let us consider the case where the system size $\vec{n}$ is large. Then, since the number of logical qubits $k$ does not depend on the system size, $k$ is relatively small compared with the system size $\vec{n}$. Now, due to the translation symmetries of the system Hamiltonian, translations of a given logical operator $\ell$ are also logical operators. However, there are only $2k$ independent logical operators. Then, there must be a finite integer $a_{m}$ such that $\ell \sim T_{m}^{a_{m}}(\ell)$ for all the logical operators $\ell$. (Otherwise, there would be so many independent logical operators). It turns out that $a_{m}=1$ for any $\ell$ and $m$. While we have used only the condition that the number of logical operators $k$ is small, due to scale symmetries ($k$ is constant), one can prove the above theorem by showing $a_{m}=1$ for any $\ell$, $m$ and $\vec{n}$.

\section{Feasibility of self-correcting quantum memory}\label{sec:result}

In this section, we present one of our main results in this paper, concerning coding the feasibility of self-correcting quantum memory. In section~\ref{sec:result1}, we begin by giving brief discussion on the relation between the qubit storage time and the energy barrier. In section~\ref{sec:result2}, we describe geometric shapes of logical operators in STS models, and discuss whether a three-dimensional STS model works as self-correcting quantum memory or not. 

\subsection{Self-correction, energy barrier and storage time}\label{sec:result1}

In subsection, we review how self-correcting quantum memory works by establishing the connection between the energy barrier and the qubit storage time.

\textbf{Self-correcting classical memory:} For simplicity of discussion, let us start by analyzing an example of self-correcting \emph{classical} memory. Consider two-dimensional Ising model:
\begin{align}
H\ =\ - \sum_{i,j}Z_{i,j}Z_{i+1,j} - \sum_{i,j}Z_{i,j}Z_{i,j+1} 
\end{align}
which consists of $L \times L$ qubits with periodic boundary conditions. The model works as a classical code since one can encode a classical bit in the ground space by labeling $|0\cdots0\rangle$ as $0$ and $|1\cdots1\rangle$ as $1$. 

Now, let us see why this model works as self-correcting classical memory. Assume that the system is originally $|0\cdots0\rangle$. Then, in order for errors to change a ground state $|0\cdots0\rangle$ into another ground state $|1\cdots1\rangle$, errors must flip all the spins from $|0\rangle$ to $|1\rangle$. However, during these spin flips, the excitation energy becomes at least $O(L)$ because there is a domain wall separating the regions with $|0\rangle$s and $|1\rangle$s (Fig.~\ref{fig_self}). In other words, ground states $|0\cdots0\rangle$ and $|1\cdots1\rangle$ are separated by a \emph{large energy barrier}. Then, before errors accumulate, natural thermal dissipation processes restore the system into the original encoded state\footnote{Precisely speaking, the system does not return to the original state, but return to a state which is sufficiently close to the original state with a probability approaching to unity at the thermodynamic limit. Therefore, one can reliably read out the encoded bit from such a state.}. Therefore, the system corrects errors by itself.

One may estimate the bit storage time of two-dimensional Ising model by using the so-called Arrhenius law:
\begin{align}
\tau \ \sim \ \exp(\Delta E /T) 
\end{align}
up to some polynomial corrections where $\tau$ is the storage time, $\Delta E$ is an energy barrier and $T$ is the temperature. According to the Arrhenius law, the bit storage time can be estimated as $\tau \sim \mbox{EXP}(L)$ since the energy barrier is $\Delta E \sim O(L)$. While it should be noted that this is an empirical law, it is commonly believed that the law correctly estimates the bit or qubit storage time of spin systems (for systems with well defined energy barrier) below a critical temperature. Indeed, the law is rigorously verified in various models of classical and quantum memory~\cite{Brey96, Nussinov08, Alicki09, Alicki10, Chesi10}. However, it should be noted that verification of the Arrhenius law for classical and quantum memory is usually a very difficult task, involving an evaluation of time evolution of the system in a presence of interactions with an external environment. 

Next, let us consider an example of a classical code which is not self-correcting; one-dimensional Ising model:
\begin{align}
H\ =\ - \sum_{j}Z_{j}Z_{j+1} .
\end{align}
One may see that this model does not work as self-correcting classical memory since one can change from $|0\cdots0\rangle$ to $|1\cdots1\rangle$ by costing only a finite energy. In other words, one can create a kink by costing only a finite energy:
\begin{align}
|0 0 0 0 \cdots 0 \rangle \ \rightarrow \ |0 1 1 1 \cdots 1 \rangle \ \rightarrow \ |0 0 1 1 \cdots 1 \rangle \ \rightarrow \ \cdots \rightarrow \ |1 1 1 1 \cdots 1 \rangle
\end{align}
which may propagate the lattice freely without costing any extra energy. Then, according to the Arrhenius law, the bit storage time is $\tau \sim O(1)$ which is independent of the system size since the energy barrier is $\Delta E \sim O(1)$.

\begin{figure}[htb!]
\centering
\includegraphics[width=0.60\linewidth]{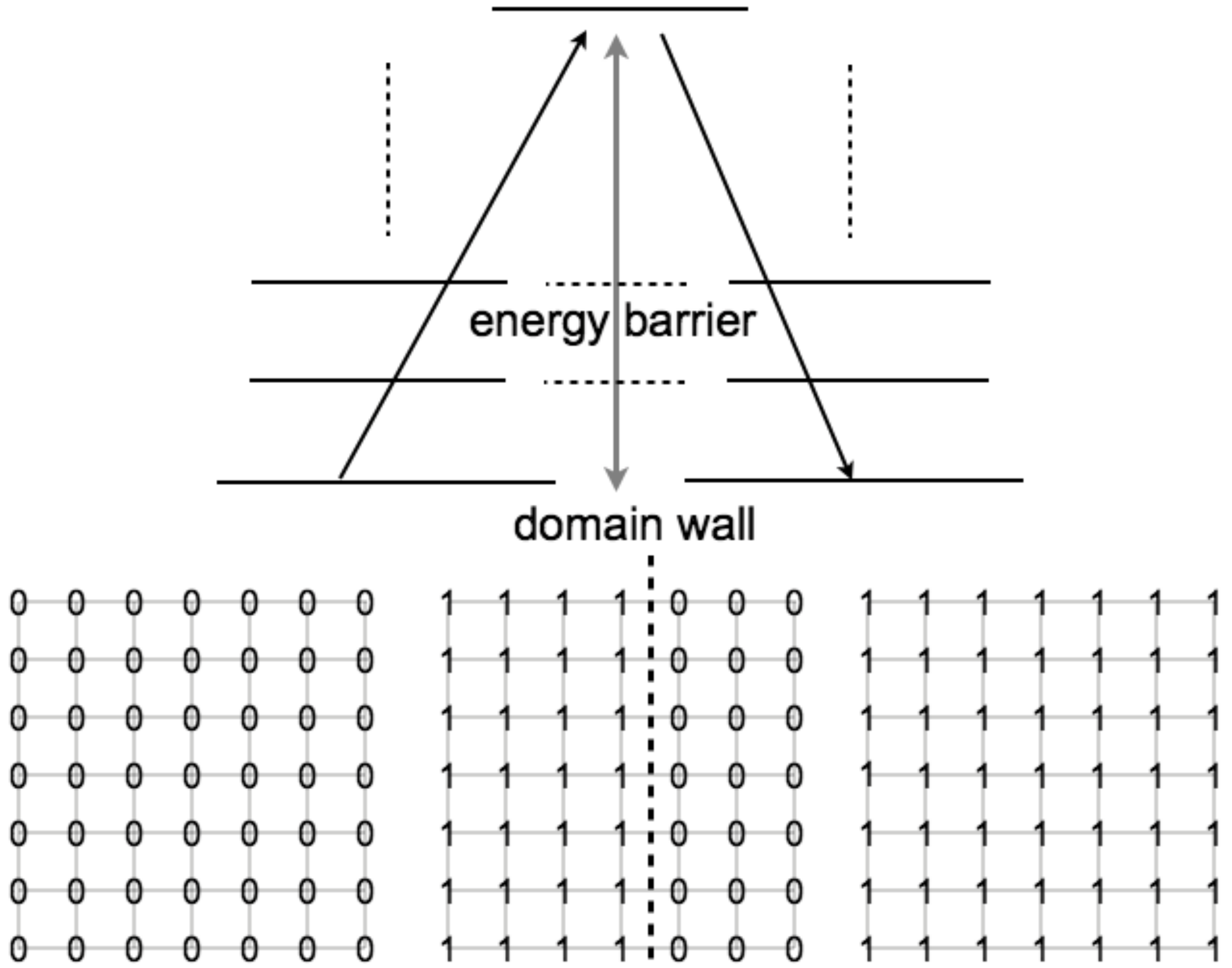}
\caption{How self-correcting \emph{classical} memory works in two-dimensional Ising model.
} 
\label{fig_self}
\end{figure}

\textbf{Energy barrier and logical operators:}
One can associate the self-correcting property of two-dimensional Ising model with geometric shapes of logical operators by viewing the model as a stabilizer code. Note that $D$-dimensional Ising model satisfies the definition of STS models since interaction terms are local and translation symmetric, and there is a single logical qubit ($k=1$ and $d=1$) regardless of the system size. Two-dimensional Ising model has the following pair of zero-dimensional and two-dimensional logical operators:
\begin{align}
\ell\ = \ Z_{1,1},\qquad r\ =\ \prod_{i,j}X_{i,j}.
\end{align}
Then, a classical bit is encoded in eigenstates of a zero-dimensional logical operator $\ell$. In order to change the encoded bit, one needs to apply a two-dimensional logical operator $r$ since $|1\cdots1\rangle = r|0\cdots0\rangle$. Then, an intermediate state during the change from $|0\cdots0\rangle$ to $|1\cdots1\rangle$ may be represented as $r^{*}|0\cdots0\rangle$ where $r^{*}$ is some ``subpart'' of the original two-dimensional logical operator $r$ (see Fig~\ref{fig_self2}(a)). Since interaction terms anti-commute with Pauli operators at the boundary of $r^{*}$, the excitation energy associated with $r^{*}|0\cdots0\rangle$ is proportional to the perimeter of $r^{*}$. Thus, during the change from $|0\cdots0\rangle$ to $|1\cdots1\rangle$, the excitation energy must become $O(L)$ since the perimeters of subparts of a two-dimensional logical operator $r$ are always one-dimensional. 

One can also understand why one-dimensional Ising model does not work as self-correcting classical memory through geometric shapes of logical operators. One-dimensional Ising model has the following pair of zero-dimensional and one-dimensional logical operators: $\ell = Z_{1}$ and $r = \prod_{i}X_{i}$. Since a subpart $r^{*}$ of a one-dimensional logical operator $r$ is always zero-dimensional, the energy barrier is $\Delta E \sim O(1)$ (see Fig~\ref{fig_self2}(b)). 

\begin{figure}[htb!]
\centering
\includegraphics[width=0.55\linewidth]{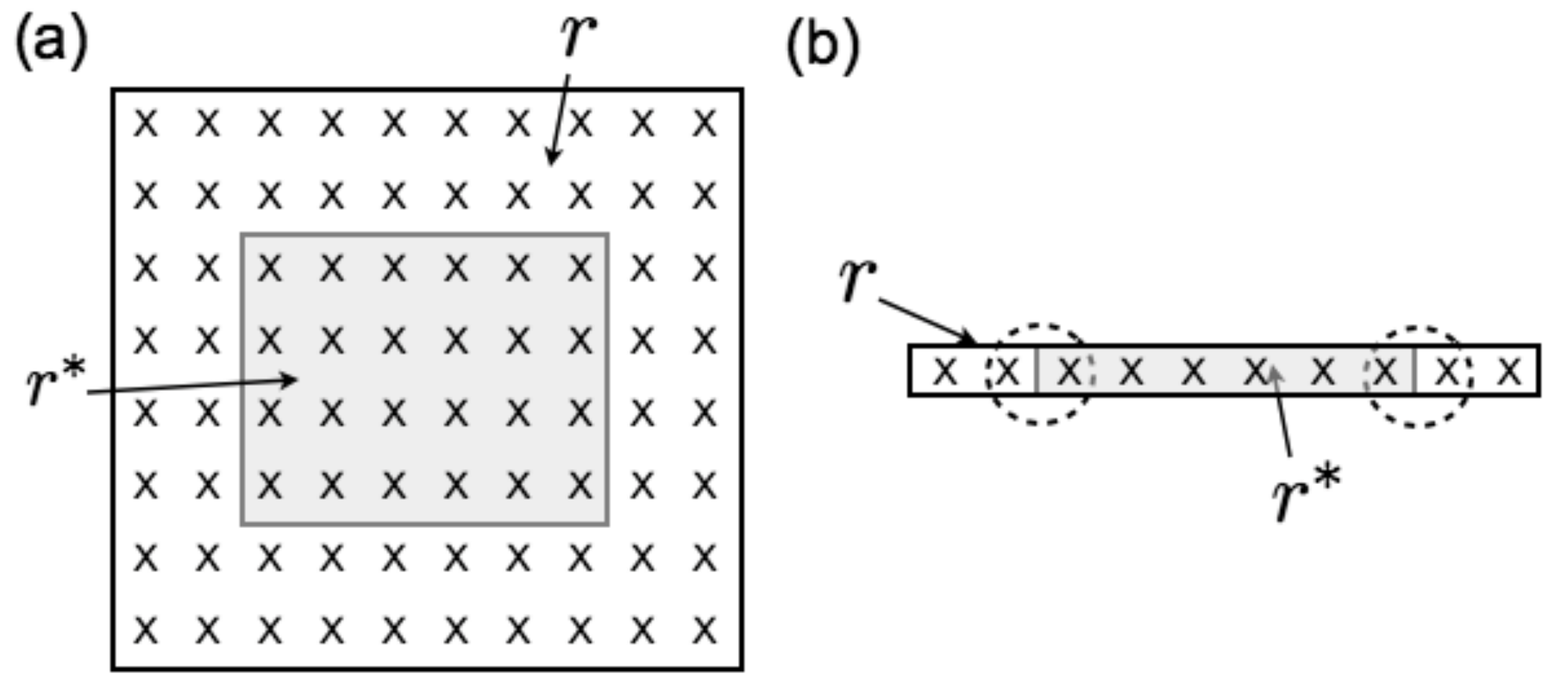}
\caption{(a) A subpart of a two-dimensional logical operator. The excitation energy is $O(L)$. (b) A subpart of a one-dimensional logical operator. The excitation energy is $O(1)$.
} 
\label{fig_self2}
\end{figure}

\textbf{Self-correcting quantum memory:}
Next, let us discuss how self-correcting quantum memory works. It is known that many of good local stabilizer codes with macroscopic code distances do not have self-correcting properties since they have string-like logical operators which lead to $O(1)$ energy barrier. As an example, let us consider two-dimensional Toric code:
\begin{align}
H_{Toric} \ = \ - \sum_{s} A_{s} - \sum_{p} B_{p} 
\end{align}
where qubits live on edges of $L \times L$ square lattice. Interaction terms $A_{s}$ and $B_{p}$ are shown in Fig.~\ref{fig_Toric_code}. Two-dimensional Toric code can be viewed as an STS model since the number of logical qubits is constant: $k=2$ regardless of the system size. 

\begin{figure}[htb!]
\centering
\includegraphics[width=0.55\linewidth]{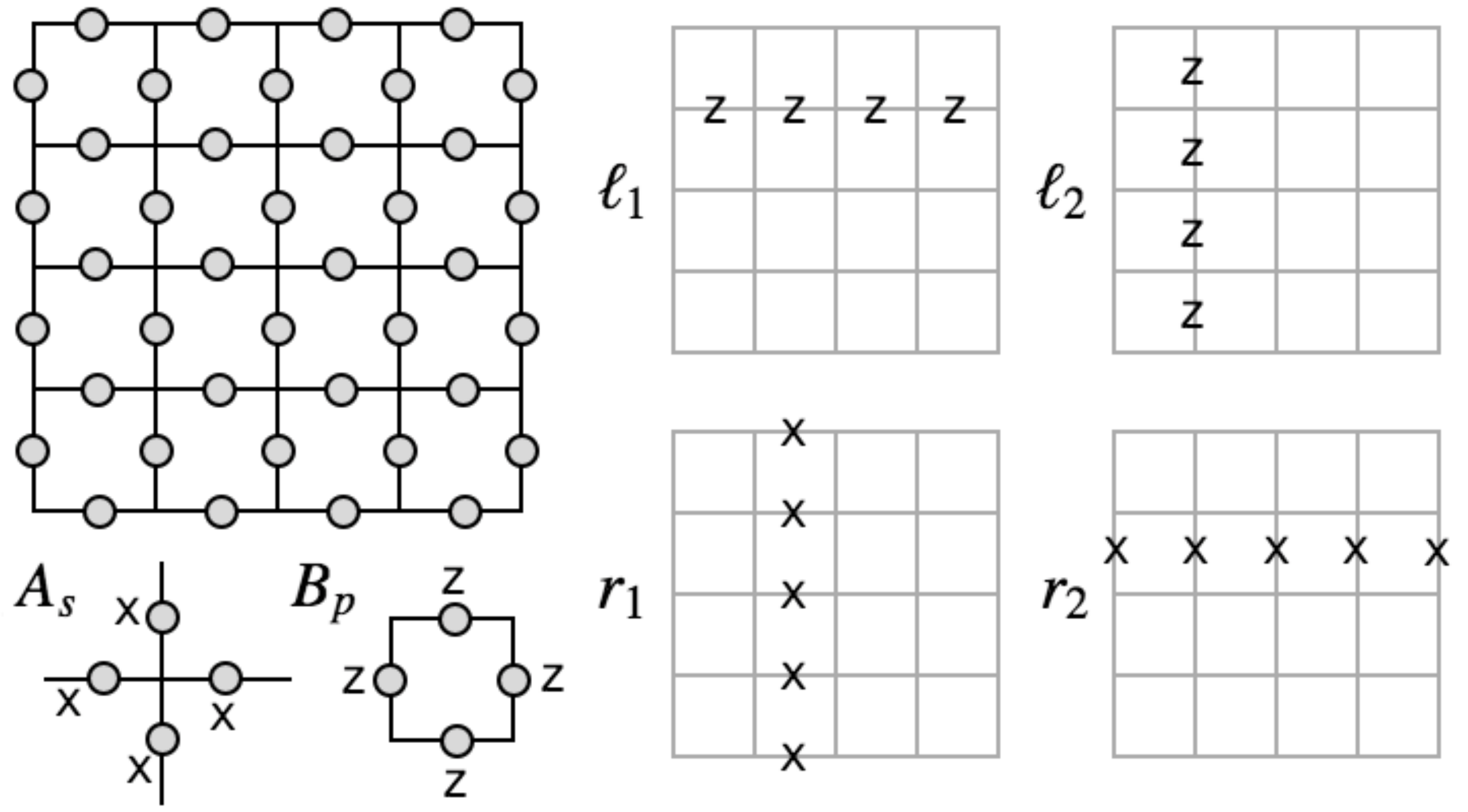}
\caption{(a) A subpart of a two-dimensional logical operator. The excitation energy is $O(L)$. (b) A subpart of a one-dimensional logical operator. The excitation energy is $O(1)$.
} 
\label{fig_Toric_code}
\end{figure}

One can understand why two-dimensional Toric code does not work as self-correcting quantum memory through geometric shapes of logical operators. Since the Toric code has pairs of one-dimensional logical operators as shown in Fig.~\ref{fig_Toric_code}, the energy barrier is $\Delta E \sim O(1)$. As a result, the qubit storage time is $\tau \sim O(1)$ according to the Arrhenius law. A similar discussion holds for three-dimensional Toric code too. 

As two-dimensional and three-dimensional Toric codes have $O(1)$ qubit storage time, they are not reliable quantum memories. Then, what kinds of quantum codes may serve as reliable quantum memory devices? It is worth noting that two-dimensional Toric code works as a reliable quantum memory device if one perform sufficiently frequent and accurate error-corrections. In particular, it has been shown~\cite{Dennis02} that the qubit storage time can be exponentially long: $\tau\sim \mbox{EXP($L$)}$ in the presence of active error-corrections; however, it seems very difficult to reach such accuracy and frequency which are necessary for reliably storing qubits. 

In the same paper~\cite{Dennis02}, it has been also pointed out that four-dimensional Toric code may have exponentially long storage time $\tau \sim \mbox{EXP($L$)}$ below the critical temperature since the model has a large energy barrier which scales as $O(L)$. This remarkable insight has been further investigated in~\cite{Takeda04}, and later, rigorously verified in~\cite{Alicki10}. One can associate the self-correcting property of four-dimensional Toric code with geometric shapes of logical operators, as the model has pairs of two-dimensional logical operators which lead to $O(L)$ energy barrier. Yet, since one cannot embed four-dimensional Toric code in a three-dimensional space, one hopes to have three-dimensional self-correcting quantum memory whose qubit storage time grows as the system size increases. 

\textbf{Previous proposals:}
So far, there have been several interesting proposals of three-dimensional self-correcting quantum memory. First, a certain model of three-dimensional subsystem codes~\cite{Bacon06}, called three-dimensional quantum compass model, was proposed as a candidate of self-correcting quantum memory. Subsystem codes may be considered as generalizations of stabilizer codes which are supported by Hamiltonians whose interaction terms are Pauli operators, but may anti-commute with each other. While the model opened a new possibility of quantum memory devices supported by frustrated Hamiltonians and initiated studies of Hamiltonian realizations of subsystem codes, its validity as self-correcting quantum memory has not been verified since it is difficult to solve the frustrated Hamiltonian, and its properties are hard to determine. Also, it seems difficult for such frustrated systems to have the thermal stability since the model supports gapless excitations. Furthermore, its validity has been denied in several literature including~\cite{Nussinov08}. 

Later, an interesting model of the mixture of two-dimensional Toric code and three-dimensional Bosonic gas, called the Toric-Boson model, was proposed~\cite{Hamma09} where Bosonic gas induces confining potential between anyonic excitations. While the model opened new capacities of quantum codes constructed in the so called mixed-dimensional configurations, which are of particular interest in ultracold atom physics community, the model itself has two serious drawbacks as a candidate of self-correcting quantum memory. First, its storage time is only polynomial in $L$: $\tau \sim \mbox{POLY}(L)$ since an effective energy barrier is: $\Delta E \sim \mbox{LOG($L$)}$. However, it is known that it takes at least $O(d)$ gate operations, which cannot be implemented simultaneously, to read out the encoded qubit, where $d$ is the code distance of the system. Then, since the code distance $d$ scales polynomially in $L$ in the model, polynomially long storage time is not sufficient for a model to work as efficient quantum memory device~\cite{Alicki09}. Second, the model needs a very strong coupling between the Toric code and Bosonic gas whose strength scales polynomially in $L$. Therefore, the model has a problem in the scalability. 

Recently, an interesting proposal of three-dimensional spin glass models~\cite{Haah11} has been made whose relaxation dynamics is very slow due to the existence of a large number of energy local minima. The model seems to have a polynomially long storage time: $\tau \sim \mbox{POLY($L$)}$ with logarithmically large energy barrier: $\Delta E \sim \mbox{LOG($L$)}$, though it needs verifications (so, it could be longer or shorter). The system undergoes a phase transition at $T=0$, which mat imply a potential thermal instability of the system properties. Finally, it seems difficult o find an efficient decoding scheme for the model. Here, it should be noted that the model does not have scale symmetries since the number of logical qubits is highly sensitive to the system size, and there is no finite upper bound on it. See section~\ref{sec:summary} for discussion on stabilizer codes without scale symmetries too.

At this moment, validities of none of these proposals of three-dimensional self-correcting quantum memory has been verified yet, and the feasibility of self-correcting quantum memory remains open. Below, we summarize the qubit storage time and the energy barrier in these proposals of self-correcting quantum memory.

\begin{center}
\begin{tabular}{cccc}
\mbox{Code} & \mbox{Energy Barrier} & \mbox{Storage Time} & \mbox{Comment} \\ \hline
\mbox{2-dim Toric Code} & \mbox{O(1)} & \mbox{O(1)} &      \\
\mbox{3-dim Toric Code} & \mbox{O(1)} & \mbox{O(1)} &      \\
\mbox{2-dim Toric Code} & \mbox{O(1)} & \mbox{EXP(L)} & \mbox{with error-correction}\\
\mbox{4-dim Toric Code} & \mbox{O(L)} & \mbox{EXP(L)} & \mbox{self-correcting} \\
\mbox{3-dim subsystem code} & \mbox{?} & \mbox{?} & \mbox{gapless excitation?} \\
\mbox{Toric-Boson model} & \mbox{LOG(L)} & \mbox{POLY(L)} & \mbox{POLY(L) coupling} \\
\mbox{Spin glass model} & \mbox{LOG(L)} & \mbox{POLY(L)} & \mbox{$T=0$ phase transition} 
\end{tabular}
\end{center}

\subsection{Dimensional duality of logical operators}\label{sec:result2}

In this subsection, we describe geometric shapes of logical operators in STS models, and discuss the feasibility of self-correcting quantum memory.

\textbf{One-dimension:} We begin by reviewing geometric shapes of logical operators in one-dimensional STS models. In a one-dimensional chain of $n_{1}$ composite particles, we denote $j$th composite particle as $P_{j}$. Then, the following theorem holds~\cite{Beni10b}.

\begin{theorem}[\textbf{Dimensional Duality}]\label{theorem_1dim}
For a one-dimensional STS model, there exists a canonical set of logical operators:
\begin{align}
\left\{
\begin{array}{cccccc}
 \ell_{1}, & \cdots , &  \ell_{k}      \\
  r_{1}, & \cdots , &   r_{k}
\end{array}
 \right\}
\end{align}
whose pair of anti-commuting operators $\ell_{j}$ and $r_{j}$ has one of the following property ($j=1,\cdots,k$).
\begin{itemize}
\item $\ell_{j}$ is a zero-dimensional logical operator defined inside $P_{1}$, while $r_{j}$ is a one-dimensional logical operator defined in a periodic way: $T_{1}(r_{j})=r_{j}$.
\end{itemize}
\end{theorem}

It is worth presenting geometric shapes of logical operators graphically (Fig.~\ref{fig_1D_logical}). The code distance of one-dimensional models is upper bounded by $v$ which is the inner dimension of a composite particle. Since there always exist zero-dimensional logical operators, such systems cannot have topological order.

\begin{figure}[htb!]
\centering
\includegraphics[width=0.30\linewidth]{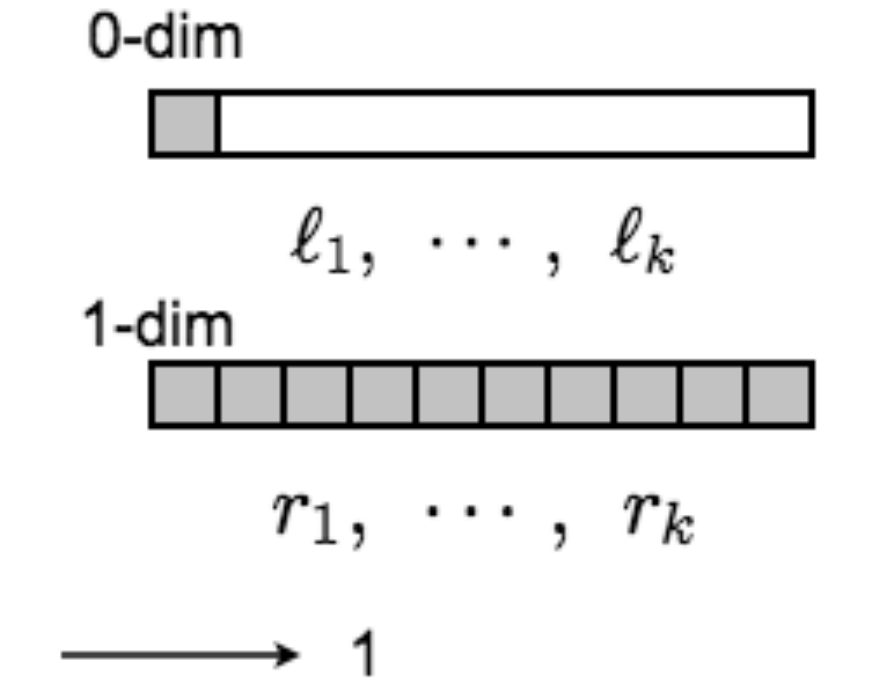}
\caption{Logical operators in a one-dimensional system.} 
\label{fig_1D_logical}
\end{figure}

\textbf{Two-dimension:} Next, we analyze geometric shapes of logical operators for two-dimensional STS models. Let us first introduce some regions of composite particles in order to define geometric shapes of logical operators concisely. A square region of $x_{1}\times x_{2}$ composite particles is denoted as $P(x_{1},x_{2})$ (Fig.~\ref{fig_2D_logical}(a)):
\begin{align}
P(x_{1},x_{2}) \ \equiv \ \Big\{ \ P_{r_{1},r_{2}} \ : \ 1 \ \leq \ r_{1} \ \leq \ x_{1}, \ 1 \ \leq \ r_{2} \ \leq \ x_{2} \ \Big\}
\end{align}
where $1 \leq x_{1} \leq n_{1}$ and $1 \leq x_{2} \leq n_{2}$. Note that a composite particle at the position $(r_{1},r_{2})$ is denoted as $P_{r_{1},r_{2}}$. Then, the following theorem holds~\cite{Beni10b}.

\begin{theorem}[\textbf{Dimensional Duality}]\label{theorem_2dim}
For a two-dimensional STS model, there exists a canonical set of logical operators:
\begin{align}
\left\{
\begin{array}{cccccc}
 \ell_{1}, & \cdots , &  \ell_{k}      \\
  r_{1}, & \cdots , &   r_{k}
\end{array}
 \right\}
\end{align}
whose pair of anti-commuting operators $\ell_{j}$ and $r_{j}$ has one of the following two properties ($j=1,\cdots,k$).
\begin{itemize}
\item $\ell_{j}$ is a zero-dimensional logical operator defined inside $P(1,2v)$, while $r_{j}$ is a two-dimensional logical operator defined in a periodic way: $T_{1}(r_{j})=r_{j}$ and $T_{2}(r_{j})=r_{j}$.
\item $\ell_{j}$ is a one-dimensional logical operator defined inside $P(1,n_{2})$ in a periodic way: $T_{2}(\ell_{j})=\ell_{j}$, while $r_{j}$ is a one-dimensional logical operator defined inside $P(n_{1},1)$ in a periodic way: $T_{1}(r_{j})=r_{j}$.
\end{itemize}
\end{theorem}

It is worth presenting geometric shapes of logical operators graphically (Fig.~\ref{fig_2D_logical}(b)). There is a dimensional duality on geometric shapes of logical operators as follows:
\begin{align}
\left\{
\begin{array}{cc}
\mbox{0 dim}, & \mbox{1 dim} \\
\mbox{2 dim}, &  \mbox{1 dim}
\end{array}
\right\}.
\end{align}
Logical operators are periodic in the directions in which they are stretched. The code distance is upper bounded by $O(L)$ which is consistent with the upper bound $O(L^{D-1})$ obtained in~\cite{Bravyi09}. The existence of one-dimensional logical operators gives rise to topological order.

\begin{figure}[htb!]
\centering
\includegraphics[width=0.65\linewidth]{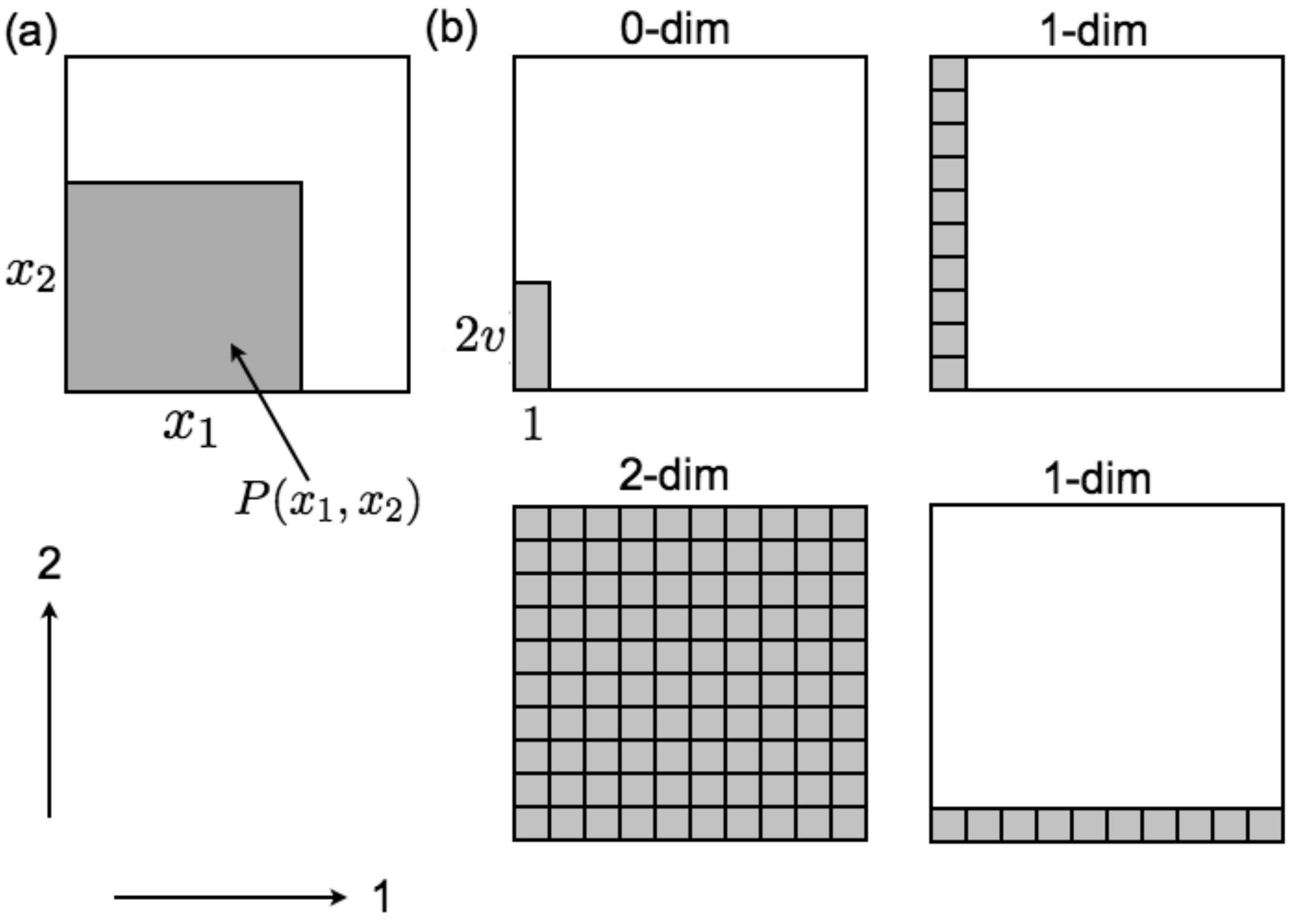}
\caption{Dimensional duality in a two-dimensional system. (a) A region of $x_{1}\times x_{2}$ composite particles is denoted as $P(x_{1},x_{2})$. (b) Geometric shapes of logical operators in a two-dimensional STS model.
} 
\label{fig_2D_logical}
\end{figure}

\textbf{Three-dimension:} Finally, let us proceed to coding properties of a three-dimensional STS model. A region with $x_{1}\times x_{2} \times x_{3}$ composite particles is denoted as $P(x_{1},x_{2},x_{3})$:
\begin{align}
P(x_{1},x_{2},x_{3}) \ \equiv \ \Big\{ \ P_{r_{1},r_{2},r_{3}} \ : \ 1 \ \leq \ r_{m} \ \leq \ x_{m},\ \ m \ = \ 1,2,3 \ \Big\}
\end{align}
where $P_{r_{1},r_{2},r_{3}}$ represents a composite particle at $(r_{1},r_{2},r_{3})$. Then, for logical operators in a three-dimensional STS model, the following theorem holds.

\begin{theorem}[\textbf{Dimensional Duality}]\label{theorem_3dim}
For a three-dimensional STS model, there exists a canonical set of logical operators:
\begin{align}
\left\{
\begin{array}{cccccc}
 \ell_{1}, & \cdots , &  \ell_{k}      \\
  r_{1}, & \cdots , &  r_{k}
\end{array}
 \right\}
\end{align}
whose pair of anti-commuting operators $\ell_{j}$ and $r_{j}$ has one of the following four properties.
\begin{itemize}
\item $\ell_{j}$ is a zero-dimensional logical operator defined inside $P(1,2v,(2v)^{2})$, while $r_{j}$ is a three-dimensional logical operator defined in a periodic way: $T_{1}(r_{j})=r_{j}$, $T_{2}(r_{j})=r_{j}$ and $T_{3}(r_{j})=r_{j}$.
\item $\ell_{j}$ is a one-dimensional logical operator defined inside $P(n_{1},2v,1)$ in a periodic way: $T_{1}(\ell_{j})=\ell_{j}$, while $r_{j}$ is a two-dimensional logical operator defined inside $P(1,n_{2},n_{3})$ in a periodic way: $T_{2}(r_{j})=r_{j}$ and $T_{3}(r_{j})=r_{j}$.
\item $\ell_{j}$ is a one-dimensional logical operator defined inside $P(1,n_{2},2v)$ in a periodic way: $T_{2}(\ell_{j})=\ell_{j}$, while $r_{j}$ is a two-dimensional logical operator defined inside $P(n_{1},1,n_{3})$ in a periodic way: $T_{1}(r_{j})=r_{j}$ and $T_{3}(r_{j})=r_{j}$.
\item $\ell_{j}$ is a one-dimensional logical operator defined inside $P(2v,1,n_{3})$ in a periodic way: $T_{3}(\ell_{j})=\ell_{j}$, while $r_{j}$ is a two-dimensional logical operator defined inside $P(n_{1},n_{2},1)$ in a periodic way: $T_{1}(r_{j})=r_{j}$ and $T_{2}(r_{j})=r_{j}$.
\end{itemize}
\end{theorem}

This is the main technical result of the present paper. We present the proof of the theorem in appendices~\ref{sec:decomposition} and~\ref{sec:construction}. It is worth representing geometric shapes of logical operators graphically (Fig.~\ref{fig_3D_logical}). Note that logical operators are periodic in the directions in which they are stretched. There is a dimensional duality on geometric shapes of logical operators as follows:
\begin{align}\left\{
\begin{array}{cc}
\mbox{0 dim}, &  \mbox{1 dim} \\
\mbox{3 dim}, &  \mbox{2 dim}
\end{array}\right\}
\end{align}

\begin{figure}[htb!]
\centering
\includegraphics[width=0.90\linewidth]{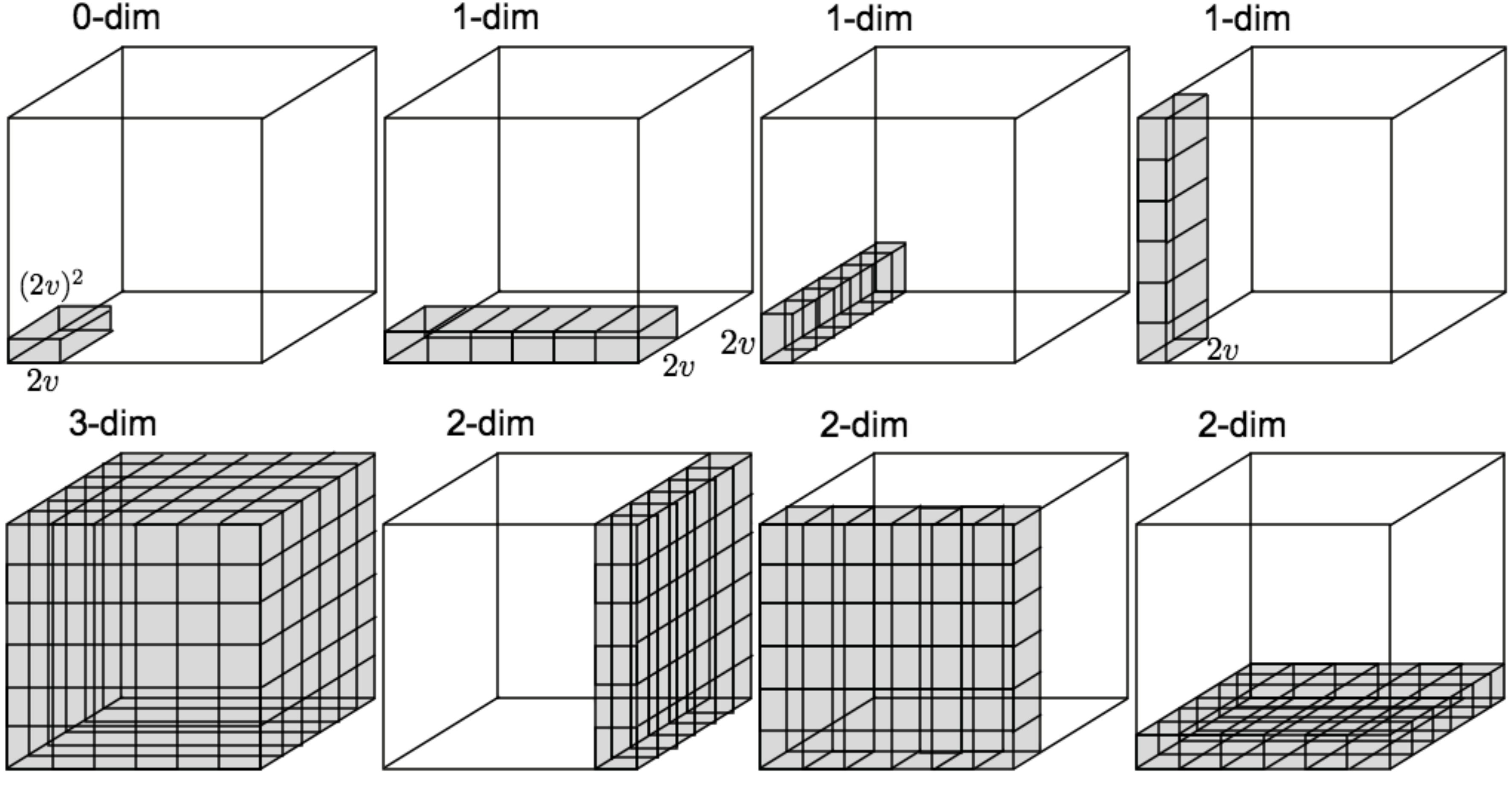}
\caption{Dimensional duality in a three-dimensional system. 
} 
\label{fig_3D_logical}
\end{figure}

As a result of this dimensional duality, one may find the upper bound on the code distance. When $n_{1}=n_{2}=n_{3}=L$, the code distance of a three-dimensional STS model is upper bounded as follows:
\begin{align}
d \ \leq \ 2vL \ \sim \ O(L).
\end{align}
Note that this bound is tight for the three-dimensional Toric code. The system may have topological order when it is free from zero-dimensional logical operators.

\textbf{Higher-dimensions:}
Though our primary interests are in coding properties of three-dimensional STS models, it is worth extending the analysis to higher dimensions. For a $D$-dimensional STS model ($D \geq 4$), we make the following \emph{conjectures}:

\begin{itemize}
\item In a $D$-dimensional system, $m$-dimensional and $(D-m)$-dimensional logical operators form anti-commuting pairs where $m$ is an integer.
\item The code distance is tightly upper bounded by $O(L^{\frac{D}{2}})$ when $D$ is even and by $O(L^{\frac{D-1}{2}})$ when $D$ is odd.
\end{itemize}

Note that generalizations of the Toric code to $D$-dimensional systems have the above dimensional duality for arbitrary integer $m$. See~\ref{sec:topology2} for constructions of $D$-dimensional Toric code. Such a dimensional duality arising in geometric shapes of logical operators may be viewed as a manifestation of Poincar\'{e} duality on coding properties of stabilizer codes supported on a $D$-torus, as very briefly discussed in~\ref{sec:topology}. 

However, we report that straightforward generalizations of analysis tools presented in this paper did not work for STS models with $D \geq 4$, and the conjecture above needs to be verified.

\textbf{Feasibility for a three-dimensional STS model:}
Finally, let us show that three-dimensional STS models do not work as self-correcting quantum memory. First, in order for the system to work as a quantum code, there must be an anti-commuting pair of one-dimensional and two-dimensional logical operators. We denote them as $\ell$ and $r$ respectively. Let us assume that the system is initially in the eigenstate of $\ell=1$ denoted as $|\psi (\ell=1)\rangle$. Then, $|\psi (\ell=-1)\rangle = r |\psi (\ell=1)\rangle$ and $r$ is a two-dimensional logical operator. Since the boundary of a subpart of $r$ is one-dimensional, the excitation energies associated with intermediate states are $O(L)$. Thus, the encoding with respect to $\ell$ is self-correcting. 

Next, let us assume that the system is initially in the eigenstate of $r=1$ denoted as $|\psi (r=1)\rangle$. Then, $|\psi (r=-1)\rangle = \ell |\psi (r=1)\rangle$ and $\ell$ is a two-dimensional logical operator. Since the boundary of a subpart of $\ell$ is zero-dimensional, the excitation energies associated with intermediate states are $O(1)$. Thus, the encoding with respect to $r$ is not self-correcting, and the qubit storage time is $\tau \sim O(1)$, independent of the system size. Therefore, while such a system is a good quantum code with the code distance $O(L)$, it works only as self-correcting \emph{classical} memory, with the bit storage time $\tau \sim \mbox{EXP}(L)$. However, in order for the system to work as self-correcting quantum memory, there must be a pair of anti-commuting two-dimensional logical operators. 

\textbf{Summary and discussion:}
We summarize coding properties of STS models based on dimensions of pairs of logical operators:
\begin{align}
\begin{array}{ccccccc}
\mbox{Spatial dim} & \mbox{Logical operators}  & \mbox{Code distance} & \mbox{Bit storage time} & \mbox{Qubit storage time} & \mbox{Self-correction} \\ \hline
\mbox{1 dim} & \mbox{0 dim + 1 dim} & O(1)     & O(1)         & O(1)&  \\
\mbox{2 dim} & \mbox{0 dim + 2 dim} & O(1)     &\mbox{EXP}(L) & O(1)&\mbox{classical}    \\
\mbox{2 dim} & \mbox{1 dim + 1 dim} & O(L)     & O(1)         & O(1)&    \\
\mbox{3 dim} & \mbox{0 dim + 3 dim} & O(1)     &\mbox{EXP}(L) & O(1)& \mbox{classical}    \\
\mbox{3 dim} & \mbox{1 dim + 2 dim} & O(L)     &\mbox{EXP}(L) & O(1)&\mbox{classical}     \\
\mbox{4 dim} & \mbox{2 dim + 2 dim} & O(L^{2}) &\mbox{EXP}(L) & \mbox{EXP}(L)&\mbox{quantum}
\end{array}\notag
\end{align}
where, for $D=4$, we presented coding properties of four-dimensional Toric code. 

\section{Thermal stability of topological order}\label{sec:topo}

In this section, we present another main result of the present paper, concerning the thermal stability of topological order at finite temperature by establishing the connection between the feasibility of self-correcting quantum memory and the stability of topological order. 

In section~\ref{sec:topo1}, we begin by establishing the connection between self-correcting classical memory and the thermal stability of ferromagnetic order at finite temperature. In section~\ref{sec:topo2}, we establish the connection between quantum codes and topological order at zero temperature. In section~\ref{sec:topo3}, we establish the connection between self-correcting quantum memory and topological order at finite temperature, and analyze the thermal stability of topological order arising in STS models. Note that discussion in this section is rather heuristic, and more rigorous treatment follows in~\ref{sec:topo_ap}

\subsection{Classical equivalence}\label{sec:topo1}

In this subsection, we establish the connection between self-correcting classical memory and thermal stability of ferromagnetic order at finite temperature:
\begin{align}
\begin{array}{ccc}
\mbox{Classical code} & \leftrightarrow & \mbox{Ferromagnetic order at $T=0$}\\
\mbox{Self-correcting classical memory} & \leftrightarrow & \mbox{Ferromagnetic order at $T>0$}
\end{array}\notag
\end{align}

\textbf{Two-dimensional ferromagnet:} We have seen that two-dimensional Ising model works as self-correcting classical memory since the energy barrier separating two degenerate ground states is $O(L)$. This self-correcting property of two-dimensional Ising model is closely related to the thermal stability of ferromagnetic order as seen from the following thermodynamic argument. 

Consider two-dimensional Ising model:
\begin{align}
H(\epsilon) \ = \ - \sum_{i,j}Z_{i,j}Z_{i+1,j} - \sum_{i,j}Z_{i,j}Z_{i,j+1} - \epsilon \sum_{i,j} Z_{i,j}
\end{align}
with an initial bias (symmetry-breaking field); $-\epsilon \sum_{i,j} Z_{i,j}$ for $\epsilon \geq 0$. As a result of an initial bias, the system is not degenerate anymore, and the ground state is $|0\cdots0\rangle$. The thermal stability of ferromagnetic order can be analyzed through the expectation value of the total magnetization $m$:
\begin{align}
m \ = \  \frac{1}{N} \sum_{i,j} Z_{i,j}
\end{align}
at finite temperature, which can be computed as follows:
\begin{align}
\langle m \rangle \ = \ \lim_{\epsilon \rightarrow 0} \frac{1}{\beta}\frac{\partial \log Z }{\partial \epsilon}
\end{align}
where the partition function is: $Z(\beta, \epsilon)  =  \text{Tr} e^{- \beta H(\epsilon)}$. Here, we evaluate the expectation value of $m$ at the limit where $\epsilon \rightarrow 0$ after we take the limit where $N$ goes to infinity. (Otherwise, there is no use of introducing an initial bias). Then, we have:
\begin{align}
\left\langle  \ \frac{1}{N} \sum_{i,j} Z_{i,j} \ \right\rangle_{\epsilon \rightarrow 0} \ &= \ 1 \qquad (T \ = \ 0) \\
1 \ > \ \left\langle  \ \frac{1}{N} \sum_{i,j} Z_{i,j} \ \right\rangle_{\epsilon \rightarrow 0} \ &> \ 0 \qquad ( T_{c} \ > T \ > \ 0) \\
\left\langle  \ \frac{1}{N} \sum_{i,j} Z_{i,j} \ \right\rangle_{\epsilon \rightarrow 0} \ &= \ 0 \qquad ( T \ > \ T_{c} ) 
\end{align}
where $T_{c}$ is some finite transition temperature, as plotted in Fig.~\ref{fig_topo_2}(a). 

One may notice that the total magnetization has some non-zero value as long as the temperature is below the transition temperature $T_{c}$. This indicates that the system is stable against thermal fluctuations for $T<T_{c}$. In particular, the system properties for $T_{c} > T > 0$ are close to the ground state properties at $T=0$. However, the system properties change drastically at $T=T_{c}$, and for $T > T_{c}$, the total magnetization vanishes. This indicates that the system undergoes a phase transition at $T=T_{c}$, and is unstable against thermal fluctuations for $T > T_{c}$.

\begin{figure}[htb!]
\centering
\includegraphics[width=0.70\linewidth]{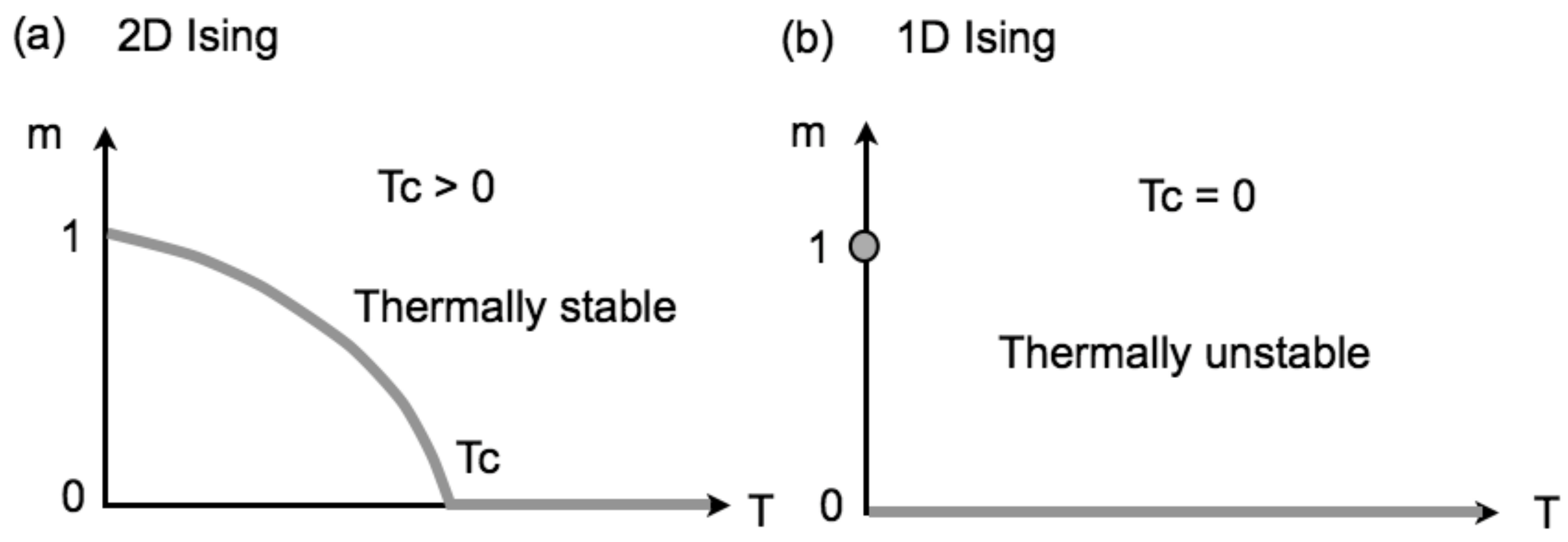}
\caption{Total magnetizations. (a) Two-dimensional Ising model. (b) One-dimensional Ising model.
} 
\label{fig_topo_2}
\end{figure}

Now, one can establish the connection between self-correcting properties and the thermal stability. In a coding theoretical language, the computation of the magnetization at finite temperature may be interpreted as follows. We first apply an initial bias and encode a bit $0$ to a ground state $|0\cdots\rangle$. Then, we let the system interact with the external environment at finite temperature, and wait until the system reaches the equilibrium while slowly turning off the initial bias. Finally, we decode the initially encoded bit by measuring the total magnetization. Therefore, self-correcting property of two-dimensional Ising model with exponentially long bit storage time is a direct consequence of the thermal stability of ferromagnetic order at finite temperature. 

In general, the total magnetization we have used above is called an \emph{order parameter} in condensed matter physics since it can be used to distinguish different phases of matters. On the other hand, quantum coding theoretically, one may view the total magnetization as a symmetric summation of a zero-dimensional logical operator $Z_{i,j}$ with which a classical bit is encoded. Later, we shall see that expectation values of logical operators can be used as order parameters to characterize the stability of topological order.

\textbf{One-dimensional ferromagnet:} Next, let us see that one-dimensional Ising model, which does not work as self-correcting classical memory, is not stable against thermal fluctuations. Consider one-dimensional Ising model:
\begin{align}
H(\epsilon) \ = \ - \sum Z_{j}Z_{j+1} - \epsilon \sum Z_{j}
\end{align}
with an initial bias, and compute the expectation values of the total magnetization. Then, we have
\begin{align}
\left\langle  \ \frac{1}{N} \sum_{j} Z_{j} \ \right\rangle_{\epsilon \rightarrow 0} \ = \ 1 \qquad ( T \ = \ 0) \\
\left\langle  \ \frac{1}{N} \sum_{j} Z_{j} \ \right\rangle_{\epsilon \rightarrow 0} \ = \ 0 \qquad ( T \ > \ 0)
\end{align}
as plotted in Fig.~\ref{fig_topo_2}(b). One may notice that the total magnetization becomes zero at any finite temperature. This indicates that the system is not stable against thermal fluctuations at any finite temperature, and the system undergoes a phase transition at $T=0$ as seen from the sudden change of the magnetization at $T=0$. The system properties for $T=0$ and for $T>0$ are significantly different, and the ground state properties are not stable against thermal fluctuations. This is the underlying reason why one-dimensional Ising model has $O(1)$ bit storage time, and does not work as self-correcting classical memory. 

\textbf{Summary of the equivalence:}
With observations above, one may notice that a large energy barrier, which is essential to self-correcting properties, is the key to the thermal stability of ferromagnetic order. In Fig.~\ref{fig_topo_3}, we give a summary of the equivalence concerning classical memories. However, it is not clear whether a large energy barrier is sufficient for the thermal stability of the system or not. We address this issue briefly in~\ref{sec:topo_ap3}.

\begin{figure}[htb!]
\centering
\includegraphics[width=0.70\linewidth]{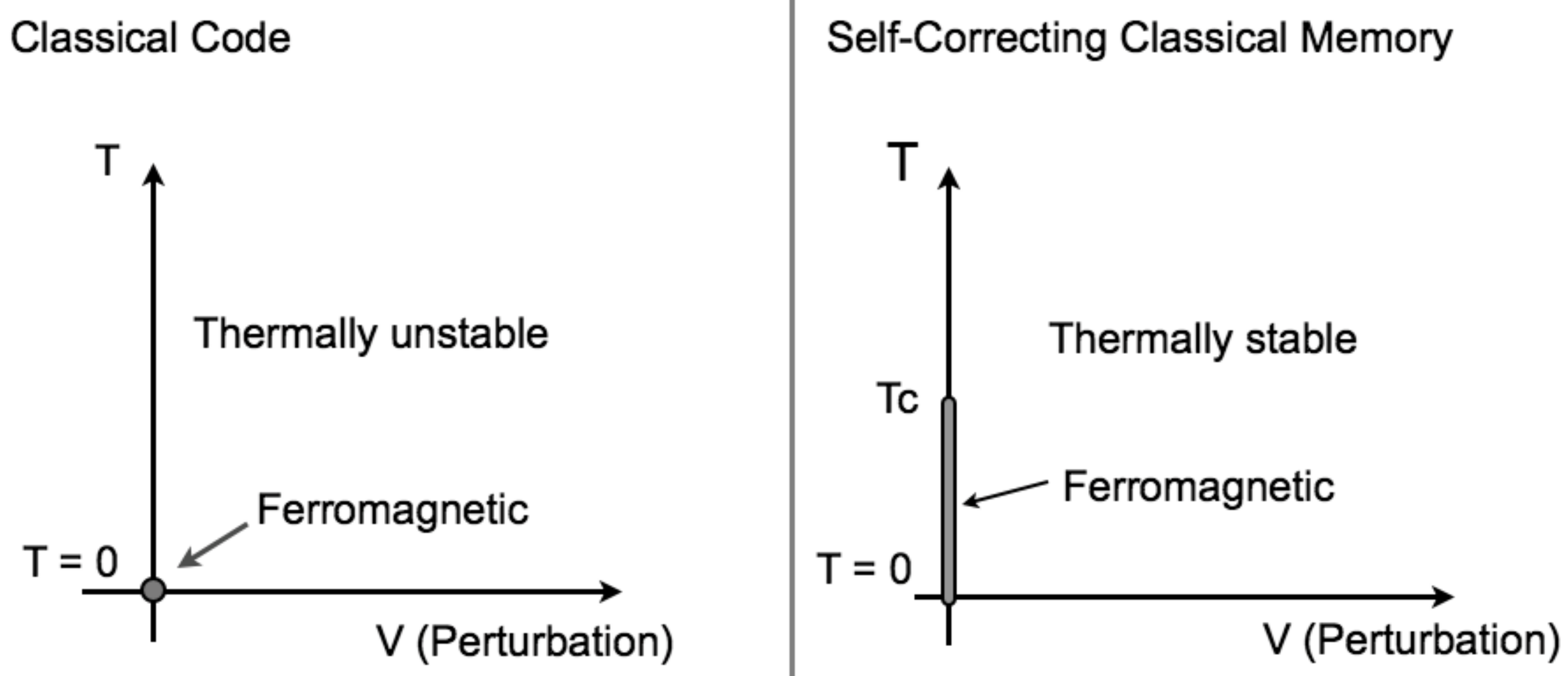}
\caption{A summary of the classical equivalence.
} 
\label{fig_topo_3}
\end{figure}

Here, it is worth noting that ferromagnetic systems, such as one-dimensional and two-dimensional Ising model are not stable against perturbations, as described in Fig.~\ref{fig_topo_3}. If one chooses $V$ as the initial bias, one can easily break the ground state degeneracy, and thus, the ground state property is not stable against perturbations. This indicates that ferromagnetic systems are not topologically ordered (see section~\ref{sec:topo2} too).

\subsection{Existence of topological order and quantum codes}\label{sec:topo2}

Now, we return to the main discussion in this section, concerning the stability of topological order at finite temperature. In this subsection and the next, we establish the connection between self-correcting quantum memory and the thermal stability of topological order:
\begin{align}
\begin{array}{ccc}
\mbox{Quantum code} & \leftrightarrow & \mbox{Topological order at $T=0$}\\
\mbox{Self-correcting quantum memory} & \leftrightarrow & \mbox{Topological order at $T>0$}
\end{array}\notag
\end{align}
In this subsection, we describe the definition of topological order in spin systems on a lattice, and argue that quantum codes have topological order at zero temperature. Note that our discussion closely follows pioneering works~\cite{Nussinov09, Bravyi10b}.

\textbf{Phenomenological definition:}
Let us first begin by describing the phenomenological definition of topological order which is commonly used in physics community. A system is said to have topological order when its ground state properties do not change significantly under any types of small perturbations~\cite{Wen90}. 

\begin{definition}[Stability against perturbations]\label{def:topo}
Consider a degenerate spin system defined on some closed geometric manifold governed by a Hamiltonian $H$. The system is said to be topologically ordered at $T=0$ if and only if the system satisfies the following conditions:
\begin{itemize}
\item There exists some finite positive number $\delta$ such that, for any perturbations $V$:
\begin{align}
H' \ = \ H + V, \qquad V \ = \ \sum_{j} V_{j}
\end{align}
where $V_{j}$ are locally defined and $|V_{j}|\leq \delta$, the ground state degeneracy is not broken at the thermodynamic limit (Fig.~\ref{fig_topo_def}). 
\item The perturbed ground space $G$ can be approximated by the original ground space $G_{0}$ through some local unitary transformation:
\begin{align}
UG_{0}U^{\dagger} \ \approx \ G
\end{align}
which can be represented as
\begin{align}
U \ = \ \int_{0}^{1} \exp(- i \hat{h} t)dt
\end{align}
where $\hat{h}$ is some hermitian operator which is a summation of local terms with finite amplitudes.
\end{itemize}
\end{definition}
 
\begin{figure}[htb!]
\centering
\includegraphics[width=0.70\linewidth]{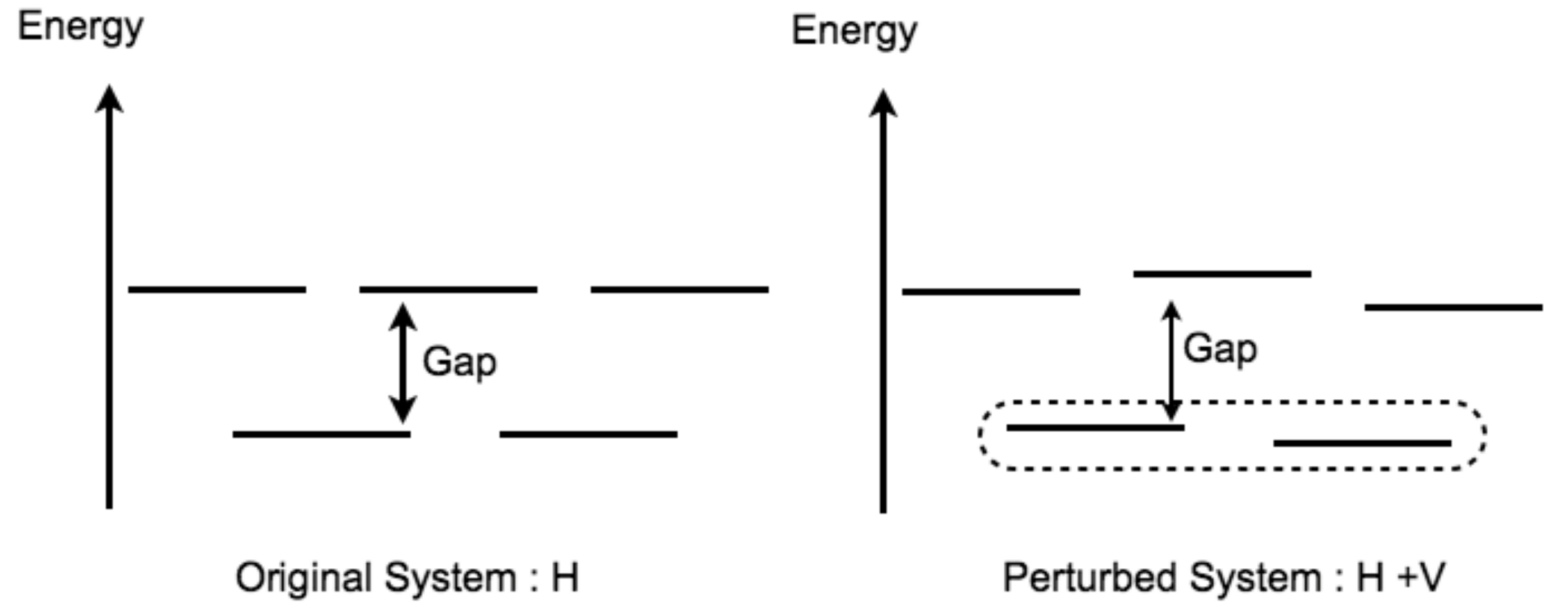}
\caption{The stability of the ground state properties against local perturbations. 
} 
\label{fig_topo_def}
\end{figure}
 
It is worth noting that a classical ferromagnet (Ising model) is not topologically ordered while it works as self-correcting classical memory. First of all, the ground state degeneracy of a classical ferromagnet is broken if a local magnetic field term $Z_{j}$ is added as a perturbation:
\begin{align}
V \ = \ - \epsilon \sum_{j} Z_{j}.
\end{align}
While the original Hamiltonian has $|0\cdots0\rangle$ and $|1\cdots1\rangle$ as ground states, the perturbed Hamiltonian has $|0\cdots0\rangle$ as a single ground state for positive $\epsilon$. Second, the ground state property significantly changes as a result of perturbations. While a classical ferromagnet has $\frac{1}{\sqrt{2}}(|0\cdots0\rangle + |1\cdots1\rangle)$ as a ground state, the perturbed ground state is $|0\cdots0\rangle$, and there is no local unitary transformation which transforms $\frac{1}{\sqrt{2}}(|0\cdots0\rangle + |1\cdots1\rangle)$ into $|0\cdots0\rangle$~\cite{Bravyi06}. As a straightforward extension of this discussion, one may notice that STS models with zero-dimensional logical operators are not topologically ordered, and topological order may exist only in two or higher-dimensional systems due to the dimensional duality of logical operators.

\textbf{Coding theoretical definition:}
While the definition above captures the stability of ground state properties in topologically ordered systems, one may wonder what kinds of systems are actually topologically ordered. Also, one may wonder if there is a universal feature commonly shared among all the topologically ordered systems. 

Somewhat surprisingly, a quantum coding theoretical viewpoint gives a hint to answer these questions by providing an alternative definition of topological order which may capture universal properties of topologically ordered systems~\cite{Bravyi06}. In fact, a system of spins on a lattice may be said to be topologically ordered when the ground space is separated from excited states by a finite energy gap at the thermodynamic limit and the ground space realizes a quantum code whose code distance $d$ is comparable to the system size~\cite{Bravyi10b}. One may rewrite this characterization of topological order through coding properties more explicitly in the following way.\footnote{Strictly speaking, this definition is not complete. See~\cite{Bravyi10b} for a complete definition.} 

\begin{definition}[Quantum code]\label{def:red}
Consider a degenerate spin system defined on some closed geometric manifold. The system is said to be topologically ordered when it satisfies the following condition: 
\begin{itemize}
\item Let the degenerate ground state space be $G$, which is protected by a finite energy gap. Consider reduced density matrices of ground states for a region $R$ which is contractable to a single point: 
\begin{align} 
\rho_{R}(|\psi\rangle) \ \equiv \ Tr_{\bar{R}} (|\psi\rangle\langle\psi|).
\end{align}
So, $R$ is a topologically trivial \emph{zero-dimensional region} (Fig.~\ref{fig_reduce}). Then, all the reduced density matrices of degenerate ground states are the same at the thermodynamic limit
\begin{align}
\rho_{R}(|\psi\rangle) \ = \ \rho_{R}(|\psi'\rangle) \qquad \mbox{for all} \quad |\psi\rangle, \ |\psi'\rangle \ \in \ G \quad \mbox{at}\quad N \ \rightarrow \ \infty.
\end{align}
\end{itemize}
\end{definition}

\begin{figure}[htb!]
\centering
\includegraphics[width=0.60\linewidth]{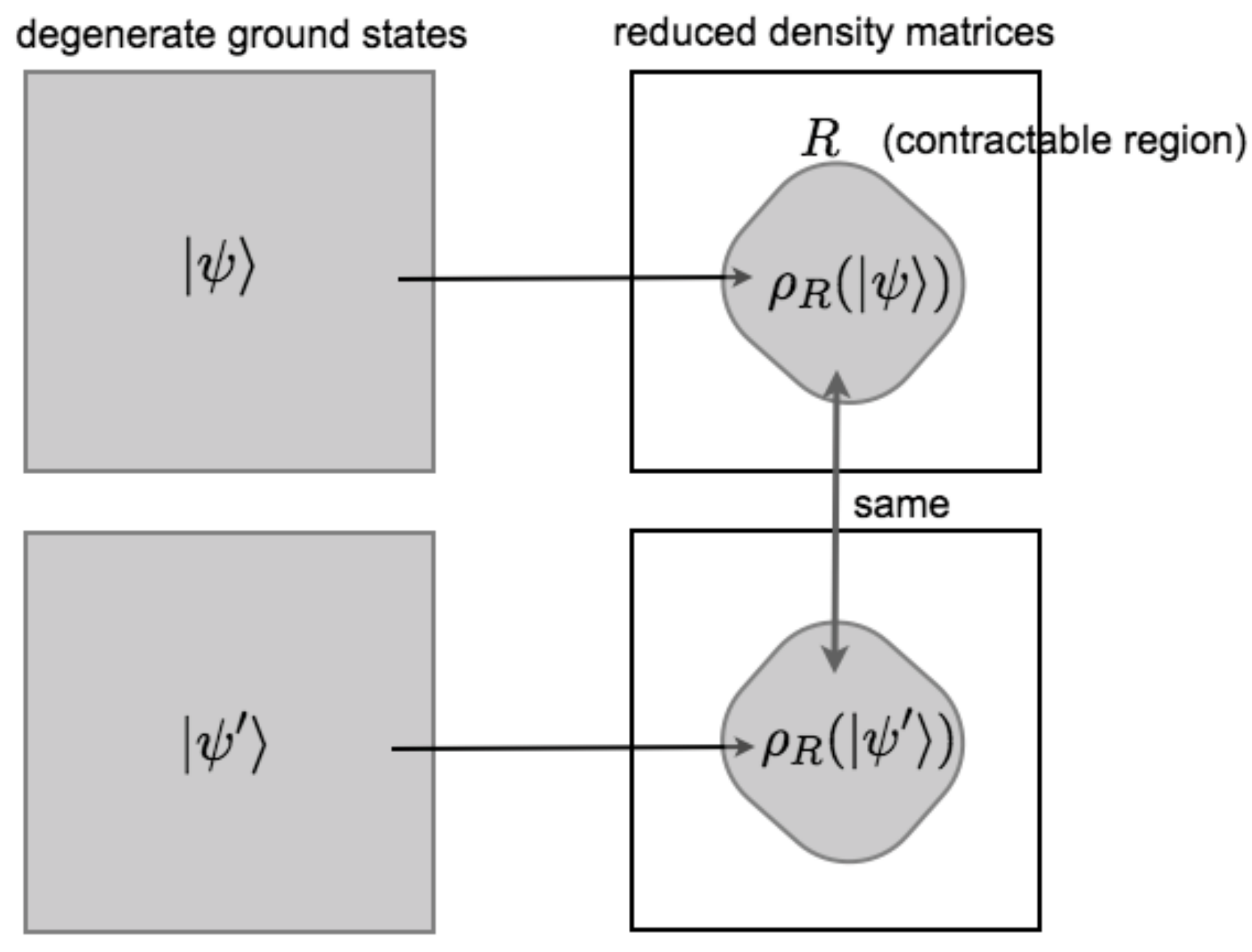}
\caption{A coding theoretical definition of topological order.
} 
\label{fig_reduce}
\end{figure}

Let us see why the condition above ensures that the system serves as a quantum code with a macroscopic code distance. When all the reduced density matrices are the same for all the topologically trivial regions, any errors acting inside such regions cannot change the encoded logical qubits. In order to change logical qubits, some global operator acting on a region with a topologically non-trivial geometric shape is necessary. Such a global operator has a weight equal to the code distance. Therefore, the ground space of such a system works as a good quantum code with a macroscopic code distance.\footnote{The definition used here is slightly different from the one used in~\cite{Nussinov09} where the existence of the energy gap is not required. The difference comes apparent whether a two-dimensional Bacon's subsystem code~\cite{Bacon06} is classified as topologically ordered system or not. The ground state space of this subsystem code has the same reduced density matrix for any zero-dimensional regions. However, it is not likely to have a finite energy gap since it is highly frustrated, and thus, its coding properties are likely to be unstable against perturbations.}

\textbf{Equivalence of two definitions:}
Now, let us see that definition~\ref{def:red} implies definition~\ref{def:topo}. Under a small perturbation $V$ which is a summation of local terms, we perform a perturbative analysis. Since the ground states are separated by excited states by a finite energy gap, only the coupling between degenerate ground states is dominant. Since there is no local operator to couple two ground states, one needs to consider $O(L)$th order perturbation to have a coupling between degenerate ground states. Since such a coupling will be exponentially suppressed by a factor $\exp(- L / L_{0})$ where $L_{0}$ is some finite length scale, and the energy splitting is exponentially suppressed and goes to zero as $L$ goes to infinity. Thus, the ground state degeneracy is protected at the thermodynamic limit. When the ground state degeneracy is protected and the ground space is separated from excited states by a finite energy gap, one can approximate the ground space from the original unperturbed ground space according to the theory of adiabatic continuation~\cite{Bravyi10b, Hastings05}. A complete mathematical proof was presented in~\cite{Bravyi10b}. However, it is not known if definition~\ref{def:topo} implies definition~\ref{def:red} or not. 

Here, we briefly mention the usage of the expression ``topological order'' since it is sometimes used in a fuzzy and elusive way. In the present paper, by topological order, we mean the ground state properties of topologically ordered systems. By its definition, topological order is stable against any types of small local perturbations, meaning that the ground state degeneracy is protected and perturbed ground states can be approximated through local unitary transformations. Therefore, by the ground state properties, we mean any global features of ground states which are free from length scales and are not affected by local unitary transformations. We will return to discussion on the meaning of ``global features'' later.

\textbf{Topological order in STS models:}
Let us conclude this subsection on the definition of topological order by seeing that two-dimensional and three-dimensional STS models can have topological order according to definition~\ref{def:red}. Consider the case when two-dimensional STS model has $k$ pairs of one-dimensional logical operators. Let us split the entire system into two complementary regions $P(n_{1}-1,n_{2}-1)$ and $\overline{P(n_{1}-1,n_{2}-1)}$ where $P(n_{1}-1,n_{2}-1)$ is a region with $n_{1}-1 \times n_{2}-1$ composite particles. Note that $P(n_{1}-1,n_{2}-1)$ is topologically trivial zero-dimensional region since one can shrink it into a single point continuously. Since there are $2k$ logical operators defined inside $\overline{P(n_{1}-1,n_{2}-1)}$, we have $g_{\overline{P(n_{1}-1,n_{2}-1)}}=2k$. However, this means that $g_{P(n_{1}-1,n_{2}-1)}=0$ from theorem~\ref{theorem_partition}, and there is no logical operator defined inside $P(n_{1}-1,n_{2}-1)$. Then, all the reduced density matrices are the same, and such a two-dimensional STS model is topologically ordered since it satisfies definition~\ref{def:red}, and thus, definition~\ref{def:topo}. A similar discussion holds for a three-dimensional STS model with pairs of one-dimensional and two-dimensional logical operators. 

\subsection{Quantum equivalence}\label{sec:topo3}

In this subsection, we establish the connection between self-correcting quantum memory and topological order at finite temperature, and analyze the thermal stability of topological order arising in STS models.

\textbf{Two-dimensional Toric code:} Two-dimensional Toric code is topologically ordered since it has a macroscopic code distance: $d \sim O(1)$. Yet, it does not work as self-correcting quantum memory since the energy barrier is $\Delta E \sim O(1)$. This is closely related to the thermal instability of topological order as we shall see below. Let us perform a thermodynamic analysis on two-dimensional Toric code in a way similar to Ising model by adding an initial bias:
\begin{align}
H_{\ell}(\epsilon) \ = \ H_{Toric} - \epsilon \sum_{j=1}^{L} T_{2}^{j}(\ell) \\
H_{r}(\epsilon) \ = \ H_{Toric} - \epsilon \sum_{i=1}^{L} T_{1}^{i}(r)
\end{align}
where $\ell$ and $r$ are one-dimensional logical operators extending in the $\hat{1}$ and $\hat{2}$ directions respectively. Here, we define the following ``normalized logical operators'':
\begin{align}
m_{\ell} \ &= \ \frac{1}{L}\sum_{j=1}^{L} T_{2}^{j}(\ell) \\
m_{r} \ &= \ \frac{1}{L}\sum_{i=1}^{L} T_{1}^{i}(r)
\end{align}
by taking symmetric summations of logical operators, as we did for Ising model. Then, we have
\begin{align}
\left\langle   m_{\ell} \right\rangle_{\epsilon \rightarrow 0} \ = \ \left\langle   m_{r} \right\rangle_{\epsilon \rightarrow 0} \ = \ 1 \qquad ( T \ = \ 0) \\
\left\langle   m_{\ell} \right\rangle_{\epsilon \rightarrow 0} \ =  \ \left\langle   m_{r} \right\rangle_{\epsilon \rightarrow 0} \ = \ 0 \qquad ( T \ > \ 0)
\end{align}
as plotted in Fig.~\ref{fig_topo_4}(a) where $\langle m_{\ell}\rangle$ is evaluated for $H_{\ell}(\epsilon)$, and $\langle m_{r}\rangle$ is evaluated for $H_{r}(\epsilon)$. This indicates that the system is not stable against thermal fluctuations at any finite temperature, meaning that topological order arising in two-dimensional Toric code is thermally unstable. This implies that one cannot read out initially encoded qubit by measuring $m_{\ell}$ and $m_{r}$.

\textbf{Three-dimensional Toric code:} Next, let us consider thermodynamic properties of three-dimensional Toric code. While three-dimensional Toric code does not work as self-correcting quantum memory, it works as self-correcting \emph{classical} memory. This is because it has pairs of one-dimensional and two-dimensional logical operators, and as a result, the bit storage time is $\tau \sim \mbox{EXP($L$)}$ while the qubit storage time is $\tau \sim O(1)$. These coding properties are closely related to thermodynamic properties of three-dimensional Toric code, as seen from expectation values of logical operators:
\begin{align}
m_{\ell} \ &= \ \frac{1}{L}\sum_{z=1}^{L} T_{3}^{z}(\ell) \\
m_{r} \ &= \ \frac{1}{L^{2}}\sum_{x,y=1}^{L} T_{1}^{x}T_{2}^{y}(r)
\end{align}
where $\ell$ is a two-dimensional logical operator extending in the $\hat{1}$ and $\hat{2}$ directions, while $r$ is a one-dimensional logical operator extending in the $\hat{3}$ direction. Their expectation values behave as follows:
\begin{align}
\left\langle   m_{\ell} \right\rangle_{\epsilon \rightarrow 0} \ = \ 1 \qquad ( T \ = \ 0) \\
\left\langle   m_{\ell} \right\rangle_{\epsilon \rightarrow 0} \ = \ 0 \qquad ( T \ > \ 0)
\end{align}
and
\begin{align}
\left\langle  m_{r} \right\rangle_{\epsilon \rightarrow 0} \ &= \ 1 \qquad (T \ = \ 0) \\
1 \ > \ \left\langle m_{r} \right\rangle_{\epsilon \rightarrow 0} \ &> \ 0 \qquad ( T_{c} \ > T \ > \ 0) \\
\left\langle m_{r} \right\rangle_{\epsilon \rightarrow 0} \ &= \ 0 \qquad ( T \ > \ T_{c} ) 
\end{align}
where $T_{c}$ is some finite transition temperature, as plotted in Fig.~\ref{fig_topo_4}(b). This implies that three-dimensional Toric code undergoes phase transitions both at $T=0$ and $T=T_{c}$, and the ground state properties are not completely stable against thermal fluctuations at any finite temperature. Yet, the ground state properties partially survive at finite temperature as a direct consequence of being self-correcting classical memory. 

\begin{figure}[htb!]
\centering
\includegraphics[width=0.70\linewidth]{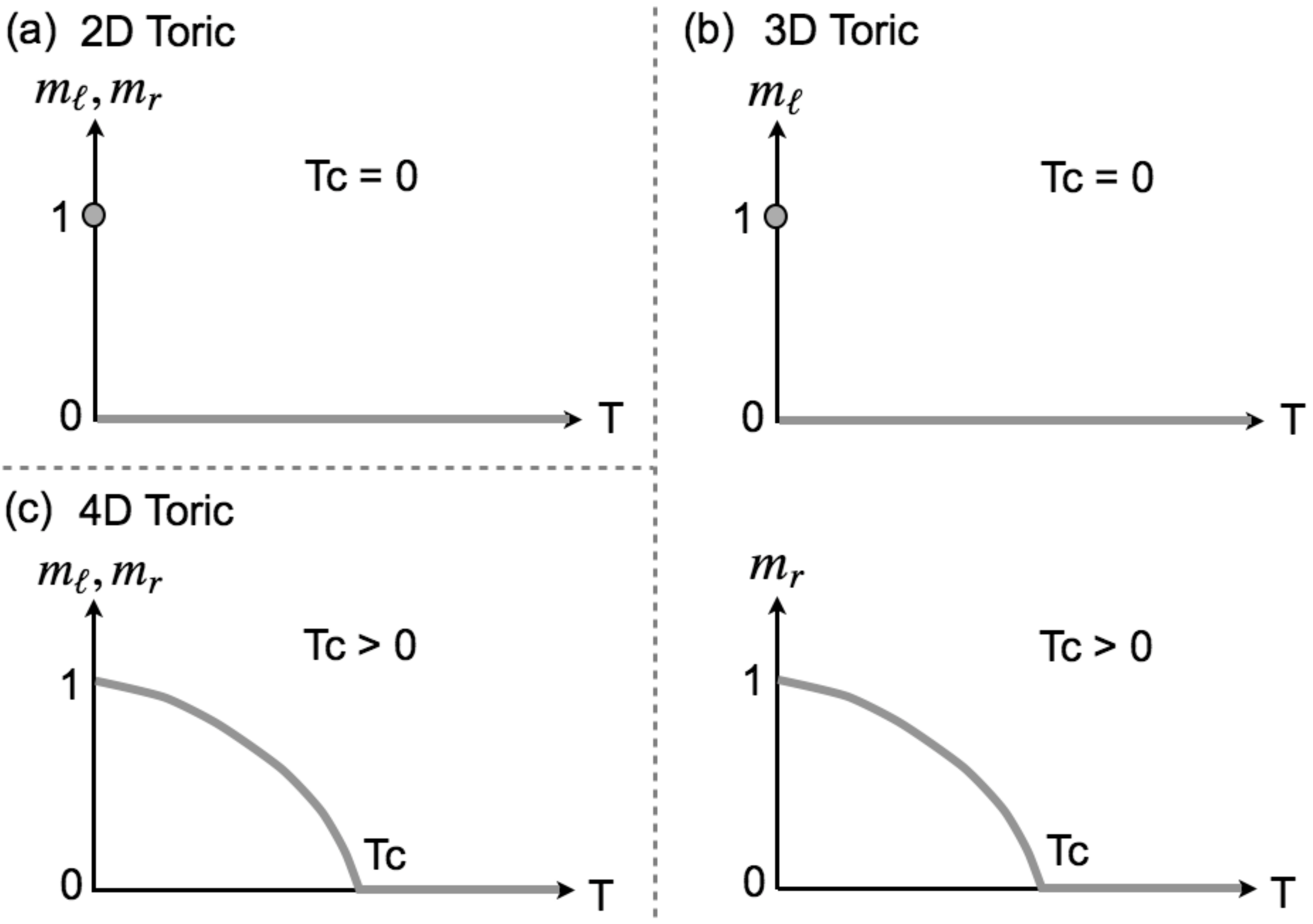}
\caption{Expectation values of logical operators. (a) Two-dimensional Toric code. (b) Three-dimensional Toric code. (c) Four-dimensional Toric code. 
} 
\label{fig_topo_4}
\end{figure}

\textbf{Four-dimensional Toric code:} Finally, let us see that four-dimensional Toric code, which works as self-correcting quantum memory, is stable against thermal fluctuations. Here, we define 
\begin{align}
m_{\ell} \ = \ \frac{1}{L^{2}}\sum_{z,w=1}^{L} T_{3}^{z}T_{4}^{w}(\ell) \\
m_{r} \ = \ \frac{1}{L^{2}}\sum_{x,y=1}^{L} T_{1}^{x}T_{2}^{y}(r)
\end{align}
where $\ell$ is a two-dimensional logical operator extending in the $\hat{1}$ and $\hat{2}$ directions, and $r$ is a two-dimensional logical operator extending in the $\hat{3}$ and $\hat{4}$ directions. Then, we have
\begin{align}
\left\langle  m_{\ell} \right\rangle_{\epsilon \rightarrow 0} \ &= \ \left\langle  m_{r} \right\rangle_{\epsilon \rightarrow 0} \ = \ 1 \qquad (T \ = \ 0) \\
1 \ > \ \left\langle  m_{\ell} \right\rangle_{\epsilon \rightarrow 0} \ &= \ \left\langle m_{r} \right\rangle_{\epsilon \rightarrow 0} \ > \ 0 \qquad ( T_{c} \ > T \ > \ 0) \\
\left\langle  m_{\ell} \right\rangle_{\epsilon \rightarrow 0} \ &= \ \left\langle m_{r} \right\rangle_{\epsilon \rightarrow 0} \ = \ 0 \qquad ( T \ > \ T_{c} ) 
\end{align}
where $T_{c}$ is some finite transition temperature, as plotted in Fig.~\ref{fig_topo_4}(c). This implies that the ground state properties are stable against thermal fluctuations and topological order arising in four-dimensional Toric code is stable at finite temperature. 

\textbf{Summary of the equivalence:}
With these observations, one may notice that large energy barrier inside the ground space, which is essential to self-correcting properties, is the key to the thermal stability of topological order. In Fig.~\ref{fig_topo_7}, we give a summary of the equivalence concerning quantum memory. 

\begin{figure}[htb!]
\centering
\includegraphics[width=0.70\linewidth]{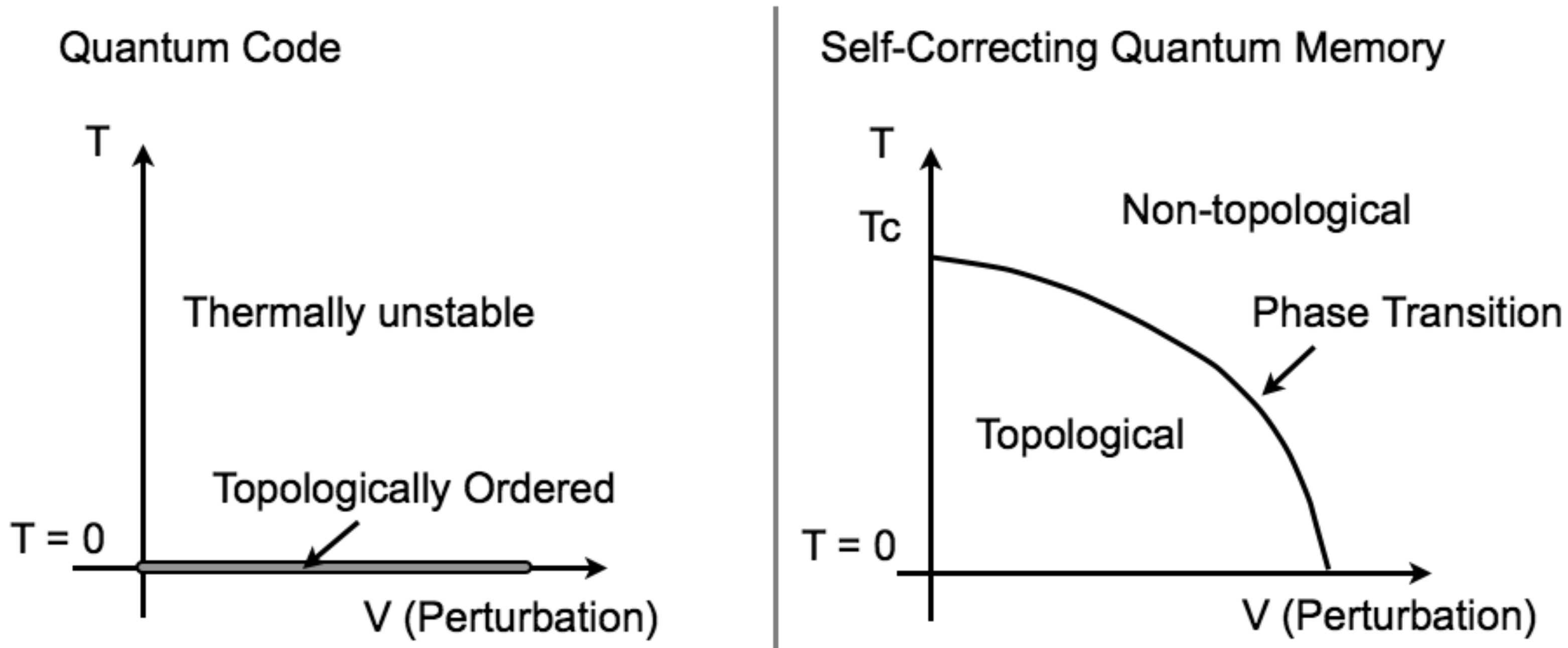}
\caption{The quantum equivalence.
} 
\label{fig_topo_7}
\end{figure}

With this connection between the feasibility of self-correcting quantum memory and the thermal stability of topological order, we conclude that topological order arising in STS models is not stable against thermal fluctuations. In other words, for $D \leq 3$, there is no system which is stable against both thermal fluctuations and local perturbations simultaneously. While discussion here is rather heuristic, we give a more rigorous treatment on the definition of the thermal stability of topological order in~\ref{sec:topo_ap}. We summarize physical properties of STS models based on dimensions of pairs of logical operators:
\begin{align}
\begin{array}{ccccccc}
\mbox{Spatial dim} & \mbox{Logical operators}  & \mbox{Local perturbations} & \mbox{Thermal fluctuations} & \mbox{Memory property}  \\ \hline
\mbox{1 dim} & \mbox{0 dim + 1 dim} &               &               & \mbox{Classical code}    \\
\mbox{2 dim} & \mbox{0 dim + 2 dim} &               & \mbox{stable} & \mbox{Classical self-correction}    \\
\mbox{2 dim} & \mbox{1 dim + 1 dim} & \mbox{stable} &               & \mbox{Quantum code}    \\
\mbox{3 dim} & \mbox{0 dim + 3 dim} &               & \mbox{stable} & \mbox{Classical self-correction}    \\
\mbox{3 dim} & \mbox{1 dim + 2 dim} & \mbox{stable} &               & \mbox{Quantum code}     \\
\mbox{4 dim} & \mbox{2 dim + 2 dim} & \mbox{stable} & \mbox{stable} & \mbox{Quantum self-correction}
\end{array}\notag
\end{align}
where, for $D=4$, we presented coding properties of four-dimensional Toric code. 

\section{Emergence of topology in logical operators}\label{sec:physics}

So far, we have addressed specific questions concerning coding and physical properties of gapped spin systems on a lattice. While these two questions are of particular importance in quantum information science and condensed matter physics, the ultimate goal is to find universal properties of arbitrary gapped spin systems on a lattice and develop a unified theoretical framework to discuss their coding and physical properties. In particular, it would be very beautiful if arbitrary gapped spin systems can be analyzed universally through some unified theoretical tool which is yet to be discovered. 

In this section, as a necessary first step toward this goal, we make an attempt to find universal properties of logical operators which are commonly shared among all the STS models. Logical operators in STS models have a certain interesting topological property which is a direct consequence of physical constraints we have imposed on stabilizer Hamiltonians. Here, we show that one can always deform geometric shapes of logical operators continuously while keeping them equivalent by applying appropriate stabilizers as long as one does not change geometric shapes in a topologically non-trivial way. We prove this continuous deformability of logical operators in STS models for $D=1,2,3$ with various examples of continuous deformations of logical operators. Some of these results were also presented in our previous papers~\cite{Beni10, Beni10b}, especially for $D=2$.

The program of finding a universal theory of gapped spin systems is continued in~\ref{sec:topology} where the role of topology in analyzing coding and physical properties of stabilizer codes is further examined, while we concentrate on demonstrating topological properties of logical operators in this section. 

\textbf{One-dimension:}
We begin by illustrating a topological property of logical operators for one-dimensional STS models. Recall that $P(x)$ represents a region of $x$ composite particles, and $g_{P(1)} = k$ in one-dimensional systems regardless of the system size since there always exist $k$ zero-dimensional logical operators. Then, as a result of the bi-partition theorem (theorem~\ref{theorem_partition}), we have 
\begin{align}
g_{P(1)} \ = \ g_{\overline{P(1)}} \ = \ k.
\end{align}

\begin{figure}[htb!]
\centering
\includegraphics[width=0.45\linewidth]{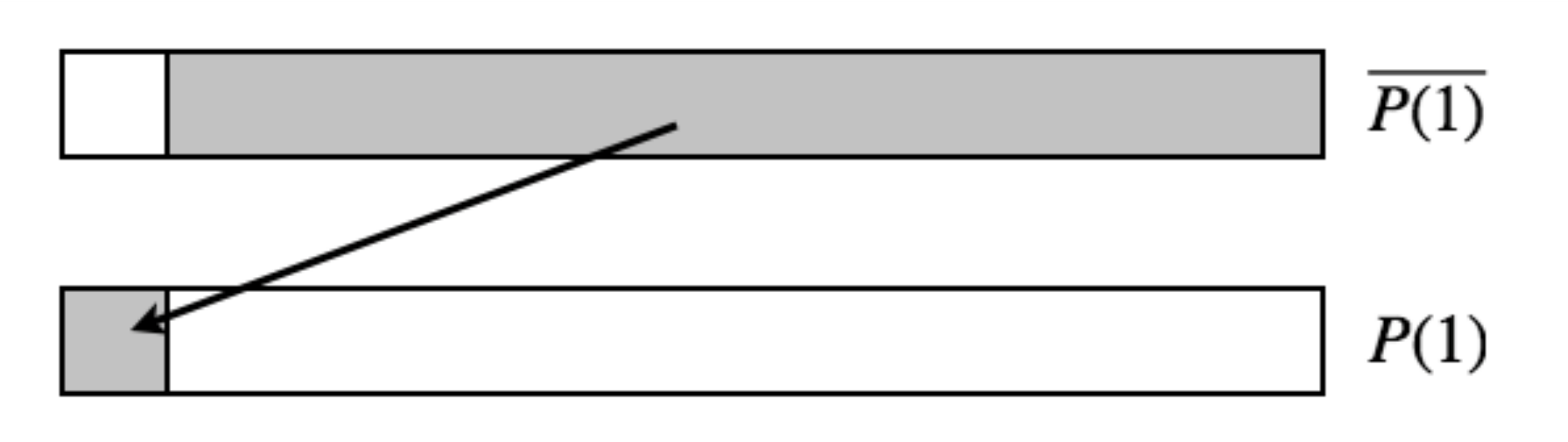}
\caption{A deformation of a zero-dimensional logical operator.
} 
\label{deform_1dim}
\end{figure}

Now, let us interpret the equation above from a geometric viewpoint. Consider some logical operator $\ell$ defined inside $\overline{P(1)}$. Then, since $g_{P(1)}=g_{\overline{P(1)}}=k$, there must exist some equivalent logical operator $\ell' \sim \ell$ which is defined inside $P(1)$. In other words, if there exists a logical operator defined inside $\overline{P(1)}$, one can shrink its geometric shape into $P(1)$ by applying some stabilizer (Fig.~\ref{deform_1dim}):
\begin{align}
\overline{P(1)} \ \rightarrow \ P(1).
\end{align}
Here, one may notice that $\overline{P(1)}$ and $P(1)$ are topologically equivalent regions since one can deform $\overline{P(1)}$ into $P(1)$ continuously. On the other hand, if a logical operator is a one-dimensional logical operator defined all over the lattice (i.e. defined inside $P(n)$ where $n$ is the linear length of the system), one may not be able to deform the logical operator into $P(1)$. 

\textbf{Two-dimensions:} Next, we analyze a topological property of logical operators for two-dimensional STS models. For the convenience of presentation, we assume that zero-dimensional logical operators defined inside $P(1,2v)$ can be actually defined inside $P(1,1)$ in a two-dimensional STS model. This may be done through some appropriate coarse-graining. 

We first list regions which serve as references to classify geometric shapes of logical operators in a two-dimensional system. We define \emph{topological unit regions} as follows (Fig.~\ref{fig_2D_unit}(a)):
\begin{align}
Q(0,0) \equiv P(1,1), \quad Q(1,0) \equiv P(n_{1},1), \quad Q(0,1) \equiv P(1,n_{2}), \quad Q(1,1) \equiv P(n_{1},n_{2}).
\end{align}
``$1$'' and ``$0$'' represent whether a region extends in the corresponding direction or not. For example, $Q(1,0)$ and $Q(0,1)$ are one-dimensional unit regions which extend in the directions of $\hat{1}$ and $\hat{2}$ respectively. $Q(0,0)$ is a zero-dimensional unit region with a single composite particle. $Q(1,1)$ is a two-dimensional unit region which consists of all the composite particles in the system. These topological unit regions are shown graphically in Fig.~\ref{fig_2D_unit}(a). We also denote a union of all the $m$-dimensional topological unit regions as $R_{m}$:
\begin{align}
R_{0} \equiv Q(0,0), \quad R_{1} \equiv Q(1,0) \cup Q(0,1), \quad R_{2} \equiv Q(1,1).
\end{align}
We call $R_{m}$ \emph{$m$-dimensional concatenated topological unit regions}. All the concatenated unit regions are graphically shown in Fig.~\ref{fig_2D_unit}(b). 

\begin{figure}[htb!]
\centering
\includegraphics[width=0.70\linewidth]{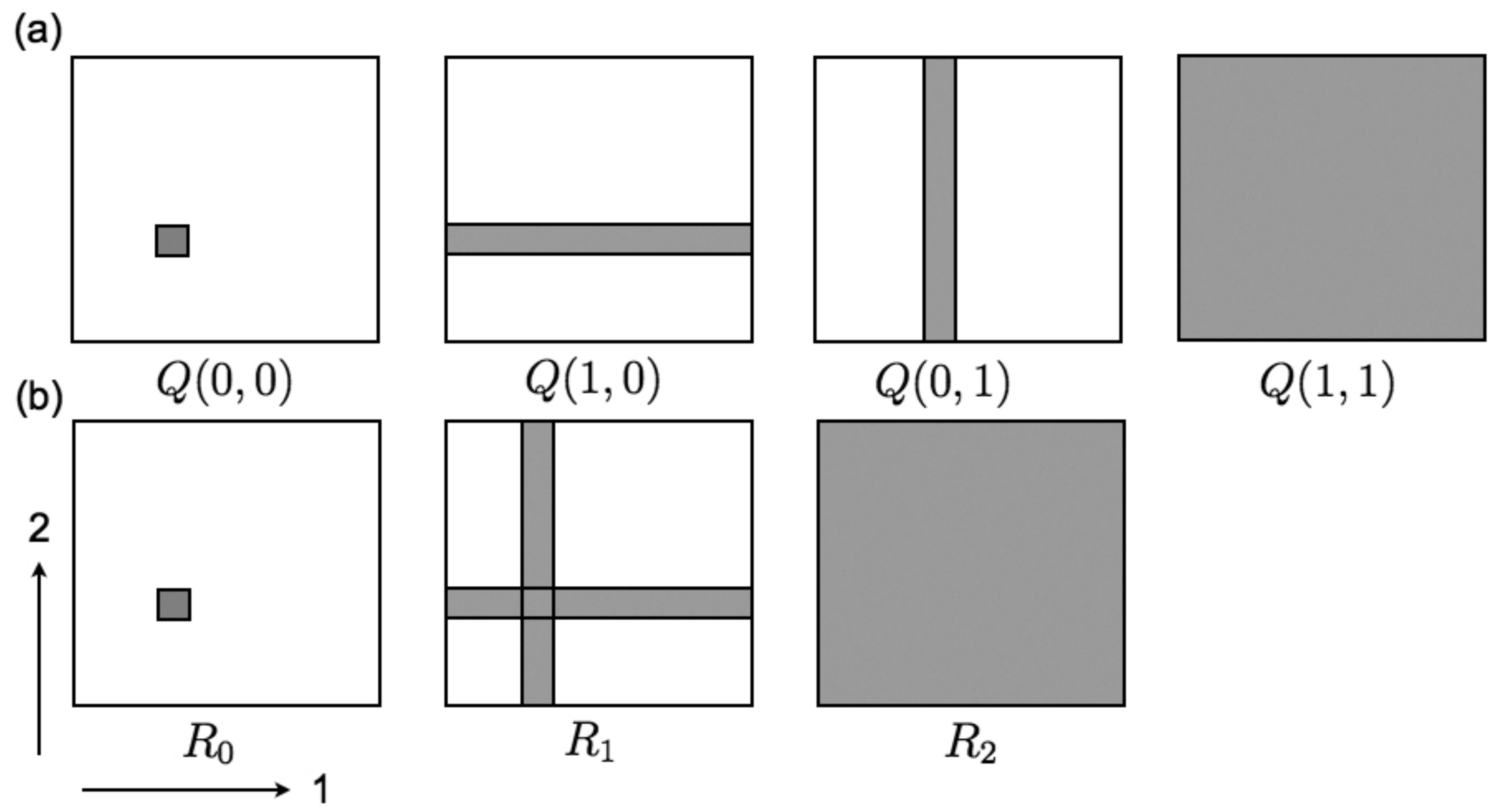}
\caption{Reference regions. (a) Topological unit regions. (b) Concatenated unit regions.
} 
\label{fig_2D_unit}
\end{figure}

In a two-dimensional system, there are five different unions of topological unit regions: $R_{0}$, $Q(1,0)$, $Q(0,1)$, $R_{1}$ and $R_{2}$. We call these regions, except $R_{2}$, \emph{reference regions}, whose set is denoted as $R_{ref}$:
\begin{align}
R_{ref} = \{ R_{0}, Q(1,0), Q(0,1), R_{1} \}.
\end{align}
Then, one can introduce equivalence relations between these reference regions and their complements in terms of continuous deformations. For example, as shown in Fig.~\ref{fig_deformation}(a), $\overline{R_{0}}$ can be continuously deformed into $R_{1}$ by enlarging the hole of $\overline{R_{0}}$ gradually. Also, as shown in Fig.~\ref{fig_deformation}(b), $\overline{R_{1}}$ can be continuously deformed into $R_{0}$ since both $\overline{R_{1}}$ and $R_{0}$ are zero-dimensional regions without any winding around the torus. Finally, as shown in Fig.~\ref{fig_deformation}(c), $\overline{Q(1,0)}$ can be deformed into $Q(1,0)$ since both regions have a winding in the $\hat{1}$ direction. In summary, we have the following equivalence relations between reference regions and their complements:
\begin{align}
\overline{R_{0}} \simeq R_{1}, \qquad \overline{R_{1}}\simeq R_{0}, \qquad \overline{Q(1,0)}\simeq  Q(1,0), \qquad \overline{Q(0,1)} \simeq  Q(0,1).
\end{align}

\begin{figure}[htb!]
\centering
\includegraphics[width=0.70\linewidth]{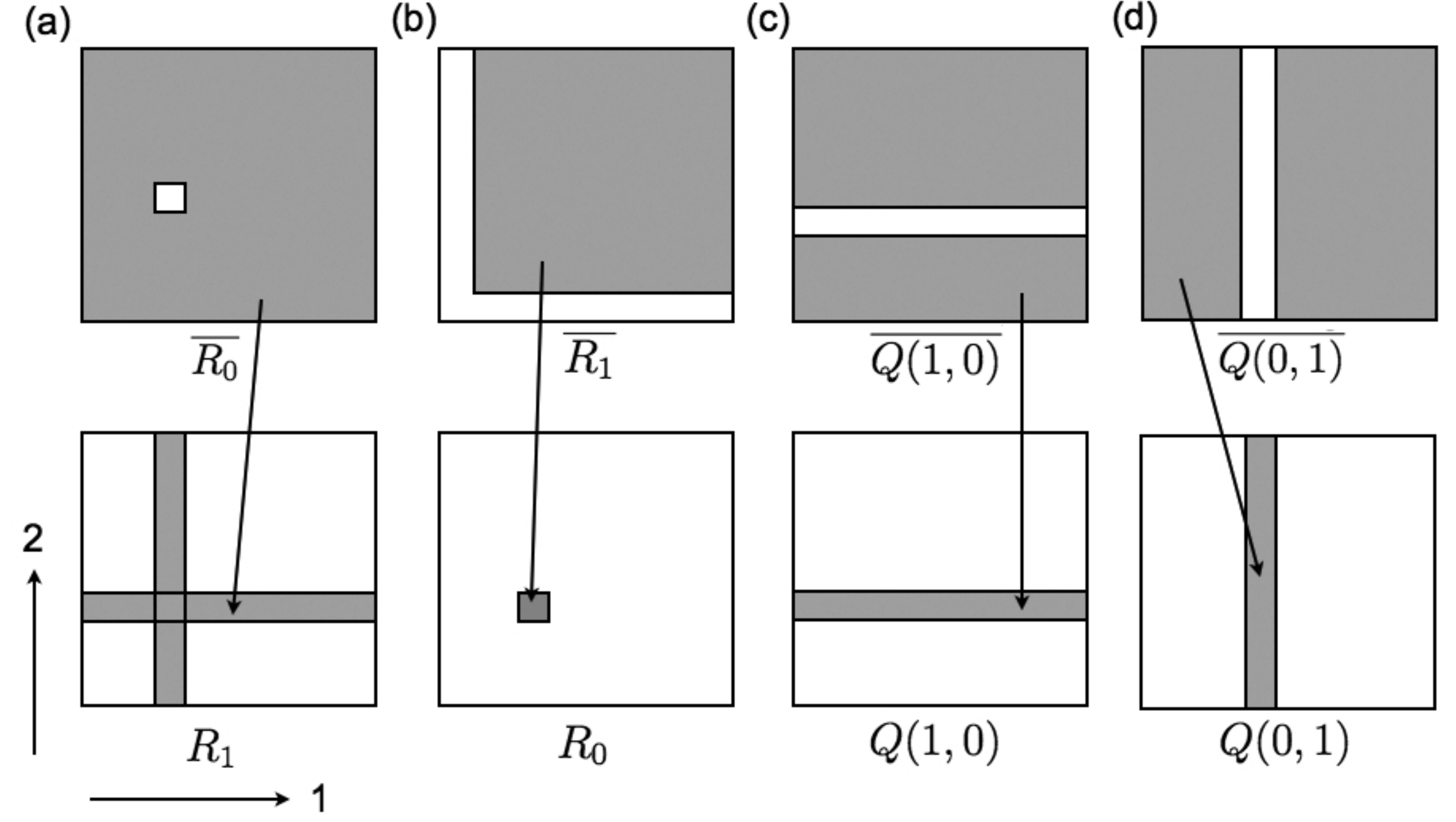}
\caption{The topological deformations of logical operators.} 
\label{fig_deformation}
\end{figure}

Now, we discuss how geometric shapes of logical operators can be determined. A useful observation regarding geometric shapes of logical operators can be obtained by considering the number of independent logical operators defined inside a region $R$. Let the number of independent logical operators inside $R$ be $g_{R}$. Here, we consider the case where we have two regions $R$ and $R'$ where $R$ is larger than $R'$, meaning that $R$ includes all the composite particles inside $R'$ (Fig.~\ref{fig_deform_intuition}). Then, if $g_{R}=g_{R'}$, $R$ and $R'$ support the same logical operators since all the logical operators defined inside $R$ have equivalent representations which are supported inside $R'$. This means that, for a given logical operators $\ell$ defined inside $R$, one can always find another equivalent logical operator $\ell'$ defined inside $R'$. In other words, one can \emph{deform} the geometric shape of $\ell$ into $\ell'$ by applying some appropriate stabilizer (Fig.~\ref{fig_deform_intuition}). Thus, by finding two connected regions $R$ and $R'$ where $R$ is larger than $R'$ and $g_{R}=g_{R'}$, one can conclude that logical operators defined inside $R$ can be deformed into $R'$. 

\begin{figure}[htb!]
\centering
\includegraphics[width=0.40\linewidth]{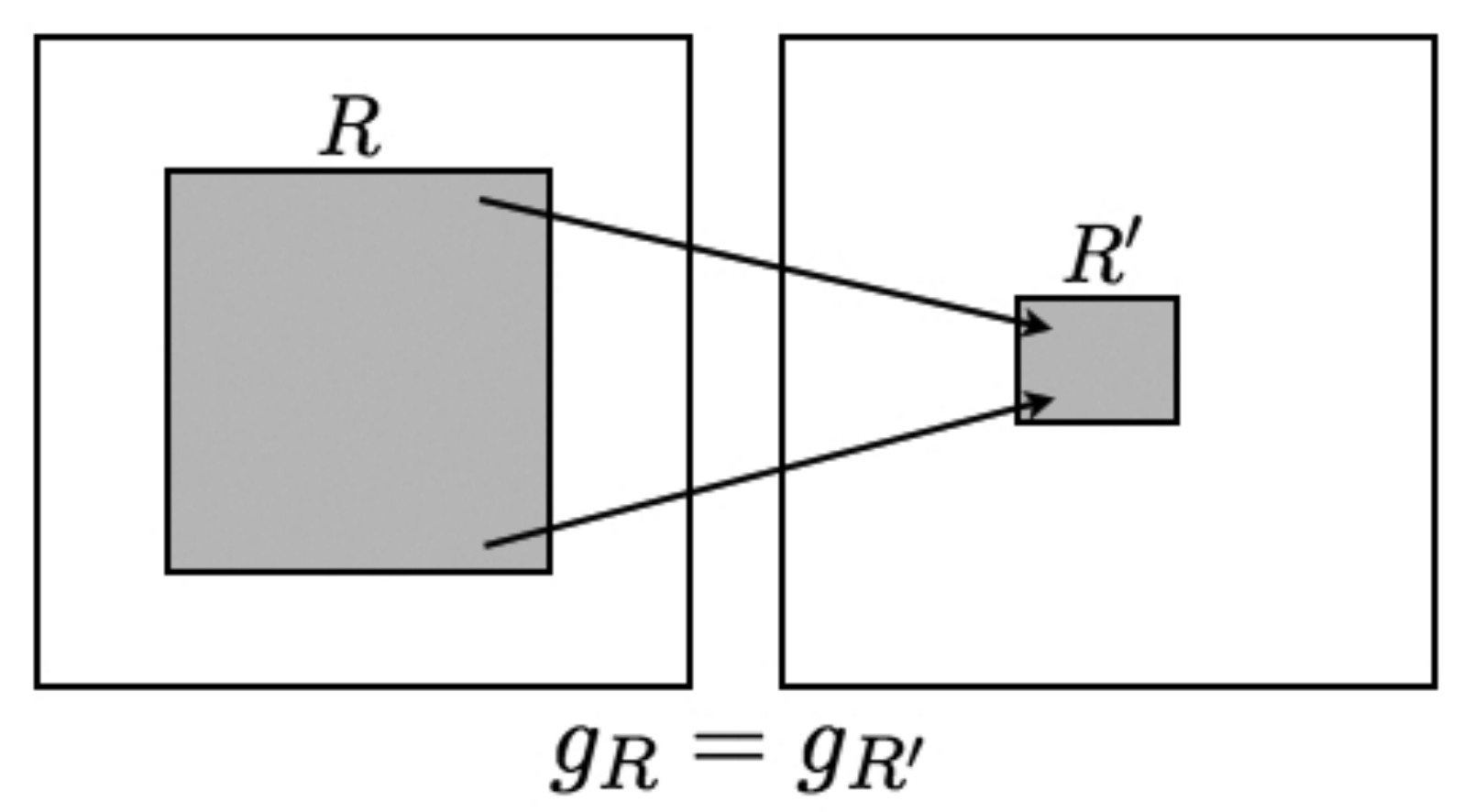}
\caption{A shrinkage from $R$ to $R'$ when $g_{R}=g_{R'}$.} 
\label{fig_deform_intuition}
\end{figure}

With the above observation on geometric shapes of logical operators and deformations in mind, let us describe a topological property of geometric shapes of logical operators. The following lemma summarizes the topological property of logical operators in two-dimensional STS models.

\begin{lemma}[Topological shrinkage]\label{lemma_2dim}
In two-dimensional STS models, the following equations hold:
\begin{align}
g_{\overline{R_{0}}} = g_{R_{1}}, \quad g_{\overline{R_{1}}} = g_{R_{0}}, \quad g_{Q(1,0)} = g_{\overline{Q(1,0)}} =  k, \quad g_{Q(0,1)} = g_{\overline{Q(0,1)}} = k.
\end{align}
\end{lemma}

In other words, one can shrink geometric shapes of logical operators by applying some appropriate stabilizers in the following ways:
\begin{align}
\overline{R_{0}} \rightarrow R_{1}, \qquad \overline{R_{1}}\rightarrow \ R_{0}, \qquad \overline{Q(1,0)}\rightarrow \ Q(1,0), \qquad \overline{Q(0,1)} \rightarrow \ Q(0,1).
\end{align}
Note that these shrinkages preserve topological properties of geometric shapes of logical operators. 

\begin{proof}
Let us begin with the proof of $g_{\overline{R_{0}}} = g_{R_{1}}$. Let $k_{0}$ be the numbers of pairs of anti-commuting zero-dimensional and two-dimensional logical operators. Let $k_{1}$ be the number of pairs of anti-commuting one-dimensional logical operators. Then, we notice that $g_{R_{1}}=2k_{1}+k_{0}$ since $R_{1}$ supports both zero-dimensional and one-dimensional logical operators. On the other hand, since $g_{\overline{R_{0}}} = 2k - g_{R_{0}}$ from theorem~\ref{theorem_partition} and $g_{R_{0}}=k_{0}$, we have $g_{\overline{R_{0}}} = g_{R_{1}}$. If we use theorem~\ref{theorem_partition} to $g_{\overline{R_{0}}} = g_{R_{1}}$, we readily obtain $g_{R_{0}} = g_{\overline{R_{1}}}$ too. Next, let us show $g_{\overline{Q(1)}} = g_{Q(1)}$. Note that $g_{Q(1)}=k$. Then, we have $g_{\overline{Q(1)}}=k$, and have $g_{\overline{Q(1)}} = g_{Q(1)}$.
\end{proof}

\textbf{Three-dimensions:} Let us continue our analysis on STS models for higher-dimensional cases ($D>2$). Below, we analyze the topological property of logical operators in three-dimensional STS models. We assume that zero-dimensional logical operators and one-dimensional logical operators in a three-dimensional STS model can be defined inside $P(1,1,1)$, $P(n_{1},1,1)$, $P(1,n_{2},1)$ and $P(1,1,n_{3})$. 

We begin by finding reference regions for $D$-dimensional systems. Reference regions for $D>2$ can be defined from topological unit regions in a way similar to two-dimensional cases. Let $\vec{d}$ be an arbitrary binary $D$ component vector $\vec{d} = (d_{1},\cdots,d_{D})$ with $d_{m}=0,1$. Then, topological unit regions are:
\begin{align}
Q(\vec{d}) \equiv P(\vec{x}), \qquad \mbox{where} \quad x_{m}=n_{m}^{d_{m}}.
\end{align}
For example, $Q(1,1,0)=P(n_{1},n_{2},1)$ for $D=3$. We denote the weight of the binary vector $\vec{d}$ as $w(\vec{d}) \equiv \sum_{m=1}^{D}d_{m}$, which represents the dimension of $Q(\vec{d})$. Then, concatenated unit regions are defined as follows:
\begin{align}
R_{m} \equiv \bigcup_{w(\vec{d})=m} Q(\vec{d}).
\end{align}
A set of reference regions can be obtained by considering all the possible unions of $Q(\vec{d})$, which is denoted as $R_{ref}$. One may notice that topological unit regions $Q(\vec{d})$ are like independent generators for $m$th homology group for a $D$-torus: $H_{m}(T^{D})$.

\begin{figure}[htb!]
\centering
\includegraphics[width=0.55\linewidth]{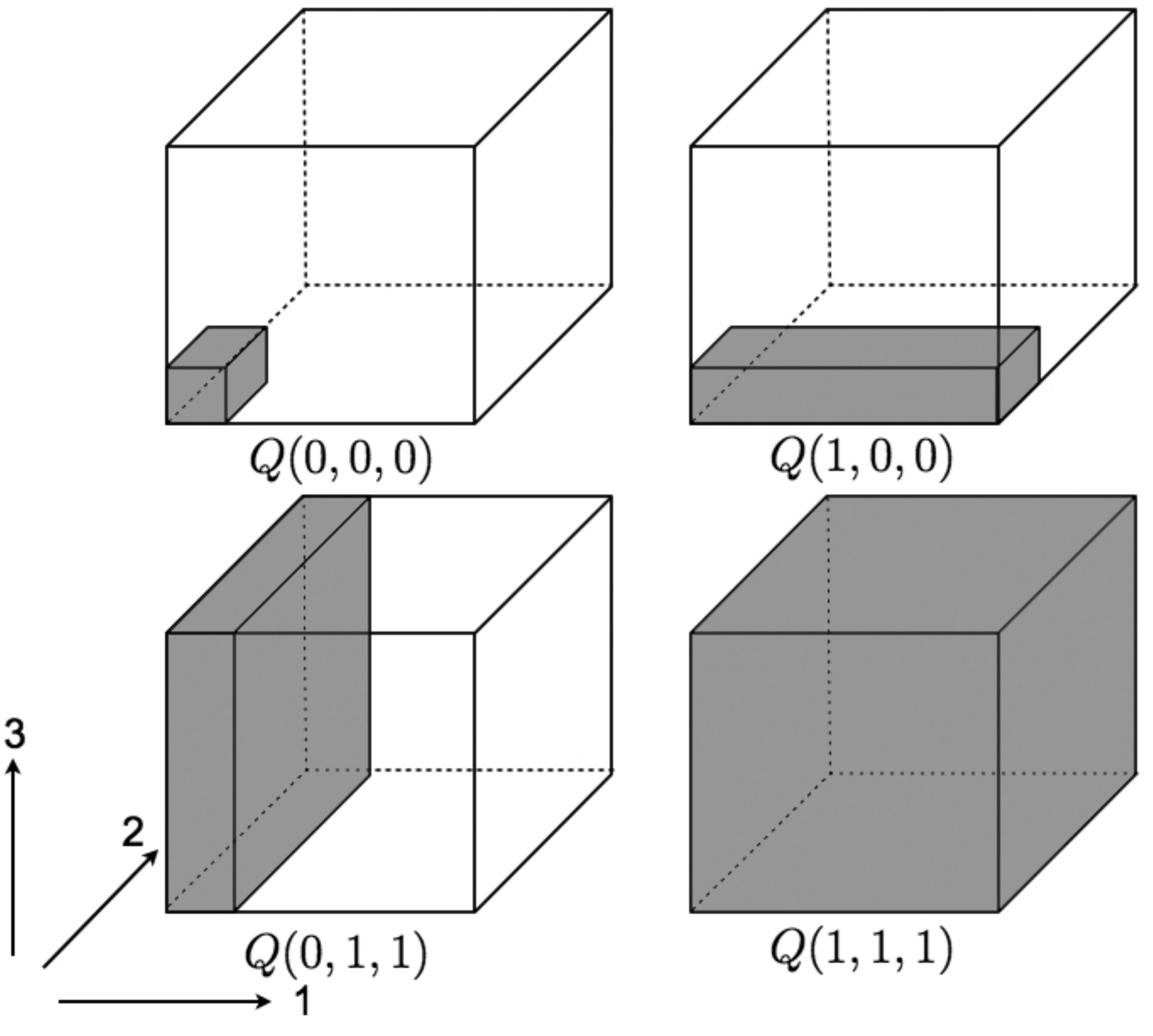}
\caption{Topological unit regions in a three-dimensional system. Recall that we set the periodic boundary conditions.
} 
\label{fig_3D_unit}
\end{figure}

It is worth presenting some examples here. In a three-dimensional system ($D=3$), we have the following topological unit regions. 
\begin{equation}
\begin{split}
\mbox{0 dim:} \qquad &Q(0,0,0)\\
\mbox{1 dim:} \qquad &Q(1,0,0),\ Q(0,1,0),\ Q(0,0,1)\\
\mbox{2 dim:} \qquad &Q(1,1,0),\ Q(0,1,1),\ Q(1,0,1)\\
\mbox{3 dim:} \qquad &Q(1,1,1).
\end{split}
\end{equation}
Some examples are shown in Fig.~\ref{fig_3D_unit}. Also, concatenated topological unit regions are:
\begin{equation}
\begin{split}
R_{0} &= Q(0,0,0)\\
R_{1} &= Q(1,0,0) \cup Q(0,1,0) \cup Q(0,0,1)\\ 
R_{2} &= Q(1,1,0)\cup Q(0,1,1)\cup Q(1,0,1)\\
R_{3} &= Q(1,1,1)
\end{split}
\end{equation}
which are described in Fig.~\ref{fig_3D_R}.

\begin{figure}[htb!]
\centering
\includegraphics[width=0.45\linewidth]{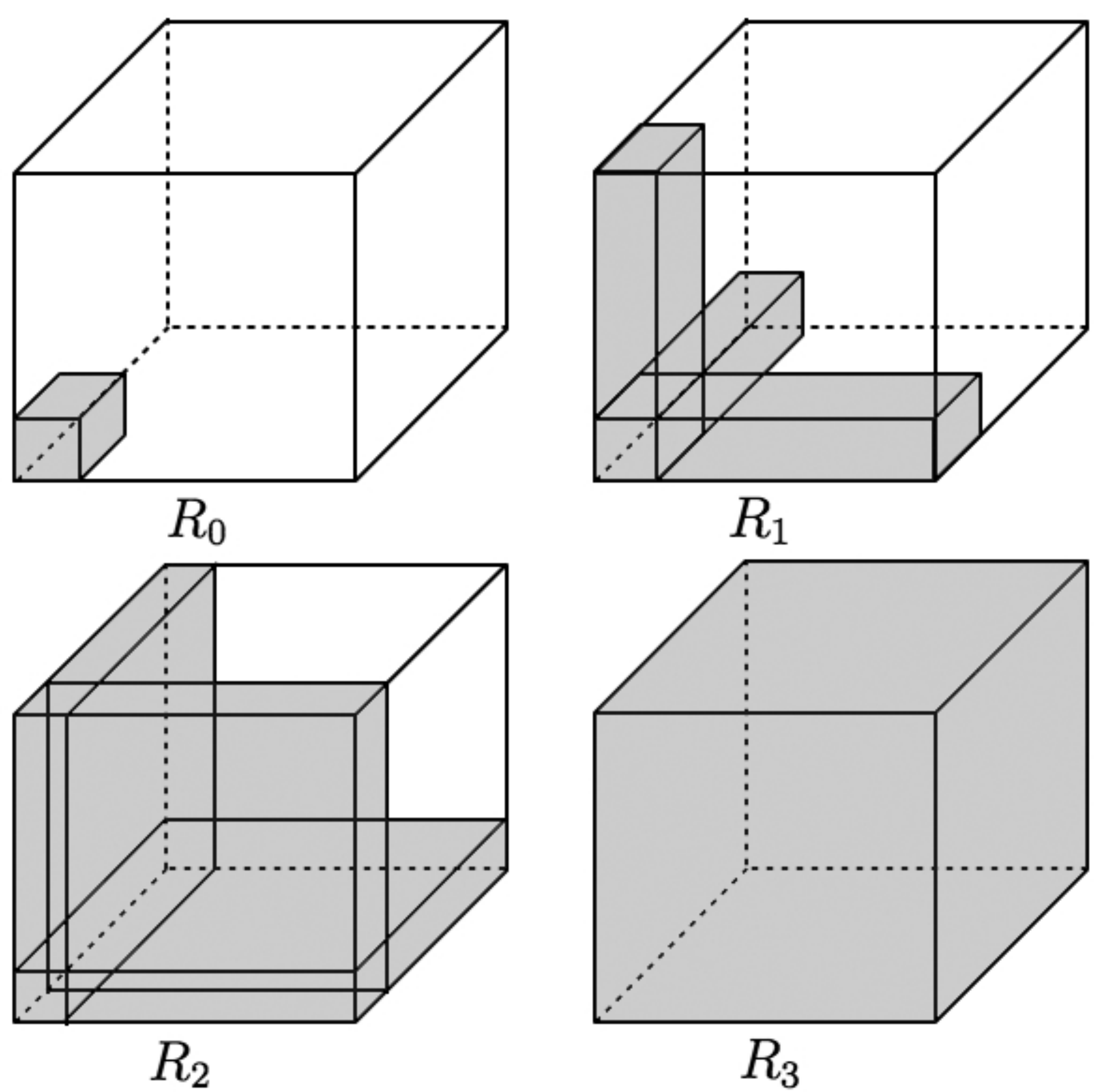}
\caption{Concatenated unit regions in a three-dimensional system. 
} 
\label{fig_3D_R}
\end{figure}

One can introduce the equivalence relations in terms of reference regions. Equivalence relations among them are shown as follows:
\begin{equation}
\begin{split}\label{eq:list}
R_{0} &\simeq \overline{R_{2}} \\
Q(1,0,0) &\simeq \overline{Q(1,1,0)\cup Q(1,0,1)}\\
Q(0,1,0) &\simeq \overline{Q(1,1,0)\cup Q(0,1,1)}\\
Q(0,0,1) &\simeq \overline{Q(1,0,1)\cup Q(0,1,1)}\\
Q(1,0,0) \cup Q(0,1,0)&\simeq \overline{Q(1,1,0)\cup Q(0,0,1)}\\
Q(0,1,0) \cup Q(0,0,1)&\simeq \overline{Q(0,1,1)\cup Q(1,0,0)}\\
Q(0,0,1) \cup Q(1,0,0)&\simeq \overline{Q(1,0,1)\cup Q(0,1,0)}\\
R_{1}  &\simeq \overline{R_{1}}\\
Q(1,1,0)&\simeq \overline{Q(1,1,0)}\\
Q(0,1,1) &\simeq \overline{Q(0,1,1)}\\
Q(1,0,1)&\simeq \overline{Q(1,0,1)}.
\end{split}
\end{equation}
Then, we have the following theorem.

\begin{lemma}[Topological shrinkage]\label{theorem_Ddim}
For $D$-dimensional STS models ($D=1,2,3$), let $R$ and $R'$ be reference regions: $R,R' \in R_{ref}$. When $\overline{R'}\simeq R$, one can deform geometric shapes of logical operators continuously from $\overline{R'}$ to $R$:
\begin{align}
g_{R} = g_{\overline{R'}} \qquad \mbox{for}\quad R \simeq \overline{R'}.
\end{align}
\end{lemma}

\begin{proof}
The proof of the lemma is straightforward from theorem~\ref{theorem_partition} and theorem~\ref{theorem_3dim}, so we present a proof only for $R_{1} \simeq \overline{R_{1}}$. Since all the zero-dimensional and one-dimensional logical operators can be supported inside $R_{1}$, we have $g_{R_{1}}\geq k$. Similarly, we have $g_{\overline{R_{1}}}\geq k$. Then, since $g_{R_{1}} + g_{\overline{R_{1}}}=2k$, we have $g_{R_{1}} = g_{\overline{R_{1}}}=k$.
\end{proof}

\textbf{Topological deformations of logical operators:}
So far, our discussion on a topological property of logical operators has been limited only to regions generated by $m$-dimensional unit regions while some connected regions (such as examples presented in Fig.~\ref{topo_deform}) cannot be generated by taking unions of unit regions. Then, a naturally arising question is whether one can perform similar continuous deformations for arbitrary connected regions of composite particles or not. Here, we provide a complete description on topological properties of logical operators in STS models by extending the notion of continuous deformations to arbitrary connected regions.

It turns out that one can continuously deform a geometric shapes of a logical operator defined on any connected region of composite particles continuously for STS models with $D=1,2,3$. Here, we summarize topological properties of logical operators in STS models as follows\footnote{Mathematically inclined readers may want to have more precise definitions of geometric shapes and continuous deformations. Here, we make some comments on this issue to make our discussion more rigorous. For a $D$-torus with $n_{1}\times \cdots \times n_{D}$ composite particles, we split the entire system into $n_{1}\times \cdots \times n_{D}$ hypercubic regions where each hypercube contains each composite particle, and cover the entire system completely, but without any overlap. For a given set of composite particles $R$, we define its geometric shape as a union of hypercubes which contain all the composite particles inside $R$. After assigning geometric shapes to sets of composite particles in this way, one can define a continuous map between two sets $R$ and $R'$, and introduce an equivalence relation between them in a straightforward way.}:

\begin{theorem}[Continuous deformation]\label{theorem_continuous_deformation}
Consider STS models for $D=1,2,3$. Consider two connected regions of composite particles $R$ and $R'$ which are topologically equivalent: $R \ \simeq \ R'$. Then, we have
\begin{align}
g_{R} \ = \ g_{R'}  \qquad \mbox{for}\quad R \ \simeq \ R'.
\end{align}
Therefore, for any given logical operator $\ell$ defined inside $R$, one can always find an equivalent logical operator defined inside $R'$ as long as $R$ can be continuously deformed into $R'$. 
\end{theorem}

\begin{figure}[htb!]
\centering
\includegraphics[width=0.65\linewidth]{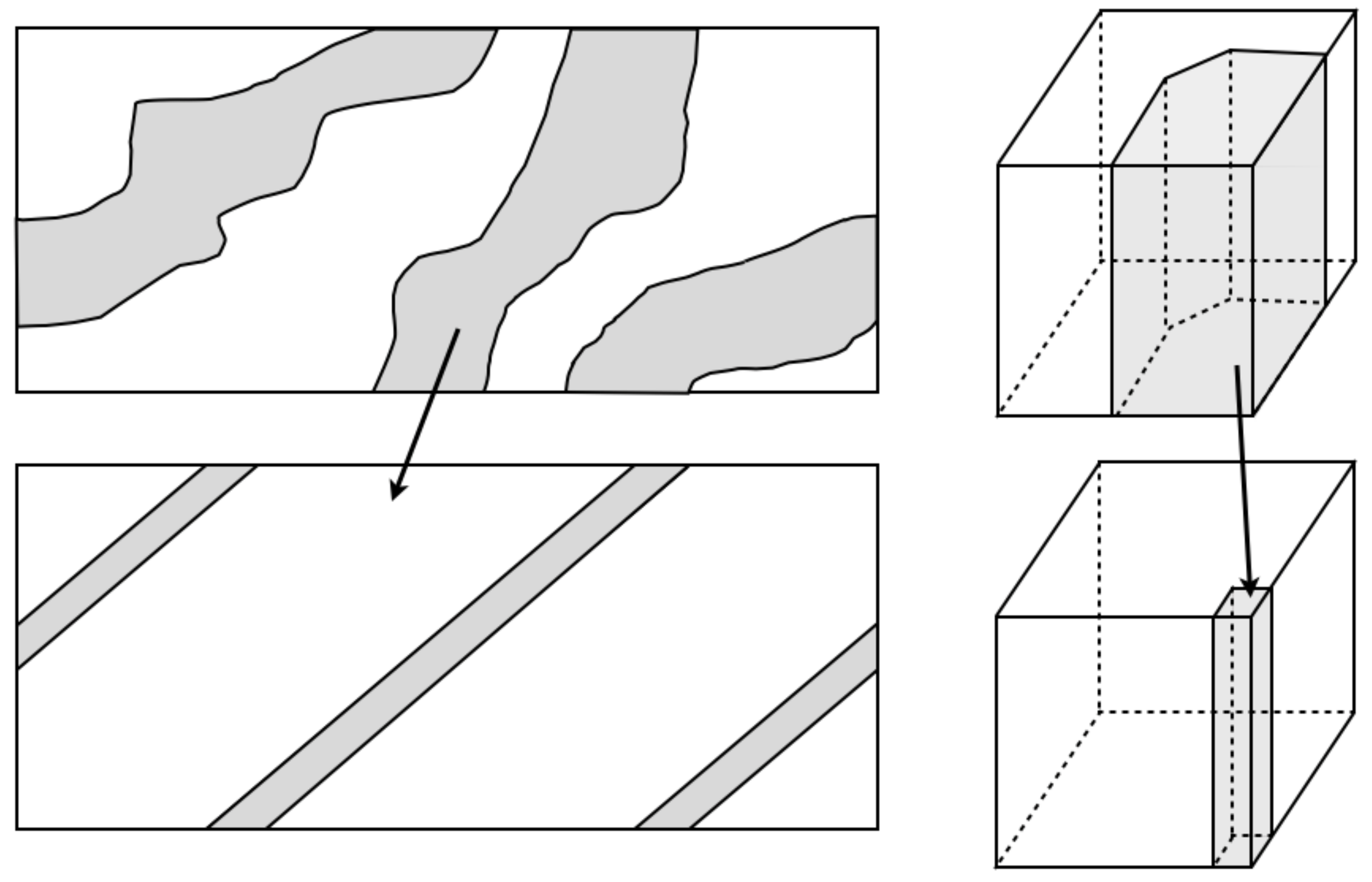}
\caption{Examples of continuous deformations. Note that systems have periodic boundary conditions.} 
\label{topo_deform}
\end{figure}

For clarity of presentation, we skip the proof of theorem~\ref{theorem_continuous_deformation}. One may verify it by using scale symmetries, a bi-partition theorem (theorem~\ref{theorem_partition}), and the translation equivalence of logical operators (theorem~\ref{theorem_TE}). We present some examples of continuous deformations in Fig.~\ref{topo_deform}.

\section{Summary and open questions}\label{sec:summary}

In this paper, we have established the connection between the feasibility of self-correcting quantum memory and the thermal stability of topological order, and provided partial answers to these two open problems by solving a model that may cover a large class of physically realizable quantum codes. Our discussion is limited to stabilizer codes with translation and scale symmetries, and these two problems still remain open for an even larger class of gapped spin systems on a lattice. Yet, we hope that our analysis will provide an important insight and a useful guidance on studies of coding and physical properties of gapped spin systems on a lattice. We also hope that our work will contribute to stimulating the use of quantum coding theoretical concepts in studying many-body quantum systems further. 

Below, we discuss possible future problems and make some comments on them.

\textbf{Frustration-free systems with scale symmetries:}
While our discussion is limited to stabilizer codes, any gapped spin systems with degenerate ground states can be used as quantum memory devices in principle. It seems that, for an arbitrary gapped spin system defined on a lattice, there exists some frustration-free Hamiltonian which approximates the original system and serve as its low energy effective Hamiltonian (except some subtle properties such as chirality). In fact, this claim has been rigorously proven for one-dimensional gapped spin systems~\cite{Hastings06} by showing that any ground state of one-dimensional gapped spin Hamiltonians can be efficiently simulated through the matrix product state formalism. Now, this claim seems to be widely believed among the condensed matter physics community, and it is at the heart of recent progress in classifications of quantum phases in gapped spin systems~\cite{Chen10, Beni10b, Chen11}. Therefore, the analyses on coding properties of arbitrary frustration-free Hamiltonians with translation and scale symmetries may provide a useful insight on questions concerning coding properties and quantum phases arising in arbitrary gapped spin systems on a lattice. 

To the best of our knowledge, all the examples of frustration-free Hamiltonians with translation and scale symmetries, such as the quantum double model~\cite{Kitaev03} and the string-net model~\cite{Levin05}, have the dimensional duality on geometric shapes of certain operators which may be considered as generalizations of logical operators. Also, these operators seem to have the continuous deformability in a way similar to logical operators in STS models. With these observations, we feel that our results are universal for all the frustration-free Hamiltonians with translation and scale symmetries, and thus, effectively true for arbitrary gapped spin systems defined on a lattice with a small number of ground states. Also, these observations imply that any gapped spin systems with scale symmetries (or a small number of ground states) may be described through TQFT, as discussed in~\ref{sec:topology}. Yet, the connection between TQFT and coding properties of gapped spin systems on a lattice with scale symmetries must be further established.  

\textbf{Beyond scale symmetries:}
While our treatments in the present paper are limited to quantum codes with scale symmetries, there are several interesting models of quantum codes which do not have scale symmetries~\cite{Chamon05, Newman99, Bravyi11b, Haah11}. Of particular interests are models proposed in~\cite{Newman99} and in~\cite{Haah11} which are classical and quantum memories respectively with partially self-correcting properties. These models do not have scale symmetries since the number of logical qubits $k_{\vec{n}}$ is highly sensitive to the system size $\vec{n}$, and there is no constant upper bound on $k_{\vec{n}}$. These models are known to have logarithmically large energy barrier $\Delta E \sim \mbox{LOG($L$)}$ with a large number of energy local minima. As a result, these models seem to have POLY($L$) relaxation time with slow dynamics, which may result in POLY($L$) bit or qubit storage time. 

Whether these models with broken scale symmetries are useful as classical and quantum memories is a complicated problem. First of all, for a system to be a efficient quantum memory device, it is desirable to have exponentially long qubit storage time: $\tau \sim \mbox{EXP($L$)}$, since it needs at least $d$ gate operations to write or readout a logical qubit, and these encoding and decoding processes takes at least polynomially long time $\tau \sim \mbox{POLY($L$)}$. Also, it seems difficult to find an efficient decoding algorithm with an efficient error-correction for these models. Finally, these models undergo phase transitions at $T=0$, which may imply their potential thermal instability. At this moment, there are many issues to be analyzed concerning implementability of quantum codes with broken scales symmetries from engineering viewpoints. 

\textbf{Beyond TQFT:}
Stabilizer codes with broken scale symmetries are remarkable examples which may be beyond the description of the standard TQFT. Such models do not have continuously deformable logical operators, and are not expected to be characterized by topological properties of logical operators as in STS models. This observation is consistent with the fact that systems described by TQFT, in a sense of the axiomatic formulation developed by Atiyah, are allowed to have only a finite number of degenerate ground states. Thus, finding an effective field theoretical description for stabilizer codes with broken scale symmetries may be an interesting future problem which may lead to discoveries of novel quantum phases that are currently unknown and are governed by some deeper mathematical formalism than topology.

To the best of our knowledge, all the examples of stabilizer codes with broken scale symmetries still possess certain discrete scale symmetries. STS models have continuous scale symmetries since $k_{\vec{n}}=k$ for all the system sizes $\vec{n}$, while stabilizer codes with broken scale symmetries seem to have discrete scale symmetries since the number of logical qubits $k_{\vec{n}}$ has good scaling properties under global scale transformations: $(n_{1},n_{2},\cdots)\rightarrow (cn_{1},cn_{2},\cdots)$ where $c$ is some appropriate integer with behaviors such as $k_{c\vec{n}}=c k_{\vec{n}}$. The distinction between continuous and discrete scale symmetries becomes particularly important when one considers effects of RG transformations. As is pointed out in~\cite{Wilson71}, systems with continuous scale symmetries correspond to fixed points of RG transformations, while systems with discrete scale symmetries correspond to limit cycles of RG transformations. Physical realizations of systems with discrete scale symmetries include the Effimov effect which is recently of particular interest in the ultracold atom physics community. This observation implies that stabilizer codes with discrete scale symmetries correspond to limit cycles of RG transformations. Yet, the connection between limit cycles and discrete scale symmetries for lattice systems must be further established since discussion in~\cite{Wilson71} is given primary for continuum systems. 

\textbf{Spin glass behaviors and translation symmetry breaking:}
Another interesting property of stabilizer codes with broken scale symmetries is a glassy behavior with slow relaxation dynamics. This glassy behavior may be understood as a direct consequence of broken scale symmetries, and associated broken translation symmetries in the ground space. 

In general, even if a Hamiltonian possesses translation symmetries, its ground states may not possess translation symmetries. Let us assume that a stabilizer Hamiltonian is translation invariant under unit translations of composite particles: $T_{m}(H)=H$ for all $m$. Then, ground states of this stabilizer Hamiltonian may break translation symmetries: $T_{m}(|\psi\rangle)\not=|\psi\rangle$. Some examples of systems with broken translation symmetries are presented in~\cite{Wen02, Kitaev06b, Beni10b}, and translation symmetry breaking of ground states arising in topologically ordered systems are studied in~\cite{Wen02, Kitaev06b}. 

An interesting connection between scale symmetries and translation symmetries is that when a stabilizer code has scale symmetries: $k_{\vec{n}}=k$, all the ground states are invariant under unit translations of composite particles: $T_{m}(|\psi\rangle)=|\psi\rangle$ regardless of the system size, as proven in~\cite{Beni10b}. In other words, if a system is coarse-grained such that scale symmetries are satisfied, translation symmetries are protected inside the ground space. On the other hand, if a system does not have scale symmetries, translation symmetries may be broken inside the ground space. In particular, when the system has a large number of logical qubits, translation symmetries are strongly broken such that there is no finite translation which keep ground states equivalent. 

This strong breaking of translation symmetries seems to be the key to the glassy behavior observed in stabilizer codes with broken scale symmetries. In conventional spin glass models, a Hamiltonian consists of mutually commuting terms whose signs are random, and translations symmetries of the Hamiltonian are initially broken due to the randomness. This randomness in the Hamiltonian gives rise to spin configurations of ground states which are not uniform over real space, and leads to a complicated and slow relaxation dynamics in spin glasses. On the other hand, in stabilizer codes with broken scale symmetries, stabilizer Hamiltonians possess translation symmetries initially. However, spin configurations are not uniform since translation symmetries are strongly broken inside the ground space due to broken scale symmetries. Thus, systems with broken scale symmetries are expected to exhibit spin glass behaviors with slow relaxation dynamics and a large number of local minima. Yet, the connection between the glassy behavior and broken scale symmetries must be further established since the observation given here is very heuristic.


\appendix

\section{Topological order at finite temperature}\label{sec:topo_ap}

In the main part of the paper, we have discussed the thermal stability of topological order by using expectation values of logical operators. In this appendix, we verify the use of expectation values of logical operators as topological order parameters, and define the thermal stability of topological order more rigorously. 

In~\ref{sec:topo_ap1}, we begin by verifying that expectation values of logical operators can be used for topological order parameter at $T=0$. In~\ref{sec:topo_ap2}, we define topological order at finite temperature by extending the definition of topological order at zero temperature to finite temperature. In~\ref{sec:topo_ap3}, we discuss whether the existence of a large energy barrier is sufficient for the thermal stability of the system or not.

\subsection{Topological order parameters and logical operators}\label{sec:topo_ap1}

While the ground state properties change only slightly under small perturbations in topologically ordered systems, the original ground state properties will be eventually lost under large perturbations. One may see whether topological order is lost by checking whether a new ground state can be approximated by the original ground state through a local unitary transformation. However, this naive approach will require a substantial amount of computations and cannot capture changes of coding and physical properties of the ground state under perturbations. 

Fortunately, the loss of the ground state properties under perturbations can be quantitatively characterized by some physical quantities, called topological order parameters. Here, we characterize the stability of topological order against perturbations through topological order parameters.

\textbf{Topological order parameter:}
Topological order parameters are physical quantities with which one can judge if topological order is lost under perturbations or not. In this light, one may deduce the necessary properties of topological order parameters. By its definition, a topological order parameter must be a global function of the system which does not change under local unitary transformations. Also, a topological order parameter must change only slightly when topological order is protected, while it undergoes some non-analytic change when topological order is lost as a result of perturbations.  

There have been several proposals for such topological order parameters. The number of ground states serves as a topological order parameter since the ground state degeneracy is protected for topologically ordered systems at the thermodynamic limit~\cite{Wen90}. Another interesting proposal is to use a certain entanglement measure, called topological entanglement entropy~\cite{Kitaev06, Levin06}, which corresponds to a constant correction to the entanglement area law. Since topological entanglement entropy is a non-local quantity which involves a large number of spins, it does not change under local unitary transformations. 

Below, we analyze the stability of topological order thorough expectation values of logical operators in order to further build the connection between topological order and quantum codes. 

\textbf{Expectation values of logical operators:}
Let us first consider two-dimensional Toric code and denote a pair of anti-commuting logical operators as $\ell$ and $r$ where $\ell$ extends in the $\hat{1}$ direction and $r$ extends in the $\hat{2}$ direction. Then, one may use the expectation values of the summations of translations of logical operators as order parameters:
\begin{align}
U(\ell) \ = \ \frac{1}{L} \sum_{y} T_{2}^{y}(\ell), \qquad U(r) \ = \ \frac{1}{L} \sum_{x} T_{1}^{x}(r)
\end{align}
where $L$ is the linear length of the system. Here, we took the summation of logical operators in a symmetric way. 

Now, let us discuss if topological order is stable under a perturbation $V$ by considering the following perturbed Hamiltonian with an initial bias:
\begin{align}
H_{\ell}(\epsilon) \ = \ H_{STS} - \epsilon U(\ell) + V.
\end{align}
Note that, when $V=0$, the small bias breaks the ground state degeneracy and
\begin{align}
\langle U(\ell)\rangle_{\epsilon \rightarrow 0} \ = \ 1.
\end{align}
Recall that, in evaluating the expectation value of $U(\ell)$ in a presence of $V$, one must take the limit where $\epsilon$ goes to zero after taking the limit where $L$ goes to infinity. Since the ground space is separated from excited states by a finite energy gap, one may neglect the effect of excited states. So, let us examine the effect of $\epsilon U(\ell)$ and $V$ on the ground space. As a result of $\epsilon U(\ell)$, the ground state space is split into two sectors with $\ell = 1$ and $\ell = -1$ where the energy splitting between two sectors is $\epsilon$. On the other hand, the energy splitting induced by a perturbation $V$ is exponentially suppressed by a factor $\exp(- L / L_{0})$. Then, in considering the effect of $V$ on a sector with $\ell =1$, one may neglect the effect of a sector with $\ell = -1$ when $\epsilon$ is $O(1)$. Now, we take the limit of $L \rightarrow \infty$. Then, the ground space corresponds to a sector with $\ell=1$, and we have 
\begin{align}
\langle U(\ell)\rangle_{\epsilon \rightarrow 0} \ > \ 0
\end{align}
even when $\epsilon$ goes to zero. However, if $V$ is large enough, the ground state properties will be lost and $\langle U(\ell)\rangle_{\epsilon \rightarrow 0}$ will be close to the value for the ground state of $V$. 
In a similar way, one may consider the following initial bias:
\begin{align}
H_{r}(\epsilon) \ = \ H_{STS} - \epsilon U(r) + V
\end{align}
and have 
\begin{align}
\langle U(r)\rangle_{\epsilon \rightarrow 0} \ > \ 0.
\end{align}
Therefore, the stability of topological order can be characterized by expectation values of $U(\ell)$ and $U(r)$.

This argument implies that the loss of the ground state properties may be captured through the stability of logical qubits encoded in the ground space. The initial bias of $\epsilon U(\ell)$ represents the initial encoding of a logical qubit. The expectation value of $U(\ell)$ under a perturbation $V$ represents how much information is protected in a presence of perturbations. 

\textbf{Stability of topological order:}
Based on these observations, let us formulate the stability of topological order in terms of logical operators. Let us first generalize the definition of logical operators. We call operators $U$ which satisfy the following conditions \emph{generalized logical operators}:
\begin{align}
&[H_{0},U] \ = \ 0 \\ 
&U \ \not= \ e^{i\theta}I \quad \mbox{and} \quad |U|_{\mbox{max}} \ = \ 1 \qquad \mbox{in the ground space}
\end{align}
where the action of $U$ is defined inside the ground space of the original unperturbed Hamiltonian $H_{0}$. In other words, generalized logical operators are any operators which does not change the energy of the system, but acts non-trivially inside the ground space. Generalized logical operators are not $c$-numbers inside the ground space. For example, a scaled summation of translations of logical operators is a generalized logical operator, while projectors onto the ground space are not generalized logical operators.

Based on generalized logical operators, we define the stability of topological order in the following way.

\begin{definition}[Stability of topological order]\label{def:order}
The system has topological order when there exists a pair of generalized logical operators $U(\ell)$ and $U(r)$ which satisfy the following conditions:
\begin{itemize}
\item $U(\ell)$ and $U(r)$ do not commute with each other: $[U(\ell),U(r)]\not =0$ inside the ground space. This implies the existence of the ground state degeneracy, and $U(\ell)$ and $U(r)$ may characterize a logical qubit or qudit.
\item For perturbed Hamiltonians:
\begin{align}
H_{\ell} \ = \ H_{0} - \epsilon U(\ell) + V, \qquad H_{r} \ = \ H_{0} - \epsilon U(r) + V
\end{align}
with any types of local perturbations $V$,
\begin{align}
\langle U(\ell) \rangle_{\epsilon \rightarrow 0, \ V \rightarrow 0} \ &= \ 1 \qquad \mbox{for $H_{\ell}$} \\
\langle U(r) \rangle_{\epsilon \rightarrow 0, \ V \rightarrow 0} \ &= \ 1    \qquad \mbox{for $H_{r}$}.
\end{align}
\end{itemize}
\end{definition}

In evaluating $\langle U(\ell) \rangle_{\epsilon \rightarrow 0, V \rightarrow 0}$, we first take the thermodynamic limit, and then, take the limit of $V \rightarrow 0$, and finally, take the limit of $\epsilon \rightarrow 0$. It should be emphasized that we consider any types of perturbations $V$ in analyzing the stability of topological order\footnote{The definition used here is similar to the one used in~\cite{Nussinov09} except for taking the limit of $V \rightarrow 0$. This is because the definition in~\cite{Nussinov09} is applicable to systems with gapless energy spectrum, while our definition is applicable only to systems with a finite energy gap.}. 

Let us look at some examples here. For a one-dimensional classical ferromagnet, let $\ell$ and $r$ be a zero-dimensional and a one-dimensional logical operators:
\begin{align}
\ell\ = \ Z_{1},\qquad r\ =\ \prod_{j}X_{j}.
\end{align}
Then, for $V = t \sum_{j}Z_{j}$, we have 
\begin{equation}
\begin{split}
H_{\ell} \ &= \ - \sum Z_{j}Z_{j+1}- \epsilon \Big(\frac{1}{L}\sum_{j}Z_{j}  \Big) + t \sum_{j} Z_{j} \\
           &= \ - \sum Z_{j}Z_{j+1} + \Big(t - \frac{\epsilon}{L} \Big) \sum_{j} Z_{j}.
\end{split}
\end{equation}
At the thermodynamic limit, we have 
\begin{align}
\langle U(\ell) \rangle \ &= \ -1     \qquad \mbox{for $H_{\ell}$}
\end{align}
regardless of $t>0$ and $\epsilon>0$. Then, taking the limit of $t \rightarrow 0$ and $\epsilon \rightarrow 0$, we have 
\begin{align}
\langle U(\ell) \rangle_{\epsilon \rightarrow 0, \ V \rightarrow 0} \ &= \ -1    \qquad \mbox{for $H_{\ell}$}.
\end{align}
Therefore, the system is not topologically ordered. 

\textbf{Topological order parameters and local unitary:} Strictly speaking, expectation values of logical operators cannot be used as topological order parameters. In fact, while logical operators are non-locally defined for topologically ordered systems, their expectation values may change under local unitary transformations. However, topological order must be characterized by some physical quantities or objects which are not affected by local unitary transformations. In fact, both the ground state degeneracy and topological entanglement entropy do not change under local unitary transformations. In this light, one might think that the above definition of the stability of topological order is not legitimate. Here, we make some comments on this issue briefly.

Despite the fact that expectation values of logical operators are not topological order parameters, we can formulate the stability of topological order through them. This is because we considered not only a specific type of local perturbations but all the types of local perturbations. Here, we demonstrate that, if we consider only one type of perturbations, expectation values of logical operators may fail to capture quantum phase transitions. In two-dimensional Toric code, let us consider a one-dimensional logical operator $\ell$ which consists only of Pauli $X$ operators. Then, consider the following perturbation:
\begin{align}
V \ = \ - t \sum_{r} X_{r},
\end{align}
where $t$ is some positive parameter and $r$ represents the position of each qubit. Then, when we increase $t$, there must be a quantum phase transition and the ground state properties will be close to the ones for $V$. However, such a perturbation may not change the value of $\langle U(\ell) \rangle$ since $V$ commutes with the logical operator $\ell$. In general, if one carefully choose the types of perturbations, expectation values of logical operators remain unchanged. This implies that if one tries to study a quantum phase transition only through $\langle U(\ell) \rangle$, one may fail to detect the transition.\footnote{While expectation values of logical operators cannot be used as topological order parameters, geometric shapes of dressed logical operators~\cite{Beni10b}, which are operators characterizing the transformations inside the ground space of the perturbed Hamiltonians~\cite{Hastings05}, may be used. Since continuous deformations between topologically distinct logical operators are not allowed, geometric shapes of dressed logical operators must undergo some discontinuous change during the transition. Thus, some non-analyticities will be induced in ground states during quantum phase transitions between Hamiltonians with topologically distinct logical operators.}

\subsection{Topological order at finite temperature}\label{sec:topo_ap2}

So far, we have focused on the ground state properties of topologically ordered systems at zero temperature and seen that topological order characterizes the ground state properties which are stable against any types of small perturbations. Here, we address the thermal stability of topological order.

The stability of topological order against perturbations can be discussed by seeing if the perturbed ground state can be reached by local unitary transformations or not. However, the stability of topological order against thermal fluctuations cannot be formulated in a similar way since there is no unitary transformation which connects a pure state and a statistical ensemble. Here, we define the stability of topological order against thermal fluctuations by the changes of topological order parameters. 

\textbf{Stability against thermal fluctuations:}
Below, we give the definition of the stability of topological order through expectation values of logical operators. 

\begin{definition}[Topological order at finite temperature]
The system is said to have topological order which is stable at finite temperature if and only if there exists a pair of non-commuting generalized logical operators $U(\ell)$ and $U(r)$, and some finite transition temperature $T_{c}$ which satisfy the following conditions:
\begin{itemize}
\item Consider perturbed Hamiltonians with initial biases
\begin{align}
H_{\ell} \ = \ H_{0} - \epsilon U(\ell) + V, \qquad H_{r} \ = \ H_{0} - \epsilon U(r) + V.
\end{align}
For $T=0$, 
\begin{align}
\langle U(\ell) \rangle_{\epsilon \rightarrow 0, \ V \rightarrow 0} \ &= \ 1 \qquad \mbox{for $H_{\ell}$} \\
\langle U(r) \rangle_{\epsilon \rightarrow 0, \ V \rightarrow 0} \ &= \ 1    \qquad \mbox{for $H_{r}$},
\end{align}
and for $T_{c}>T>0$, 
\begin{align}
\langle U(\ell) \rangle_{\epsilon \rightarrow 0, \ V \rightarrow 0} \ &> \ 0 \qquad \mbox{for $H_{\ell}$} \\
\langle U(r) \rangle_{\epsilon \rightarrow 0, \ V \rightarrow 0} \ &> \ 0    \qquad \mbox{for $H_{r}$}
\end{align}
for any types of $V$.
\end{itemize}
\end{definition}

One may see the connection between self-correcting quantum memory and the stability of topological order. In particular, if topological order in a spin system is stable at finite temperature, such a system works as self-correcting quantum memory. A logical qubit encoded with respect to $U(\ell)$ and $U(r)$ can be read out by measuring $U(\ell)$ and $U(r)$ even at finite temperature. This implies that encoded logical qubits are not lost in a presence of the interaction with the external environment, and there must be some self-correcting thermal dissipation processes. 

It should be noted that the expectation values of logical operators $\ell$ and $r$ usually vanishes at any finite temperature as one may see from a direct calculation. Only the symmetric summations of logical operators may give rise to non-vanishing expectation values of generalized logical operators. 

\textbf{Comment on the definition of stability:}
Our definition of the stability of topological order relies on topological order parameters associated with a pair of non-commuting logical operators $\ell$ and $r$. The definition of the stability of topological order in the present paper is motivated purely from quantum information theoretical viewpoints. However, the definition we have used in the present paper may be too strict. For example, if one is interested only in the ground state properties associated to the expectation value of $\ell$, it makes perfect sense to claim that a three-dimensional STS model has topological order which is stable at finite temperature. In fact, once the system is held at finite temperature, the system properties are stable against any types of small perturbations. 

Similarly, the stability of the ground state properties depends crucially on types of perturbations. For example, while there is no topological order in a one-dimensional system in a strict sense, the ground state properties of a one-dimensional classical ferromagnet are stale against small perturbations if one declares to consider only perturbations which are products of $X$ operators. Similarly, the ground state properties of the one-dimensional AKLT model is stable against small perturbations which do not break the time-reversal symmetry. Therefore, in discussing the stability of topological order, one needs to specify types of symmetries which are of interest. 

\subsection{Thermal stability of topological order and energy barrier}\label{sec:topo_ap3}

We have seen that a large energy barrier in classical or quantum memory is the key to the thermal stability of ferromagnetic order or topological order at finite temperature. Yet, it is not clear if the existence of a large energy barrier is sufficient for the thermal stability or not. 

While STS models always have energy barrier $O(L^{a})$ where $a$ is an integer ($a \geq 0$) due to the dimensional duality of logical operators, there are several examples of stabilizer codes with broken scale symmetries which have logarithmic energy barrier: $\Delta E \sim \mbox{LOG($L$)}$. Let us look at an example of classical memory with broken scale symmetries~\cite{Newman99}. The model is constructed on a square lattice with $L\times L$ qubits:
\begin{align}
H \ = \ - \sum_{i,j}Z_{i,j}Z_{i+1,j}Z_{i,j+1}.
\end{align}
The model does not have scale symmetries, and there is no upper bound on the number of logical qubits (bits) $k$ since $k$ is highly sensitive to $L$. The model has logarithmic energy barrier $\Delta E \sim \mbox{LOG($L$)}$, and as a result, the bit memory time is expected to be $\tau \sim \mbox{POLY($L$)}$ if one trusts the Arrhenius law. 

Despite a large energy barrier which scales with the system size, the model is known to be thermally unstable. This may be easily verified by computing the partition function: $Z(\beta)=\text{Tr}e^{-\beta H}$. Since $N-k$ stabilizers are independent, we have
\begin{align}
(e^{-\beta} + e^{\beta})^{N-k}(e^{-\beta})^{k} \ \leq \ Z(\beta) \ \leq \ (e^{-\beta} + e^{\beta})^{N-k} (e^{\beta})^{k}
\end{align}
and, we have
\begin{align}
\lim_{N \rightarrow \infty} \frac{1}{N} \log Z(\beta)\ = \ \log (e^{-\beta} + e^{\beta}).
\end{align}
Therefore, the thermodynamic property of this model is equivalent to a single qubit in a magnetic field, and the model does not have the thermal stability. A similar discussion holds for stabilizer codes with $N$ qubits, supported by $N$ interaction terms, which include one-dimensional Ising model, two-dimensional Toric code, models proposed in~\cite{Chamon05, Haah11}. Note that $M > N$ in two-dimensional Ising model and four-dimensional Toric code where $M$ is the number of interaction terms. 

It is not clear whether polynomial energy barrier is necessary for thermal stability of ferromagnetic order or not. Another important characteristic of models with logarithmic energy barrier is the existence of a logarithmically large number of local minima. As a result, at finite temperature, the entropic term easily dominates the free energy function. So, it seems that the thermal instability results from both the logarithmic energy barrier and the logarithmically large number of local minima.

It is worth mentioning that there is an interacting spin model with a logarithmic energy barrier, but is not defined on a lattice. The model is called two-dimensional XY model, which is known to undergo a thermal phase transition at finite temperature. The model is of particular interest since it was not expected to possess any thermal phase transitions as a result of the Mermin-Wagner theorem which states that continuous symmetries cannot be spontaneously broken at finite temperature for $D \leq 2$. The reason why two-dimensional XY model undergoes a thermal phase transition at finite temperature is because the transition, known as Kosterlitz-Thouless transition, is induced by topological defects which are not directly related to continuous symmetries of the model. In two-dimensional XY model, it takes logarithmic energy to create a vortex: $\Delta E \sim \log r$ where $r$ is the size of a vortex. As a result, at low temperature, a configuration without vortices is favored, while at high temperature, a configuration with a large number of vortices is favored. Between these configurations, a phase transition occurs. 

From the observations above, it seems that the logarithmic energy barrier may lead to the thermal stability, while the existence of local minima may break the thermal stability. Yet, at this moment, the connection between the energy barrier and the thermal stability has not been completely established.  

\section{Topology and quantum codes}\label{sec:topology}

In section~\ref{sec:physics}, we have shown that one can deform geometric shapes of logical operators continuously while keeping them equivalent. This continuous deformability of logical operators implies that one can introduce the notion of topology in analyzing and classifying coding and physical properties of STS models. In this appendix, we further discuss the role of topology in analyzing stabilizer codes. 

Many interesting properties of STS models, such as the dimensional duality of logical operators, naturally appear as corollaries of the continuous deformability of logical operators. In~\ref{sec:topology1}, we begin by showing that the dimensional duality of logical operators can be derived only by assuming the continuous deformability of logical operators. In particular, we show that, for stabilizer codes with continuously deformable logical operators, $m$-dimensional and $D-m$-dimensional logical operators always form anti-commuting pairs where $m$ is an arbitrary positive integer ($m\leq D$). As an example of $D$-dimensional stabilizer codes with continuously deformable logical operators, we present generalizations of the Toric code to $D$-dimensional systems with anti-commuting pairs of $m$-dimensional and $D-m$-dimensional logical operators for arbitrary positive integers $D$ and $m$.

The topological property of logical operators naturally leads us to consider a possible relevance to another well-celebrated theoretical framework equipped with the notion of topology; called topological quantum field theory (TQFT)~\cite{Witten89, Birmingham91, Nayak08}. In~\ref{sec:topology2}, we establish the connection between quantum codes with the continuous deformability and TQFT further by demonstrating that the braiding of anyonic excitations in a $D$-dimensional stabilizer code is characterized by a topological invariant in a $D+1$-dimensional system. This implies that such a $D$-dimensional stabilizer code can be effectively described by $D+1$-dimensional TQFT. 

\subsection{Dimensional duality as a corollary of continuous deformability}\label{sec:topology1}

The continuous deformability of logical operators implies that topology is the essential notion in analyzing coding and physical properties of STS models. Then, a naturally arising question concerns the role of topology in determining coding and physical properties of stabilizer codes. Here,  we show that the dimensional duality of logical operators is a universal property for all the stabilizer codes with continuously deformable logical operators. We also give concrete examples of $D$-dimensional stabilizer codes which have continuous deformability of logical operators and the dimensional duality by generalizing the Toric code to $D$-dimensional systems.

\textbf{Dimensional duality from topological deformation:}
Let us begin by counting the number of independent logical operators defined inside $m$-dimensional regions. Recall that $m$-dimensional concatenated unit regions are obtained by taking unions of all the $m$-dimensional topological unit regions. Then, one may call logical operators which can be defined inside $R_{m}$, but cannot be defined inside $R_{m-1}$, \emph{$m$-dimensional logical operators}. 

\begin{definition}
$m$-dimensional logical operators are logical operators which have representations supported inside $R_{m}$, but do not have representations supported inside $R_{m-1}$.
\end{definition}

Now, let us denote the number of independent $m$-dimensional logical operators as $g_{m} \equiv g_{R_{m}} - g_{R_{m-1}}$ where $g_{0} \equiv g_{R_{0}}$ by setting $g_{R_{-1}}\equiv0$. Then, there exists an interesting relation among the numbers of $m$-dimensional logical operators. In particular, the following lemma holds.

\begin{lemma}\label{lemma_duality}
There are the same number of $m$-dimensional and $D-m$-dimensional logical operators:
\begin{align}
g_{m} = g_{D-m} \qquad \mbox{for} \quad m = 0,\cdots,D.
\end{align}
\end{lemma}

The proof of this lemma can be obtained through a simple algebra by combining theorem~\ref{theorem_partition} and the deformability of logical operators.

\begin{proof}
Consider a bi-partition of the entire system into $R_{m}$ and $\overline{R_{m}}$. From the topological deformation of logical operators, we have 
\begin{align}
g_{\overline{R_{m}}} = g_{R_{D-m-1}}
\end{align}
since $\overline{R_{m}} \simeq R_{D-m-1}$. Thus, $R_{m}$ and $\overline{R_{m}}$ support the following logical operators:
\begin{align}
R_{m}            &: \quad \mbox{$0$-dim}, \ \mbox{$1$-dim}, \ \cdots, \ \mbox{$m$-dim} \notag \\
\overline{R_{m}} &: \quad \mbox{$0$-dim}, \ \mbox{$1$-dim}, \ \cdots, \ \mbox{$D-m-1$-dim}. \notag
\end{align}
Therefore, we have
\begin{align}
g_{R_{m}} = \sum_{j=0}^{m} g_{j}, \qquad g_{\overline{R_{m}}} = \sum_{j=0}^{D-m-1}g_{j}.
\end{align}

Recall that $g_{R}+g_{\bar{R}}=2k$ as presented in theorem~\ref{theorem_partition}. Using this formula for $R = R_{m}$, we have
\begin{align}
g_{R_{m}} + g_{\overline{R_{m}}} = 2k.
\end{align}
Then, we have
\begin{align}
\sum_{j=0}^{m} g_{j} + \sum_{j=0}^{D-m-1}g_{j} = 2k, \qquad \mbox{for} \quad m = 0,\cdots,D.
\end{align}
Since the total number of independent logical operators is $\sum_{j=0}^{D}g_{j}=2k$, we have $g_{m} = g_{D-m}$ for all $m$.
\end{proof}

This lemma implies the existence of a dimensional duality in geometric shapes of logical operators. One may notice that this lemma is just a manifestation of Poincar\'{e} duality in a $D$-torus where the $m$th and $D-m$th Betti numbers are equal. 

To completely establish relations between each logical operator with different dimensions, let us analyze their commutation relations. We have the following theorem.

\begin{theorem}[Dimensional duality]\label{theorem_duality}
One can choose a set of $2k$ independent logical operators of $D$-dimensional systems with the deformability of logical operators in the following way:
\begin{align}
\left\{
\begin{array}{cccccc}
 \ell_{1}, & \cdots , &  \ell_{k}        \\
  r_{1}, & \cdots , &  r_{k} 
\end{array}
 \right\}.
\end{align}
where $\ell_{p}$ are $m_{p}$-dimensional logical operators and $r_{p}$ are $D-m_{p}$-dimensional logical operators for some integer $m_{p}$ ($0 \leq m_{p} \leq D$) for any $p=1,\cdots,k$.
\end{theorem}

In other words, one can choose logical operators such that the summation of dimensions of pairs of anti-commuting logical operators is always $D$. Theorem~\ref{theorem_duality} follows immediately from the following lemma.

\begin{lemma}
$m$-dimensional and $m'$-dimensional logical operators commute with each other if $m + m' <D$.
\end{lemma}

\begin{proof}
Consider a $m$-dimensional logical operator $\ell$ and a $m'$-dimensional logical operator $\ell'$ which is defined inside $R_{m}$ and $R_{m'}$ respectively. For $m + m'<D$, there exists a translation of $R_{m}$ such that $R_{m}$ and $R_{m'}$ have no overlap. Then, due to the translation equivalence of logical operators, some translation of $\ell$ do not have an overlap with $\ell'$, which leads to $[\ell,\ell']=0$.
\end{proof}

With this lemma, the proof of theorem~\ref{theorem_duality} is immediate by using lemma~\ref{lemma_duality}. For example, from the lemma, zero-dimensional logical operators may anti-commute only with $D$-dimensional logical operators. Since there are the same number of zero-dimensional and $D$-dimensional logical operators, there exists a canonical set of logical operators where $D$-dimensional logical operators can anti-commutes only with zero-dimensional logical operators. Similarly, one can show that there exists a set of $2k$ independent logical operators such that $m$-dimensional logical operators anti-commute only with $D-m$-dimensional logical operators for all $m$. 

\textbf{Generalization of the Toric code:}
Now, we give concrete examples of $D$-dimensional stabilizer codes which have continuously deformable logical operators. The model we present here is a straightforward generalization of two-dimensional Toric code to $D$-dimensional settings for arbitrary positive integer $D$. In particular, we illustrate the construction of $D$-dimensional Toric code with anti-commuting pairs of $m$-dimensional and $D-m$ dimensional logical operators for arbitrary positive integers $D$ and $m$ ($m \leq D$). 

\begin{figure}[htb!]
\centering
\includegraphics[width=0.50\linewidth]{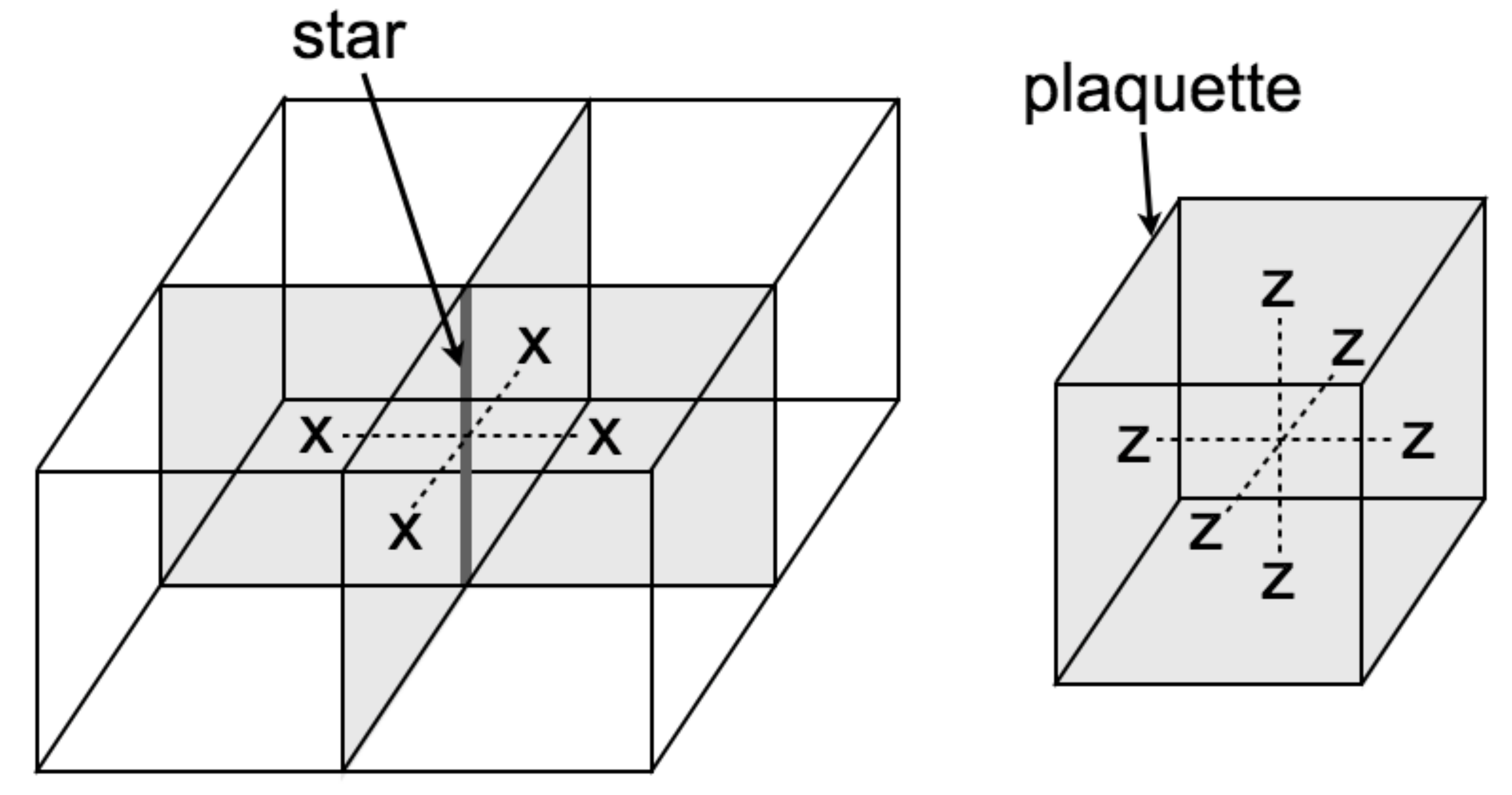}
\caption{An example of the construction for $D=3$ and $m=2$. 
} 
\label{D_dim}
\end{figure}

We first consider a $D$-dimensional hypercubic lattice which consists of $N=L\times \cdots \times L$ of $D$-dimensional unit hypercubes with periodic boundary conditions. We denote a set of $p$-dimensional unit hypercubes in this lattice as $h_{p}$. Note that $|h_{p}|  =  {}_D C_{p} \cdot N$
since one needs to specify $p$ directions from $D$ directions in defining $p$-dimensional unit hypercubes. We put qubits at the centers of $m$-dimensional hypercubes. Then, the total number of qubits is ${}_D C_{m} \cdot N$, and qubits are labeled by $m$-dimensional unit hypercubes in $h_{m}$. Fig.~\ref{D_dim} shows a construction of the model for $D=3$ and $m=2$ where qubits live at centers of two-dimensional unit squares. 

The Hamiltonian consists of plaquette terms and star terms as in the conventional two-dimensional Toric code:
\begin{align}
H \ = \  H_{plaquette} + H_{star}
\end{align}
Here, we call $(m+1)$-dimensional unit hypercubes in $h_{m+1}$ ``plaquettes''. Then, plaquette terms consist of Pauli $Z$ operators which act on qubits included in plaquettes:
\begin{align}
H_{plaquette} \ = \ - \sum_{p \in h_{m+1}} B_{p}, \qquad B_{p}\ = \ \prod_{r \subset p} Z_{r}
\end{align}
where $r$ represent $m$-dimensional unit hypercubes included inside a $m+1$-dimensional unit hypercube $p$. Note that there are ${}_D C_{m+1}\cdot N$ plaquette terms. Next, we call $(m-1)$-dimensional unit hypercubes in $h_{m-1}$ ``stars''. Then, star terms consist of Pauli $X$ operators which act on qubits neighboring to a star:
\begin{align}
H_{star} \ = \ - \sum_{s \in h_{m-1}}A_{s}, \qquad A_{s} \ = \  \prod_{s \subset r} X_{r}
\end{align}
where $s \subset r$ means that a star $s$ is included inside $m$-dimensional unit hypercube $r$. Note that there are ${}_D C_{m-1}\cdot N$ star terms. Noting that $[A_{s},B_{p}]=0$ since $A_{s}$ and $B_{p}$ share either zero or two qubits in common, the model is a stabilizer code. Fig.~\ref{D_dim} shows constructions of star terms and plaquette terms for $D=3$ and $m=2$.

One can verify that the model has $k = {}_D C_{m}$ logical qubits with $k$ anti-commuting pairs of $m$-dimensional and $(D-m)$-dimensional logical operators where a $m$-dimensional logical operator consists of Pauli $Z$ operators supported on a $m$-dimensional hyperplane, while a $(D-m)$-dimensional logical operator consists of Pauli $X$ operators supported on a $(D-m)$-dimensional hyperplane. Since they may share either zero or one qubit, they may commute or anti-commute with each other. One may easily see that the construction above reproduces two-dimensional, three-dimensional and four-dimensional Toric code for choices of $(D,m)=(2,1), (3,1), (4,2)$. For $m=0$ and $m=D$, the model is reduced to the $D$-dimensional Ising model. One can verify that logical operators arising in this model can be deformed continuously by using the bi-partition theorem (theorem~\ref{theorem_partition}). The construction above can be easily generalized to arbitrary $D$-dimensional graph embedded in a $D$-dimensional geometric manifold in a way similar to two-dimensional Toric code.

\subsection{$D$-dimensional STS model and $D+1$-dimensional TQFT}\label{sec:topology2}

The emergence of the notion of topology in geometric shapes of logical operators leads us to consider a possible relevance to topological quantum field theory (TQFT), which also deals with systems whose physical properties depend heavily on topological characteristics of the systems. Here, we make an attempt to establish the connection between stabilizer codes with continuously deformable logical operators and TQFT further.

Roughly speaking, TQFT is a field theory which is invariant under diffeomorphism (continuous deformation), and particularly suited for describing topologically ordered systems. The most important characteristic in systems described by TQFT is the invariance of all the correlation functions under diffeomorphism. Consider an arbitrary diffeomorphism through a continuous change of space-time coordinates $x \rightarrow x'$. Then, the correlation function of a scalar operator $\phi(x)$ satisfies
\begin{align}
\langle 0_{i} | \phi(x_{1}) \phi(x_{2})\cdots \phi(x_{n}) | 0_{j} \rangle = \langle 0_{i} | \phi(x_{1}') \phi(x_{2}')\cdots \phi(x_{n}') | 0_{j} \rangle
\end{align}
where $|0_{i}\rangle$ represent degenerate ground states. Therefore, the vacuum expectation value of any products of scalar operators is invariant under differmorphism, and only the topological properties of products may characterize their expectation values.

Below, in order to establish the connection between stabilizer codes and TQFT, we show that the braiding of anyonic excitations in $D$-dimensional stabilizer codes can be characterized by some topological invariant in a $D+1$-dimensional system. In particular, by characterizing propagations of anyonic excitations in a $D+1$-dimensional system, we demonstrate that the braiding of anyonic excitations corresponds to a configuration of $m$-dimensional and $D-m$-dimensional closed objects with a non-zero linking number. 

We begin by analyzing propagations of anyonic excitations in two-dimensional Toric code. Anyonic excitations can be created by applying a segment of a one-dimensional logical operator to a ground state of the Toric code Hamiltonian since endpoints of a segment of a logical operator $\ell^{seg}$ may not commute with interaction terms and create localized excitations (Fig.~\ref{fig_anyon1}(a)): $\ell^{seg}|\psi_{gs}\rangle$, and one may make anyons propagate along a geometric shape of a one-dimensional logical operator. Since one can deform a geometric shape of a one-dimensional logical operator continuously in the Toric code, anyons can propagate freely on the lattice by applying a segment of a deformed one-dimensional logical operator (Fig~\ref{fig_anyon1}(b)). Therefore, the continuous deformability of logical operators is the key to propagations of anyonic excitations. 

\begin{figure}[htb!]
\centering
\includegraphics[width=0.40\linewidth]{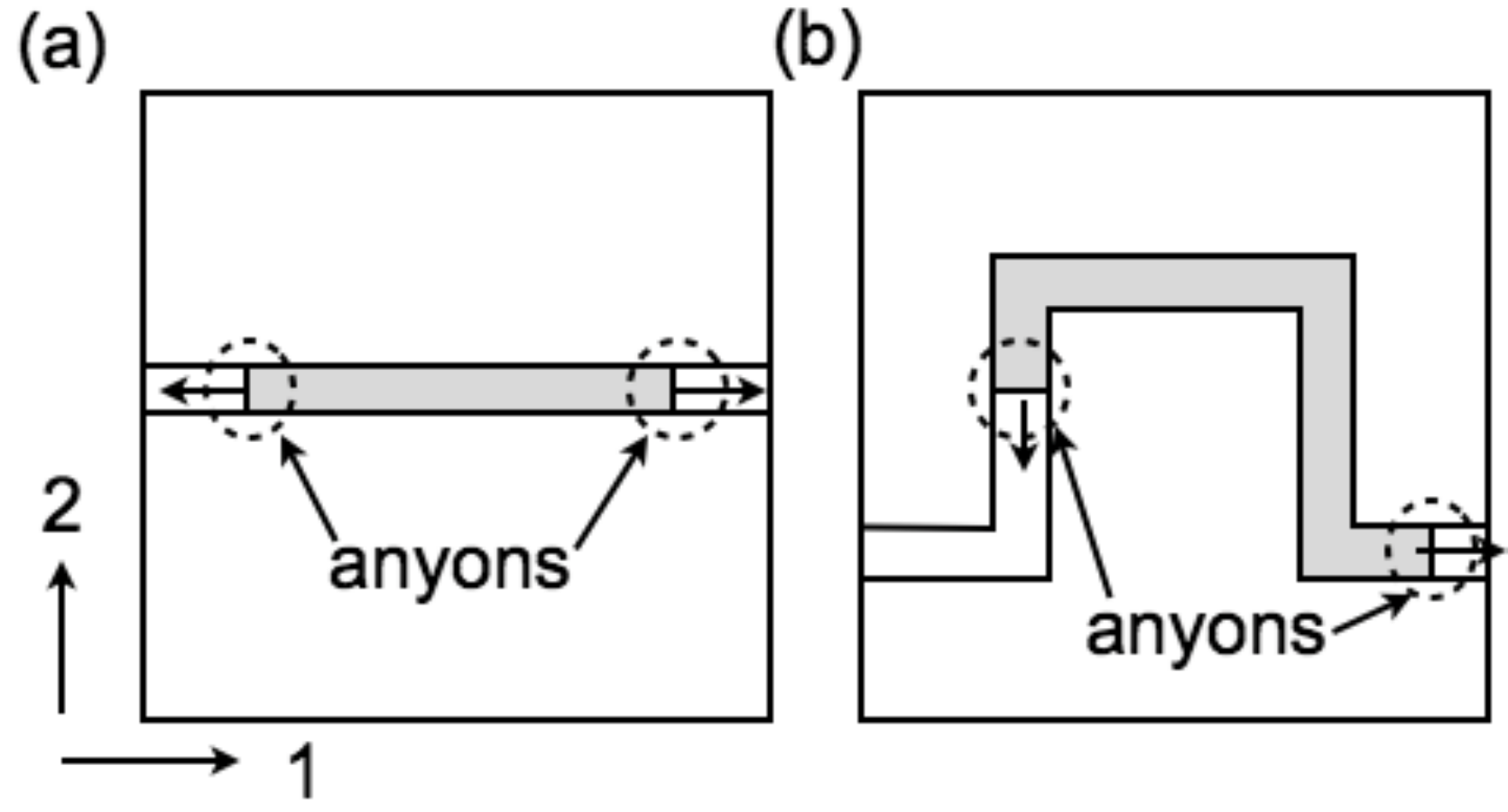}
\caption{Anyonic excitations created from segments of logical operators. (a) A one-dimensional logical operator. (b) A deformed one-dimensional logical operator. 
} 
\label{fig_anyon1}
\end{figure}

Propagations of anyonic excitations can be characterized by a one-dimensional closed loop drawn in a three-dimensional space. To illustrate this point, let us consider a process of creation, propagation, and annihilation of anyons, as described in Fig.~\ref{fig_anyon2}(a). By drawing propagations of anyonic excitations in a three-dimensional system by adding the time axis, the entire process can be represented as a one-dimensional closed loop as shown in Fig.~\ref{fig_anyon2}(b). In general, arbitrary one-dimensional closed loop in a three-dimensional system may characterize some propagations of anyonic excitations. See an example in Fig.~\ref{fig_anyon3} which involves creations of two pairs of anyonic excitations. 

\begin{figure}[htb!]
\centering
\includegraphics[width=0.60\linewidth]{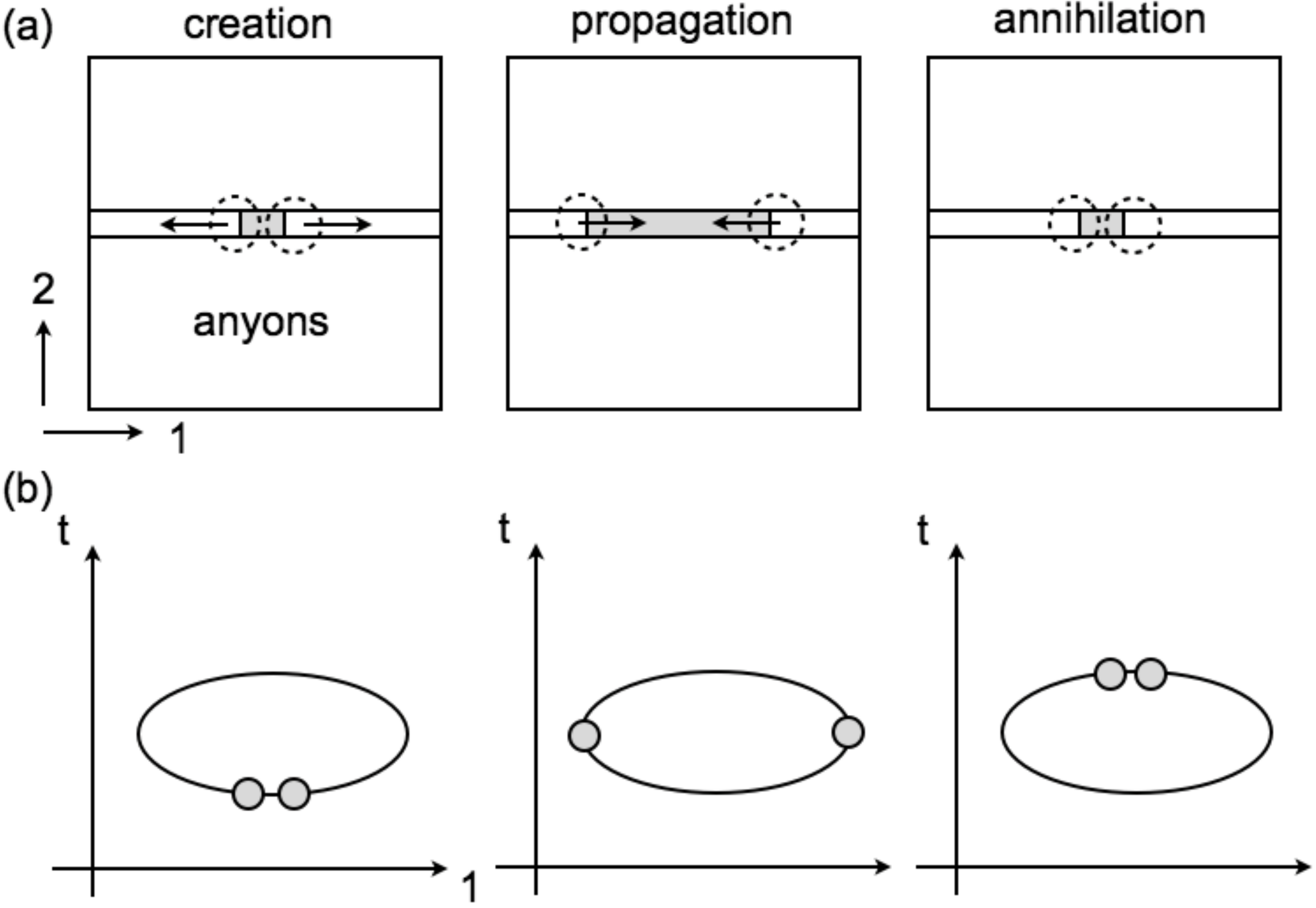}
\caption{The correspondence between anyonic excitations in two-dimensional Toric code and a closed loop in a three-dimensional system. (a) Creation, propagation and annihilation of anyonic excitations. (b) Anyonic excitations described in a three-dimensional system. 
} 
\label{fig_anyon2}
\end{figure}

\begin{figure}[htb!]
\centering
\includegraphics[width=0.35\linewidth]{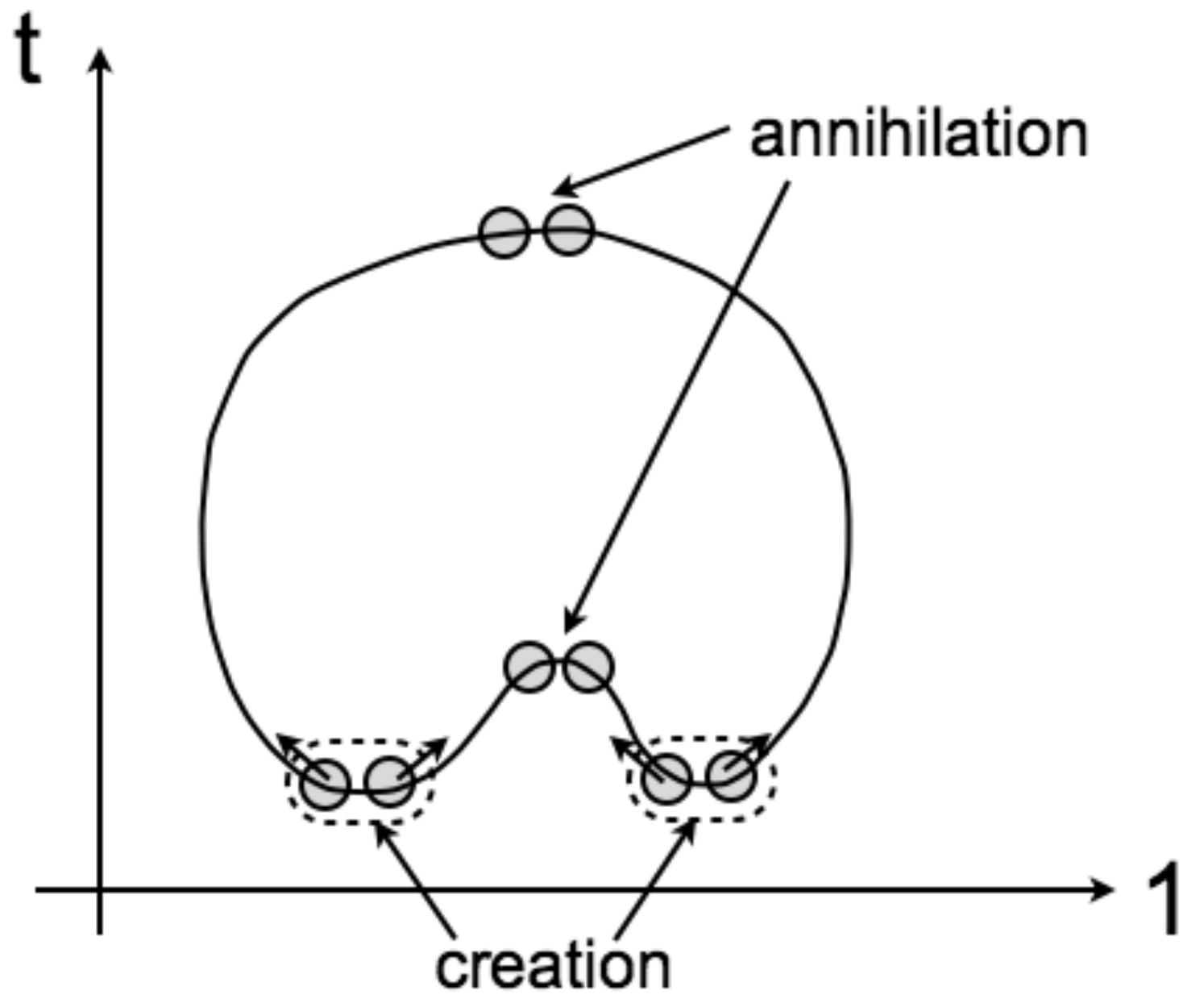}
\caption{Propagations of anyonic excitations and an associated closed loop in a three-dimensional system.
} 
\label{fig_anyon3}
\end{figure}

Next, let us consider the braiding between anyonic excitations in the Toric code. One interesting property of topologically ordered spin systems is the non-trivial braiding property between anyonic excitations. There are two distinct type of anyons in two-dimensional Toric code which are created by a pair of anti-commuting logical operators respectively, and the braiding of different types of anyonic excitations may give rise to non-trivial change inside the ground space. In two-dimensional Toric code, the braiding of different anyons give rise to an additional phase $-1$ to the original ground state $|\psi_{gs}\rangle \rightarrow - |\psi_{gs}\rangle$, which is a direct consequence of the anti-commutation between logical operators. 

One may understand this non-trivial braiding arising in two-dimensional Toric code as a topological invariant in a three-dimensional system. Let us consider the braiding of anyons described in Fig.~\ref{fig_anyon4}(a) where two types of pairs of anyonic excitations are created. One can characterize this braiding process as two closed loops in a three-dimensional system where loops are linked as described in Fig.~\ref{fig_anyon4}(b). In a more technical language, the braiding of anyons occurs only if the linking number between two loops has non-zero value. In particular, the final state is $|\psi_{gs}\rangle \rightarrow (-1)^{N_{link}} |\psi_{gs}\rangle$ where $N_{link}$ is the linking number of a given configuration of one-dimensional closed loops. This observation indicates that the braiding of anyonic excitations can be characterized by topological invariants, such as the linking number, in a ($2+1$)-dimensional system.

\begin{figure}[htb!]
\centering
\includegraphics[width=0.70\linewidth]{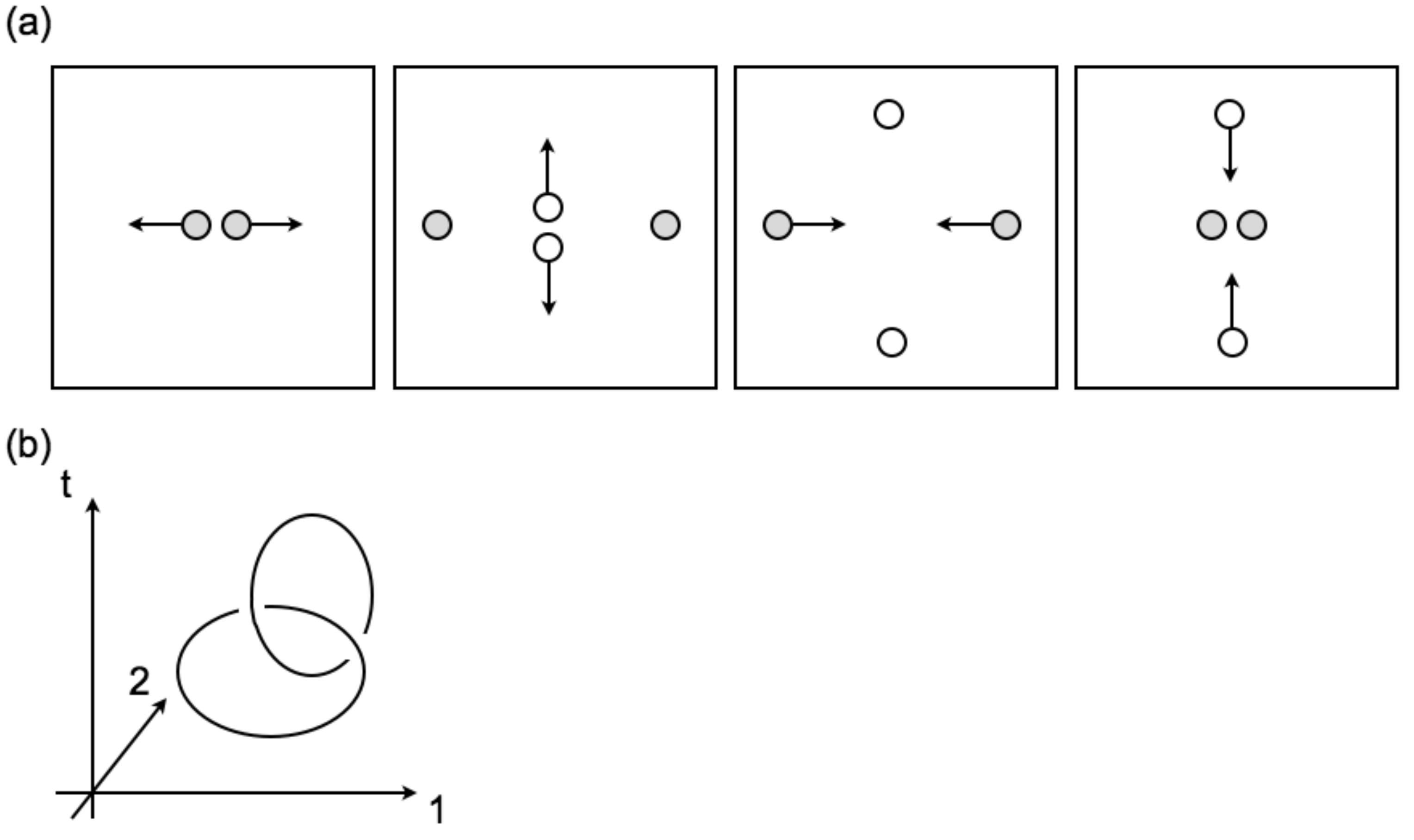}
\caption{The braiding as a topological invariant in a three-dimensional system. (a) A braiding of anyonic excitations. (b) Loops with non-zero linking number. 
} 
\label{fig_anyon4}
\end{figure}

A similar observation holds for $D$-dimensional Toric code with a pair of $m$-dimensional and $D-m$-dimensional logical operators. The first type of anyonic excitations can be created by a segment of $m$-dimensional logical operator, while the second type of anyonic excitations can be created by a segment of $D-m$-dimensional logical operator. The propagation of anyonic excitations of the first type can be characterized by a $m$-dimensional object in a $D+1$-dimensional systems, while the propagation of anyonic excitations of the second type can be characterized by a $D-m$-dimensional object. These anyonic excitations are braided when the linking number between $m$-dimensional and $D-m$-dimensional objects is non-zero. While we have discussed the braiding for $D$-dimensional Toric code, similar discussion holds for any stabilizer codes with continuously deformable logical operators. Therefore, one may expect that $D$-dimensional stabilizer codes with continuous deformability can be effectively described by $D+1$-dimensional TQFT.

\section{Decomposition of logical operators}\label{sec:decomposition}

We present the proof of theorem~\ref{theorem_3dim} in this and the next appendices. The goal of this appendix is to prove the following theorem which will be the key to the proof of theorem~\ref{theorem_3dim}. 

\begin{theorem}[Decomposition]\label{theorem_decomposition}
Consider a three-dimensional STS model with the system size $n_{1}= 2\cdot2^{2n_{2}v}!$, $n_{2}=2^{m}$ and arbitrary $n_{3}$ where $m$ is an arbitrary positive integer. For a given logical operator $\ell$ supported inside $P(n_{1},n_{2},1)$, one can decompose $\ell$ as a product of the following centralizer operators
\begin{align}
\ell \ \sim \ \ell_{a}\ell_{b}, \qquad \ell_{a},\ell_{b}\ \in \ \mathcal{C}_{P(n_{1},n_{2},1)}
\end{align}
where
\begin{align}
T_{1}^{\beta}(\ell_{b})\ = \ \ell_{b}, \qquad  \beta \ \leq \ 2^{2n_{2}v}
\end{align}
and $\ell_{a}$ is defined inside $P(2v, n_{2},1)$.
\end{theorem}

Here, $\mathcal{C}_{R}$ represents the restriction of the centralizer group $\mathcal{C}$ onto a region of composite particles $R$, meaning that $\mathcal{C}_{R}$ is a subgroup of centralizer operators defined inside $R$. Therefore, $\ell_{a},\ell_{b}\in  \mathcal{C}_{P(n_{1},n_{2},1)}$ means that $\ell_{a}$ and $\ell_{b}$ are centralizer operators defined inside $P(n_{1},n_{2},1)$. We show the claim of the theorem graphically in Fig.~\ref{fig_decomposition3D}. The theorem claims that a two-dimensional logical operator defined inside $P(n_{1},n_{2},1)$ can be decomposed as a product of a one-dimensional centralizer operator $\ell_{a}$ and a two-dimensional centralizer operator $\ell_{b}$ which is periodic in the $\hat{1}$ direction. As we shall see later, ``$2v$'' comes from the number of independent generators for the Pauli group acting on a single composite particle.

Before starting the proof of theorem~\ref{theorem_decomposition}, let us describe the entire sketch of the proof of theorem~\ref{theorem_3dim}. As a simple extension of theorem~\ref{theorem_decomposition}, one can show that a one-dimensional logical operator defined inside $P(2v, n_{2},1)$ can be further decomposed as a product of a one-dimensional and a zero-dimensional centralizer operators. After these decompositions, one can classify geometric shapes of logical operators according to their dimensions and can find commutation relations between them. 

Although theorem~\ref{theorem_decomposition} is limited to some specially chosen system sizes: $n_{1}=2\cdot2^{2n_{2}v}!$ and $n_{2}=2^{m}$, one can construct logical operators for arbitrary system sizes from theorem~\ref{theorem_decomposition}. For example, due to scale symmetries, one can show that one-dimensional logical operators found in theorem~\ref{theorem_decomposition} are also logical operators for the systems with arbitrary $n_{1}$. In fact, one can find logical operators in the forms described in theorem~\ref{theorem_3dim}. These arguments will be presented in~\ref{sec:construction}.

\begin{figure}[htb!]
\centering
\includegraphics[width=0.65\linewidth]{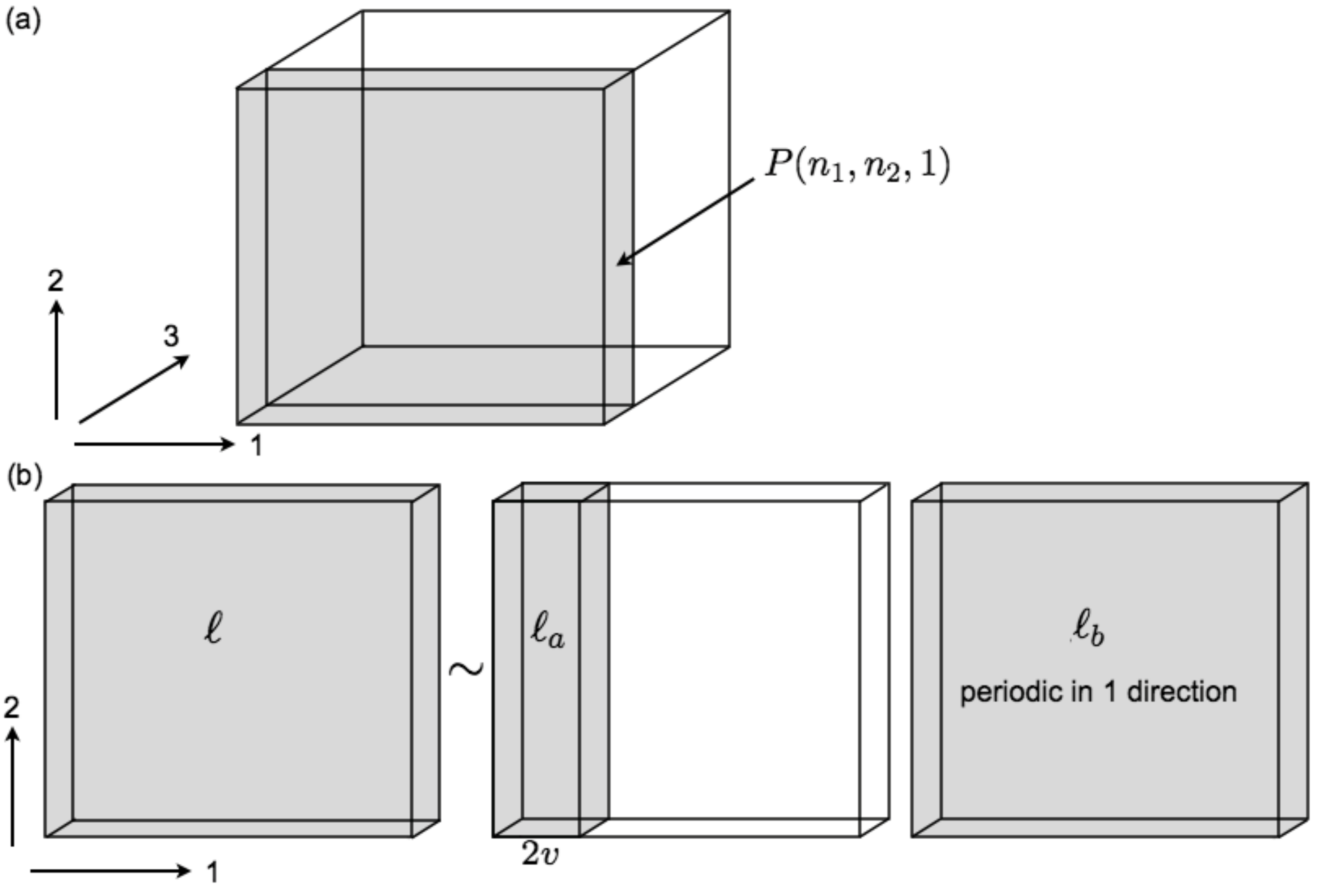}
\caption{The claim of theorem~\ref{theorem_decomposition}. One can decompose a two-dimensional logical operator as a product of a one-dimensional centralizer operator $\ell_{a}$ and a two-dimensional centralizer operator $\ell_{b}$.
} 
\label{fig_decomposition3D}
\end{figure}

\subsection{Sketch of proof of theorem~\ref{theorem_decomposition}}

First, we note that theorem~\ref{theorem_decomposition} was proven for $n_{2}=1$ ($m=0$) in~\cite{Beni10b} since such a system with $n_{2}=1$ can be considered as a two-dimensional system which extends only in the $\hat{1}$ and $\hat{3}$ directions. For a two-dimensional STS model, we have the following lemma.

\begin{lemma}\label{lemma_decomposition}
Consider a two-dimensional STS model where $n_{1}= 2\cdot2^{2v}!$ and arbitrary $n_{2}$. For a given logical operator $\ell$ supported inside $P(n_{1},1)$, one can decompose $\ell$ as a product of the following centralizer operators
\begin{align}
\ell \ \sim \ \ell_{a}\ell_{b}, \qquad \ell_{a},\ell_{b}\ \in \ \mathcal{C}_{P(n_{1},1)}
\end{align}
where
\begin{align}
T_{1}^{\beta}(\ell_{b})\ = \ \ell_{b}, \qquad  \beta \ \leq \ 2^{2v}
\end{align}
and $\ell_{a}$ is defined inside $P(2v,1)$.
\end{lemma}

We present the claim of the lemma graphically in Fig.~\ref{fig_decomposition2D}. The lemma claims that a one-dimensional logical operator defined inside $P(n_{1},1)$ can be decomposed as a product of a zero-dimensional centralizer operator $\ell_{a}$ and a one-dimensional centralizer operator $\ell_{b}$ which is periodic. 

\begin{figure}[htb!]
\centering
\includegraphics[width=0.65\linewidth]{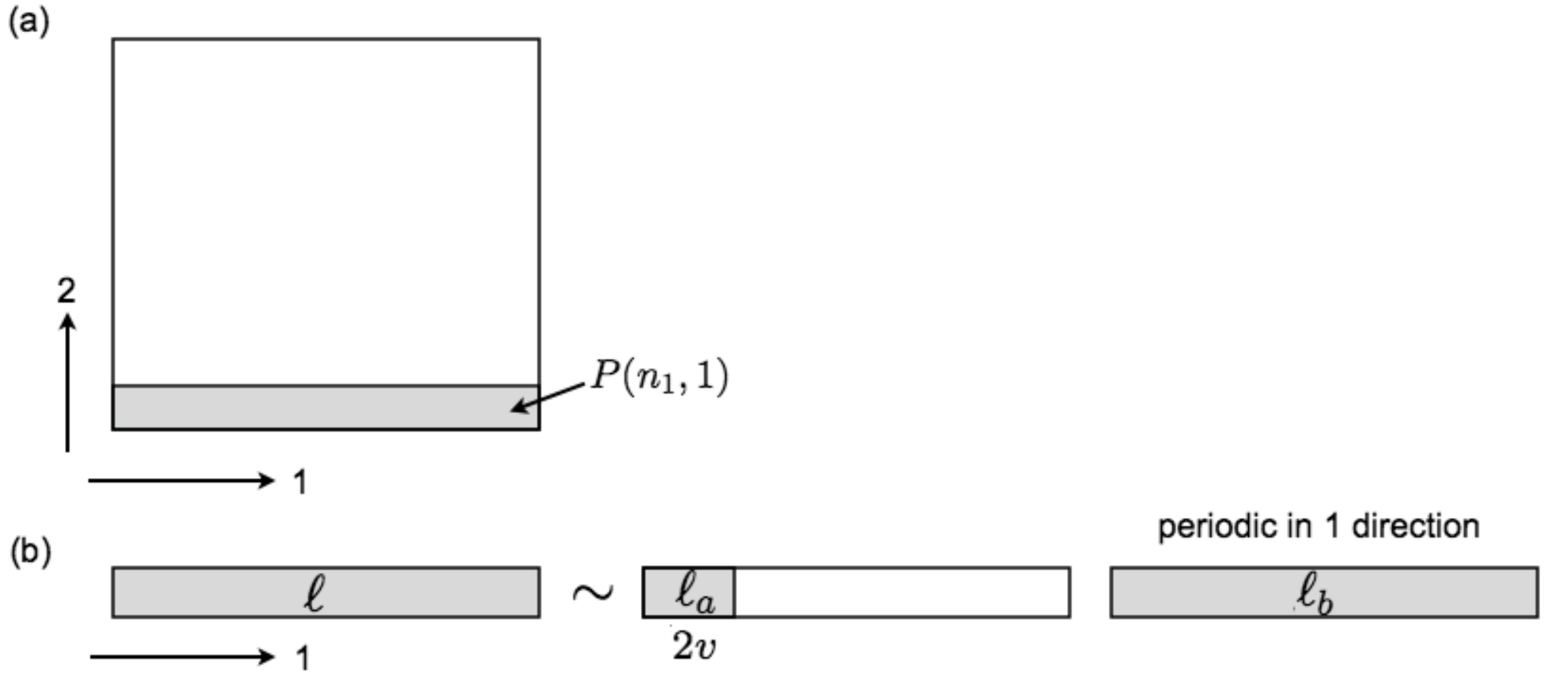}
\caption{The claim of lemma~\ref{lemma_decomposition}. One can decompose a one-dimensional logical operator as a product of a zero-dimensional centralizer operator $\ell_{a}$ and a one-dimensional centralizer operator $\ell_{b}$.
} 
\label{fig_decomposition2D}
\end{figure}

A three-dimensional STS model may be viewed as a two-dimensional system if one considers $1\times n_{2} \times 1$ composite particles as a single composite particle which consists of $vn_{2}$ qubits (see Fig.~\ref{fig_decomposition3Dsub}). In other words, we view the entire system as a two-dimensional lattice of one-dimensional tubes. Then, as a direct consequence of the lemma above, we notice the following corollary.

\begin{corollary}\label{corollary_decomposition}
A logical operator $\ell$ considered in theorem~\ref{theorem_decomposition} can be decomposed as a product of the following centralizer operators
\begin{align}
\ell \ \sim \ \ell_{a}\ell_{b}, \qquad \ell_{a},\ell_{b} \ \in \ \mathcal{C}_{P(n_{1},n_{2},1)}
\end{align}
where
\begin{align}
T_{1}^{\beta}(\ell_{b}) \ = \ \ell_{b}, \qquad \mbox{where}\quad \beta \ \leq \ 2^{2n_{2}v}
\end{align}
and $\ell_{a}$ is defined inside $P(2vn_{2}, n_{2},1)$.
\end{corollary}

We present the claim of the corollary graphically in Fig.~\ref{fig_decomposition3Dsub}. The corollary claims that a two-dimensional logical operator defined inside $P(n_{1},n_{2},1)$ can be decomposed as a product of a one-dimensional logical operator $\ell_{a}$ and a two-dimensional logical operator $\ell_{b}$ which is periodic in the $\hat{1}$ direction. 

However, a one-dimensional logical operator $\ell_{a}$ described in corollary~\ref{corollary_decomposition} is \emph{not one-dimensional in a strict sense} since it is defined inside $P(2^{m}\cdot 2v,2^{m},1)$, and its ``width'' $2^{m}\cdot 2v$ grows as $m$ increases. On the other hand, $\ell_{a}$ described in theorem~\ref{theorem_decomposition} is truly one-dimensional since its width is at most $2v$. Therefore, we need to show that a logical operator defined inside $P(2^{m}\cdot 2v,2^{m},1)$ have an equivalent logical operator defined inside $P(2v,2^{m},1)$. 

\begin{figure}[htb!]
\centering
\includegraphics[width=0.65\linewidth]{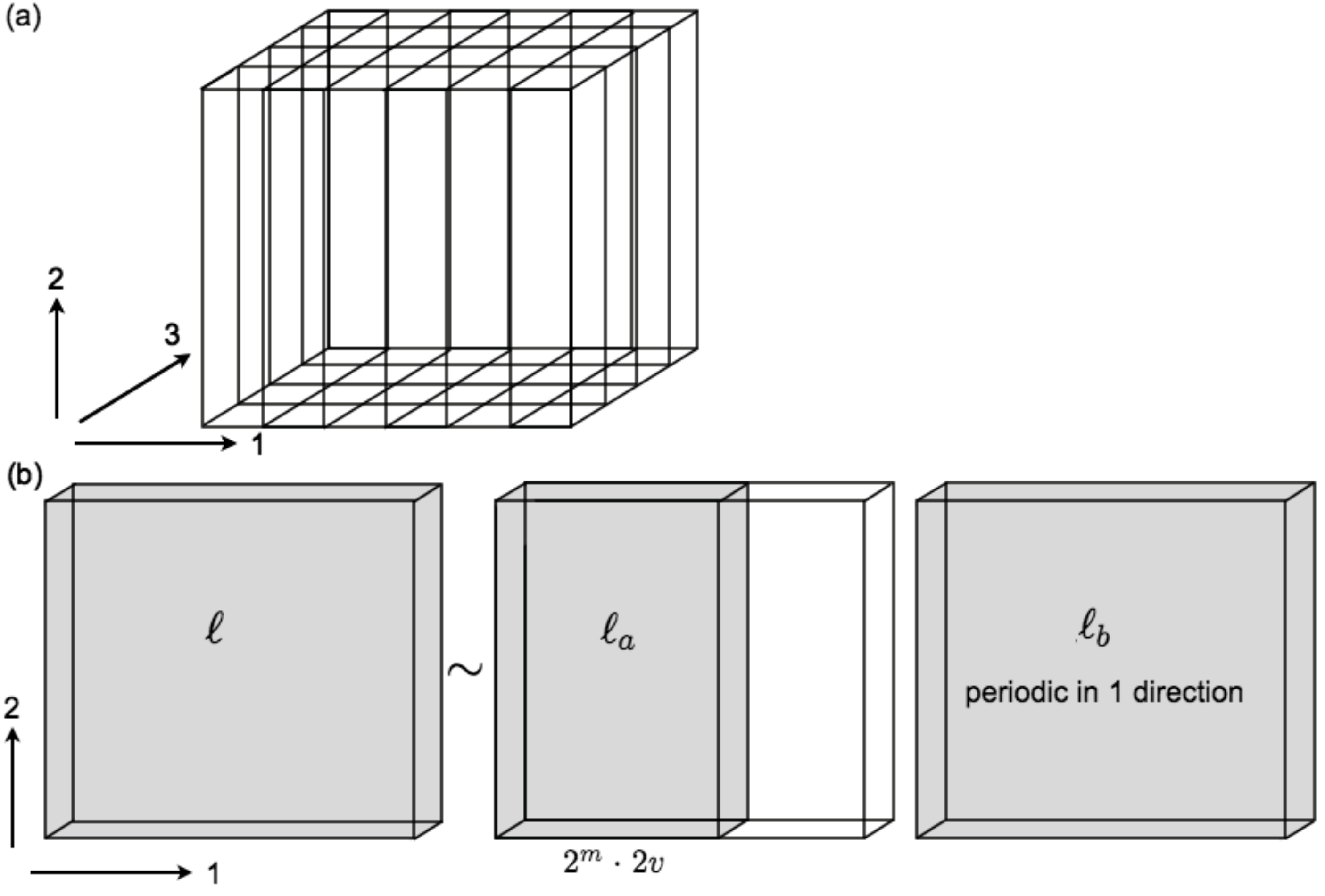}
\caption{The claim of corollary~\ref{corollary_decomposition}. The width of a one-dimensional logical operator $\ell_{a}$ increases as $m$ increases. 
} 
\label{fig_decomposition3Dsub}
\end{figure}

The rest of this appendix is dedicated to the proof of the following lemma.

\begin{lemma}[\textbf{Shrinkage}]\label{lemma_decomposition_final}
For system sizes considered in theorem~\ref{theorem_decomposition}, a logical operator operator $\ell$ defined inside $P(x,2^{m},1)$ always has an equivalent logical operator $\ell'$ which is defined inside $P(x-1,2^{m},1)$ when $2v < x < n_{1}$. 
\end{lemma}

By using this lemma, one can shorten the width of $\ell_{a}$ from $2^{m}\cdot 2v$ to $2v$.

\subsection{Identity generating matrix}

Here, we discuss how to shorten the width of $\ell_{a}$. In particular, we introduce a certain binary matrix which is essential in reducing the width of $\ell_{a}$. 

\textbf{Shrinkage in two dimensions:}
To give an intuition on how to shorten the width of $\ell_{a}$, let us first consider the case where $m=0$. (So, this is a two-dimensional system with $n_{2}=1$, and instead of the ``width'', we use the ``length''). As an example, consider the case when $m=0$, $x=4$ and $\ell$ is given by
\begin{align}
\ell \ = \ \begin{bmatrix}
A, & B,& C ,& AB
\end{bmatrix}
\end{align}
where $\ell$ is defined inside $P(4,1,1)$, and $A$, $B$ and $C$ are some Pauli operators. Then, consider the following logical operator:
\begin{align}
\ell'' \ &\equiv \ \ell T_{1}^{2}(\ell) T_{1}^{3}(\ell) \\
       &= \
\begin{bmatrix}
A, & B,& AC ,& I ,& BC ,& ABC ,& AB
\end{bmatrix}.
\end{align}
Note that $\ell'' \sim \ell$ since $\ell''$ is a product of three logical operators which are equivalent to each other due to the translation equivalence of logical operators. Then, we notice that following operators are centralizer operators:
\begin{align}
\ell_{1} \ = \ \begin{bmatrix}
A, & B,& AC 
\end{bmatrix},\qquad
\ell_{2} \ = \ \begin{bmatrix}
BC ,& ABC ,& AB
\end{bmatrix}
\end{align}
where $\ell'' = \ell_{1}T_{1}^{4}(\ell_{2})$ since stabilizers in STS models are defined inside $2\times 2$ composite particles and cannot overlap with $\ell_{1}$ and $T_{1}^{4}(\ell_{2})$ simultaneously. Now, due to the translation equivalence of logical operators, we have
\begin{align}
\ell \  \sim \ \ell' \ \equiv \ \ell_{1}\ell_{2} \ = \
\begin{bmatrix}
A, & B,& AC
\end{bmatrix}
\ \times \
\begin{bmatrix}
BC ,& ABC ,& AB
\end{bmatrix}
\ = \
\begin{bmatrix}
ABC, & AC,& BC
\end{bmatrix}.
\end{align}
Thus, a logical operator $\ell$, with the length $4$, is shrunk into an equivalent logical operator $\ell'$, with the length $3$.

An important observation is that one can form an identity operator $I$ by taking a product of Pauli operators in $\ell$. In general, if the length $x$ of $\ell$ is larger than $2v$, one can always form an identity operator $I$ by taking a product of some Pauli operators in $\ell$ since there are $2v$ independent generators for single Pauli operators acting on a single composite particle. Now, let us consider $\ell$ represented as 
\begin{align}
\ell \ =\ \begin{bmatrix}
U_{1}, & U_{2},& \cdots ,& U_{x}
\end{bmatrix}.
\end{align}
which is defined inside $P(x,1,1)$ where $x > 2v$. Then, there always exists a binary vector $B=(B_{1},\cdots,B_{x}) \not= (0,\cdots,0)$ which satisfies the following condition:
\begin{align}
\prod_{j=1}^{x} U_{j}^{B_{j}} \ = \ I.
\end{align}
Now, let us take the following product of translations of $\ell$:
\begin{align}
\prod_{j=1}^{x}T_{1}^{x-j}(\ell^{B_{j}}).
\end{align}
Then, one may readily know that the $x$th entry of the above operator is $I$. From this operator, one can find two centralizer operators. By using them, one can readily shrink the length of $\ell$ from $x$ to $x-1$. This trick is the key to the proof of lemma~\ref{lemma_decomposition}. Although the argument above works only when $B$ has an odd number of $1$ entries, a slight modification makes the shrinkage of $\ell$ possible when $B$ has an even number of $1$ entries.

\textbf{Identity generating matrix:}
Now, we consider more general cases with $m>0$. Let us represent a logical operator $\ell$ defined inside a region $P(x, 2^{m},1)$ as a $x\times 2^{m}$ matrix whose entries are single Pauli operators:
\begin{align}
\ell \ = \ \begin{bmatrix}
U_{1,1},& \cdots& U_{x,1} \\
\vdots     & \ddots& \vdots \\
U_{1,2^{m}},& \cdots& U_{x,2^{m}} 
\end{bmatrix}
\end{align}
where each Pauli operator $U_{i,j}$ acts on each composite particle. We also represent each column of $\ell$ as follows:
\begin{align}
U_{j} \ =\ \begin{bmatrix}
U_{j,1} \\
\vdots      \\
U_{j,2^{m}} 
\end{bmatrix}\qquad ( j\ =\ 1,\cdots,x  ).
\end{align}
Here, we denote a group of Pauli operators supported by a single column $P(1,2^{m},1)$ as $\mathcal{P}^{m}_{col}$ and call it the \emph{column operator group} and its elements \emph{column operators}. Note that $U_{j}\in \mathcal{P}^{m}_{col}$, and $\mathcal{P}^{m}_{col}$ has $2^{m}\cdot 2v$ independent generators: $G(\mathcal{P}^{m}_{col}) = 2^{m}\cdot 2v$.

When $m=0$ and $x >2v$, we found a binary vector $B$ with $x$ components which characterizes how to form an identity operator from $\ell$. When $m>0$ and $x >2v$, we can find an $x \times 2^{m}$ binary matrix $B$ which characterizes how to form an identity operator from $\ell$. Here, we introduce the \emph{identity generating matrices} as follows.

\begin{definition}
Consider a logical operator $\ell$ defined inside $P(x,2^{m},1)$.
\begin{itemize}
\item For a $x\times 2^{m}$ binary matrix $B$
\begin{align}
B \ = \ \begin{bmatrix}
B_{1,1}    ,& \cdots & B_{x,1}   \\
\vdots    & \ddots &  \vdots   \\  
B_{1,2^{m}} ,& \cdots & B_{x,2^{m}} 
\end{bmatrix}, \qquad B_{i,j} \ = \ 0,1, 
\end{align}
we define the following operations:
\begin{align}
\ell(B) \ &\equiv \ \prod_{i,j} T_{1}^{x -i}T_{2}^{j-1} ( \ell^{B_{i,j}} ) \\
\ell(B)_{x}\ &\equiv \ \prod_{i=1}^{x}\prod_{j=1}^{2^{m}}T_{2}^{j-1}(U_{i}^{B_{i,j}})\  \in \ \mathcal{P}^{m}_{col}
\end{align}
where $\ell(B)$ is a product of translations of $\ell$ taken according to $B$ while $\ell(B)_{x}$ is the $x$th column of $\ell(B)$. 
\item We call a binary matrix $B$ \emph{identity generating matrix} if and only if 
\begin{align}
\ell(B)_{x} \ = \ I \quad \mbox{and}\quad B \ \not= \ \begin{bmatrix}
0, & \cdots & 0 \\
\vdots  & \ddots  & \vdots \\
0, & \cdots  & 0
\end{bmatrix}.
\end{align}
\item We assign parities to each column of a binary matrix $B$ as follows:
\begin{align}
Par(B)_{i}\ \equiv \ \sum_{j} B_{i,j} \qquad ( \mbox{mod 2})
\end{align}
where $i = 1,\cdots,x$. We call a binary matrix $B$ \emph{odd} if and only if 
\begin{align}
\exists i \quad  \mbox{s.t} \quad Par(B)_{i} \ =\ 1.
\end{align}
\end{itemize} 
\end{definition}

Therefore, when we form an identity operator, we considered translations of $\ell$ both in the $\hat{1}$ and $\hat{2}$ directions. The identity generating matrix is said to be odd when there exists a column with an odd parity. 

\textbf{Shrinkage through odd matrices:}
Note that there always exists an identity generating matrix when $x >2v$. Then, with the existence of identity generating matrices, one might hope that the width of $\ell_{a}$ can be reduced until it becomes $2v$ in a way similar to the cases where $m=0$. However, there is a caveat. In fact, only identity generating matrices with some special properties can be used for shrinking. In particular, we have the following lemma.

\begin{lemma}[\textbf{Shrinkage through odd matrices}]\label{lemma_odd1}
If there exists an odd identity generating matrix for $\ell$ defined inside $P(x, 2^{m}, 1)$, $\ell$ has an equivalent logical operator $\ell'$ defined inside $P(x-1, 2^{m}, 1)$.
\end{lemma}

Therefore, if there exists an odd identity generating matrix for any $x$ with $n_{1}>x>2v$, one can complete the proof of lemma~\ref{lemma_decomposition_final}. Below, we present the proof of lemma~\ref{lemma_odd1}. The existence of an odd identity generating matrix will be proven later.

\begin{proof}
Assume that $B$ is an odd identity generating matrix for $\ell$. Assume that for some $i'$ ($1 \leq i' \leq x$), we have
\begin{align}
Par(B)_{i'} \ = \  1 \quad \mbox{and} \quad  Par(B)_{i} \ = \ 0 \quad \mbox{for} \ i \ < \ i'.
\end{align}
So, $i'$ is the smallest integer such that $i'$th column has an odd parity.
Here, we define the following binary matrix $B'$ (see Fig.~\ref{fig_proof_aid}):
\begin{align}
B'_{i,j} \ &\equiv \ B_{i,j} \qquad & \ (\ i \ \leq \ i'\ )\\
        &\equiv \ 0      \qquad & \  (\ i \ > \ i' \ ).  
\end{align}
Note that $B'$ consists of $i$th columns of $B$ with $i\leq i'$. Based on $B'$, we consider the following logical operator $\ell'$:
\begin{align}
\ell' \ \equiv \ \ell(B') \ \sim \ \ell.
\end{align}
Note that $\ell'$ is equivalent to $\ell$ since $\ell'$ is a product of an odd number of translations of $\ell$. (Note that $\ell(B)$ may not be equivalent to $\ell$ since the number of $1$ entries in $B$ may be even).

See Fig~\ref{fig_proof_aid} for graphical representations of $B$, $B'$, $\ell(B)$ and $\ell(B')$. Note that $\ell(B)$ has an identity operator at $x$th column. Note that $\ell(B')$ has identity operators in the first $x - i'$ columns since $B_{i,j}=0$ for $i >i'$.

\begin{figure}[htb!]
\centering
\includegraphics[width=0.75\linewidth]{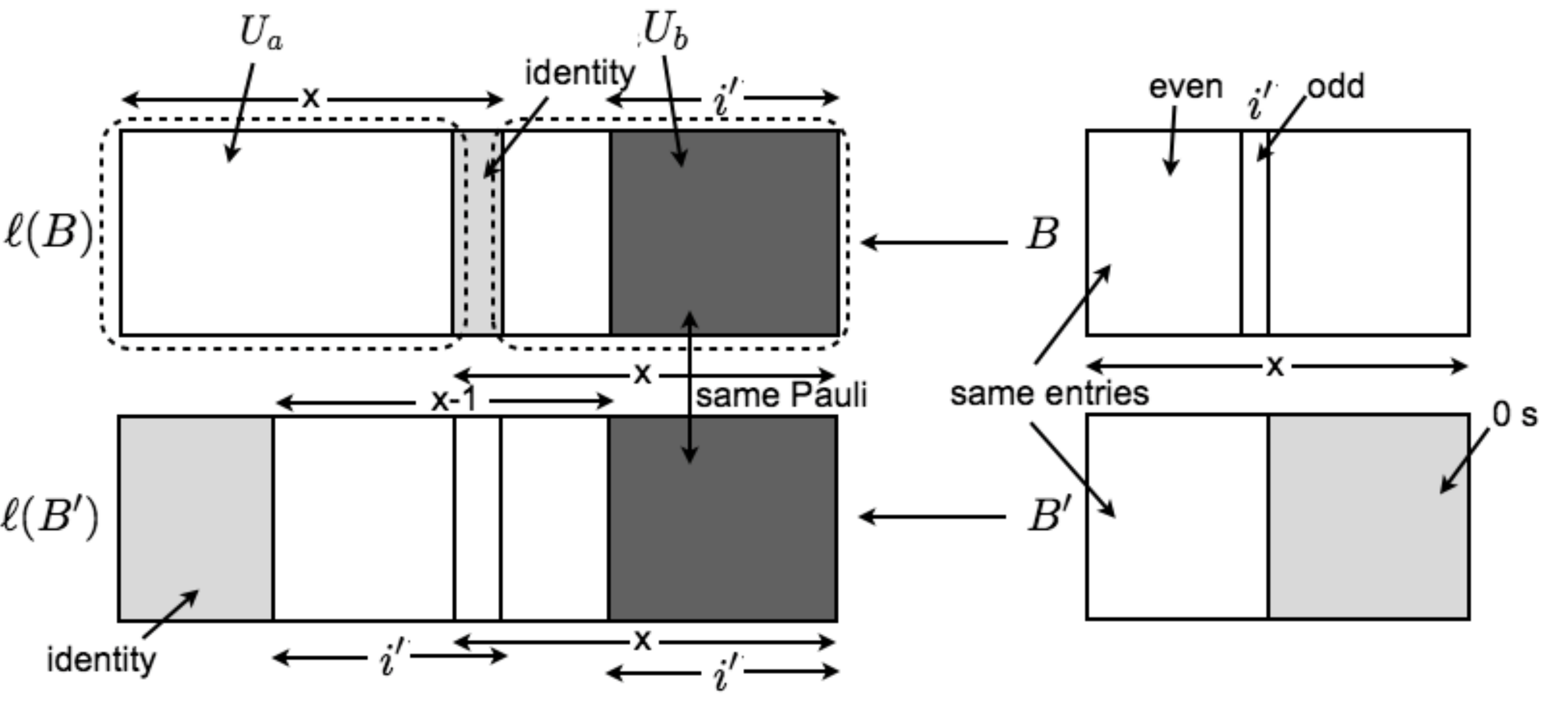}
\caption{Constructions of $B$, $B'$, $\ell(B)$ and $\ell(B')$.
} 
\label{fig_proof_aid}
\end{figure}

Since $\ell(B)$ has an identity operator at $x$th column, we can decompose it as a product of two centralizer operators whose lengths are at most $x-1$. Let us denote these centralizer operators as $U_{a}$ and $U_{b}$:
\begin{align}
\ell(B) \ = \ U_{a}U_{b}
\end{align}
where $U_{a}$ is the centralizer on the left hand side and $U_{b}$ is the centralizer on the right hand side (See Fig.~\ref{fig_proof_aid}). $U_{a}$ is defined from $1$st column to $x-1$th column, and $U_{b}$ is defined from $x+1$th column to $2x-1$th column.

Since $B$ and $B'$ have the same entries from $1$st column to $i'$th column, we notice that $\ell(B')$ and $U_{b}$ have the same Pauli operators from $2x-2-i'$st column to $2x-1$th column as shown in Fig.~\ref{fig_proof_aid}. Then, by applying $U_{b}$ to $\ell(B')$, one can shrink the size of $\ell(B')$. In particular, $U_{b} \ell(B')$ has the length at most $x-1$. Although $U_{b}\ell(B')$ may not be equivalent to $\ell(B')$ as $U_{b}$ may not be a stabilizer, one may consider the following logical operator:
\begin{align}
U_{b} \ \times \ \ell(B') \ \times \ T_{1}^{- i'}(U_{b}) \ \sim \ \ell(B') \ \sim \ \ell
\end{align}
which is equivalent to $\ell(B')$ due to the translation equivalence of logical operators, and is defined inside a region with $x-1 \times 2^{m} \times 1$ composite particles. Then, due to the translation equivalence of logical operators, there exists a logical operator $\ell'' \sim \ell$ which is defined inside $P(x-1, 2^{m}, 1)$. This completes the proof.
\end{proof}

\subsection{Existence of an odd matrix for $m=1$}

Next, we present a proof of the existence of an odd identity generating matrix for $x > 2v$, in order to complete the proof of lemma~\ref{lemma_decomposition_final} and theorem~\ref{theorem_decomposition}. In particular, we shall prove the following lemma.

\begin{lemma}\label{lemma_existence}
When $x > 2v$, there always exist an odd identity generating matrix. 
\end{lemma}

We start by discussing cases with $m=0$ and $m=1$ before presenting general discussion. First of all, when $m=0$, identity generating matrices are always odd since all the binary matrices are odd except $B=(0,\cdots,0)$. Therefore, we consider the cases where $m=1$ below. 

\textbf{Characteristic value:}
Recall that we represented $\ell$ as a $x \times 2$ binary matrix:
\begin{align}
\ell \ =\ \begin{bmatrix}
U_{1,1},& \cdots& U_{x,1} \\
U_{1,2},& \cdots& U_{x,2} 
\end{bmatrix}
\end{align}
and each column of $\ell$ as follows:
\begin{align}
U_{j} \ = \ \begin{bmatrix}
U_{j,1} \\
U_{j,2} 
\end{bmatrix}\qquad ( j \ = \ 1,\cdots,x  ).
\end{align}
Now, \emph{we suppose that there is no odd identity generating matrix} for $\ell$, in order to use the contradiction for the proof of lemma~\ref{lemma_existence}. 

First, it is convenient to classify column operators in $\mathcal{P}_{col}^{1}$ into two types as follows. For a column operator $U$ represented as
\begin{align}
U = \begin{bmatrix}
U_{1,1} \\
U_{1,2} 
\end{bmatrix},
\end{align}
we assign a \emph{characteristic value} $b$ and \emph{characteristic operator} $V$ as follows:
\begin{itemize}
\item If $U_{1,1}=U_{1,2}$, we assign a characteristic value $b = 1$ and a characteristic operator $V=U_{1,1}$.
\item If $U_{1,1}\not=U_{1,2}$, we assign a characteristic value $b= 0$ and a characteristic operator $V=U_{1,1}U_{1,2}$.
\end{itemize}
One may easily understand this classification by representing $U$ explicitly. A column operator $U$ with $b=1$ is
\begin{align}
U \ = \ \begin{bmatrix}
V \\
V 
\end{bmatrix}
\end{align}
and a column operator $U$ with $b=0$ is
\begin{align}
U \ = \ \begin{bmatrix}
VU_{1,2} \\
U_{1,2}
\end{bmatrix}.
\end{align}
So, a column operator with $b=1$ is symmetric while a column operator with $b=0$ is not. Note that a characteristic operator is not an identity operator $I$ except when $U=I$.

Next, we introduce the following $x \times 2$ binary matrices:
\begin{align}
E(i';0) \ \equiv \ \begin{bmatrix}
E_{1,1} ,& \cdots ,& E_{x,1} \\
E_{1,2} ,& \cdots ,& E_{x,2}
\end{bmatrix}
\end{align}
such that
\begin{align}
E_{i,1} \ &= \ 1 \qquad ( i \ = \ i')\\
E_{i,j} \ &= \ 0 \qquad \mbox{otherwise}
\end{align}
and 
\begin{align}
E(i';1) \ \equiv \ \begin{bmatrix}
E_{1,1} ,& \cdots ,& E_{x,1} \\
E_{1,2} ,& \cdots ,& E_{x,2}
\end{bmatrix}
\end{align}
such that
\begin{align}
E_{i,1} \ &= \ E_{i,2} \ = \  1 \qquad ( i \ = \ i')\\
E_{i,j} \ &= \ 0 \qquad \ \ \quad \qquad \mbox{otherwise}
\end{align}
For example, 
\begin{align}
E(2;0) \ = \ \begin{bmatrix}
0, & 1 ,& 0 ,& \cdots ,& 0 \\
0, & 0 ,& 0 ,& \cdots ,& 0
\end{bmatrix}
\end{align}
and 
\begin{align}
E(2;1) \ = \ \begin{bmatrix}
0 , & 1 ,& 0 ,& \cdots ,& 0 \\
0 , & 1 ,& 0 ,& \cdots ,& 0
\end{bmatrix}.
\end{align}

Now, let us represent characteristic values and characteristic operators for each column of $\ell$ as follows:
\begin{align}
\begin{array}{ccc}
U_{1}\ \rightarrow\ & b_{1}, & V_{1} \\
U_{2}\ \rightarrow\ & b_{2},& V_{2} \\
\vdots & \vdots & \vdots \\
U_{x}\ \rightarrow\  & b_{x},& V_{x}.
\end{array}
\end{align}
Then one can establish a connection between binary matrices $E(i;0)$ and $E(i;1)$, and characteristic values $b_{i}$ and operators $V_{j}$ as follows:
\begin{align}
\ell(E(i;0))_{x} \ &= \  \begin{bmatrix}
V_{i} \\
V_{i} 
\end{bmatrix}\qquad \mbox{when} \quad b_{i}\ = \ 1\\
\ell(E(i;1))_{x} \ &= \  \begin{bmatrix}
V_{i} \\
V_{i} 
\end{bmatrix}\qquad \mbox{when} \quad b_{i}\ = \ 0.
\end{align}

\textbf{Proof of lemma~\ref{lemma_existence}:}
Now, let us proceed to the proof of lemma~\ref{lemma_existence} for $m=1$. Without loss of generality, we can assume that
\begin{align}
b_{1}\ = \ \cdots\ =\ b_{x_{0}}\ =\ 0, \qquad b_{x_{0}+1}\ =\ \cdots \ = \ b_{x}\ =\ 1
\end{align}
for $x_{0}\leq x$ since permutations of columns do not affect the parities of identity generating matrices. (We will justify this later). We define the following sets of integers:
\begin{align}
\textbf{b}(0)\ \equiv \ \{ 1,\cdots, x_{0} \}, \qquad \textbf{b}(1)\ \equiv \{ x_{0}+1, \cdots, x \}.
\end{align}
We denote groups of Pauli operators generated by $V_{j}$ with $b_{j}=0$ and $V_{j}$ with $b_{j}=1$ as $\mathcal{V}_{0}$ and $\mathcal{V}_{1}$:
\begin{align}
\mathcal{V}_{0} \ &\equiv \ \left\langle \ \{ \ V_{i} \ : \ i \ \in \ \textbf{b}(0) \ \} \ \right\rangle \\
\mathcal{V}_{1} \ &\equiv \ \left\langle \ \{ \ V_{i} \ : \ i \ \in \ \textbf{b}(1) \ \} \ \right\rangle. 
\end{align}

Let us show that a set of characteristic operators $\{ V_{i} \}$ for $i \in \textbf{b}(1)$ is independent : $G(\mathcal{V}_{1})=x - x_{0}$. For this purpose, we suppose that there exists some set of integers $\textbf{A}\subseteq \textbf{b}(1)$ such that
\begin{align}
\prod_{i \in \textbf{A}} V_{i} \ = \ I.
\end{align}
Then, the following binary matrix is an identity generating matrix:
\begin{align}
B \ = \ \sum_{i \in \textbf{A}}E(i;0)  
\end{align}
since 
\begin{align}
\ell(B)_{x} \ = \ \prod_{i \in \textbf{A}} \begin{bmatrix}
V_{i}\\
V_{i}
\end{bmatrix}\ = \ \begin{bmatrix}
I\\
I
\end{bmatrix}.
\end{align}
However, since $B$ is odd with $Par(B)_{i}=1$ for $i \in \textbf{A}$, this leads to a contradiction.

Next, let us analyze $\mathcal{V}_{0}$. For simplicity of discussion, we \emph{first assume that $\{V_{i}\}$ for $i \in \textbf{b}(0)$ are independent}. We consider more general cases where $\{V_{i}\}$ for $i \in \textbf{b}(0)$ are over complete later. Then, we have $G(\mathcal{V}_{0})=x_{0}$. Now, we define the following operators for $i\in \textbf{b}(0)$:
\begin{equation}
\begin{split}
U_{i}'\ &\equiv \ U_{i}T_{2}(U_{i}) \\
       &= \ \begin{bmatrix}
       V_{i} \\
       V_{i}
       \end{bmatrix}. 
\end{split}
\end{equation}
Note that $U_{i}'$ has a characteristic value $b_{i}'=1$ and a characteristic operator $V_{i}$. Notice that
\begin{align}
U_{i}' \ = \ \ell(E(i;1))_{x}.
\end{align}

Since $x > 2v$, there exists a set of integer $\textbf{A}$ such that 
\begin{align}
\prod_{i \in \textbf{A}} V_{i} \ = \ I, \qquad \textbf{A} \ \not\subseteq \ \textbf{b}(1)\quad \mbox{and} \quad \textbf{A} \ \not\subseteq \ \textbf{b}(0).
\end{align}
Note that $\textbf{A}$ includes integers both from $\textbf{b}(0)$ and $\textbf{b}(1)$. Then, one notices that the following matrix $B$ is an identity generating matrix:
\begin{align}
B \ = \ \sum_{i \in \textbf{A} \cap \textbf{b}(0)} E(i;1) + \sum_{i \in \textbf{A} \cap \textbf{b}(1)} E(i;0)
\end{align}
since
\begin{equation}
\begin{split}
\ell(B)_{x}\ &= \ \prod_{\{ i \in \textbf{A}\cap \textbf{b}(0) \} } U_{i}' \prod_{\{ i \in \textbf{A} \cap \textbf{b}(1) \} } U_{i} \\
        &= \ \prod_{i \in \textbf{A}} \begin{bmatrix}
         V_{i} \\
         V_{i}
        \end{bmatrix} \ = \ 
        \begin{bmatrix}
         I \\
         I
        \end{bmatrix}.
\end{split}
\end{equation}
However, for an integer $i$ in $\textbf{A}\cap \textbf{b}(0)$, $Par(B)_{i}=1$, and $B$ is an odd matrix. This leads to a contradiction. Note that discussion above is valid under the permutations of columns. 

\textbf{Proof of lemma~\ref{lemma_existence}, continued:}
Next, let us consider the case where $\{ V_{i} \}$ for $i \in \textbf{b}(0)$ are not independent. Let us denote the number of generators for $\mathcal{V}_{0}$ as $\bar{x_{0}}=G(\mathcal{V}_{0})$ ($\bar{x_{0}}<x_{0}$). Without loss of generality, we may assume that $V_{1},\cdots,V_{\bar{x_{0}}}$ are independent since permutations do not affect parities of identity generating matrices. Here, \emph{we change our notations slightly}:
\begin{align}
\textbf{b}(0)\ \equiv \ \{ 1,\cdots, \bar{x_{0}} \}, \qquad \textbf{b}(0)'\ \equiv \ \{ \bar{x_{0}}+1,\cdots, x_{0} \}, \qquad \textbf{b}(1)\ \equiv \{ x_{0}+1, \cdots, x \}.
\end{align}
and
\begin{align}
\mathcal{V}_{0} \ &\equiv \ \left\langle \ \{ \ V_{i} \ : \ i \ \in \ \textbf{b}(0) \ \} \ \right\rangle  \qquad G(\mathcal{V}_{0}) \ = \ \bar{x_{0}}\\
\mathcal{V}_{1} \ &\equiv \ \left\langle \ \{ \ V_{i} \ : \ i \ \in \ \textbf{b}(1) \ \} \ \right\rangle  \qquad G(\mathcal{V}_{1}) \ = \ x - x_{0}. 
\end{align}

Since a set $\{V_{i}\}$ for $b_{i}=0$ is over complete, there are $x_{0}-\bar{x_{0}}$ sets of integers $\textbf{A}_{i}$ ($i=\bar{x_{0}}+1,\cdots,x_{0}$) such that
\begin{align}
\prod_{i' \in \textbf{A}_{i}}V_{i'}\ = \ I\quad \mbox{where} \quad i' \ \leq \ i,\ \forall i' \ \in \ \textbf{A}_{i} \quad \mbox{and}\quad i \ \in \ \textbf{A}_{i}
\end{align}
where the largest integer in $\textbf{A}_{i}$ is $i$.
For $i = \bar{x_{0}}+1,\cdots, x_{0}$, we form the following operator:
\begin{align}
U_{i}' \ \equiv \ \prod_{j \in \textbf{A}_{i}} U_{j} 
       \  \equiv \ \begin{bmatrix}
         V_{i}' \\
         V_{i}'
         \end{bmatrix}
\end{align}
which has a characteristic value $b_{i}'=1$ and a characteristic operator $V_{i}'$. Here, we denote a group of $\{ V_{i}'\}$ for $i \in \textbf{b}(0)'$ as 
\begin{align}
\mathcal{V}_{0}'\  = \ \left\langle \ \{ \ V_{i}' \ : \ i \ \in \ \textbf{b}(0)' \ \} \ \right\rangle.
\end{align}
Now, we show that $\{V_{i}'\}$ for $i\in \textbf{b}(0)'$ are independent: $G(\mathcal{V}_{0}')= x_{0}-\bar{x_{0}}$. If there exists $\textbf{A} \subseteq \textbf{b}(0)'$ such that
\begin{align}
\prod_{i \in \textbf{A}} V_{i}' \ = \ I,
\end{align}
we have the following identity generating matrix:
\begin{align}
B \ = \ \sum_{i \in \textbf{A}}\sum_{j \in \textbf{A}_{i}} E(j;0) \qquad \mbox{(mod 2)}
\end{align}
since
\begin{align}
\ell(B)_{x} \ = \ \prod_{i \in \textbf{A}}\begin{bmatrix}
V_{i}' \\
V_{i}'
\end{bmatrix}
\ = \ \begin{bmatrix}
I \\
I
\end{bmatrix}.
\end{align}
Recall that the largest integer in $\textbf{A}_{i}$ is $i$. Let the largest integer in $\textbf{A}$ be $i_{max}$. Then, we have $Par(B)_{i_{max}}=1$, and thus, $B$ is odd. This leads to a contradiction.

So far, we have shown that
\begin{align}
G(\mathcal{V}_{0}) \ = \ \bar{x_{0}},\qquad  G(\mathcal{V}_{0}') \ = \ x_{0}-\bar{x_{0}},\qquad G(\mathcal{V}_{1}) \ = \ x- x_{1}.
\end{align}
Since $x > 2v$, there exists a set of integers $\textbf{A}$ such that 
\begin{align}
\prod_{\{ i \in \textbf{A} \cap \textbf{b}(1) \} } V_{i} 
\prod_{\{ i \in \textbf{A} \cap \textbf{b}(1)' \} } V_{i}'
\prod_{\{ i \in \textbf{A} \cap \textbf{b}(0) \} } V_{i}  \ = \ I.
\end{align}
The following matrix is the identity generating matrix:
\begin{align}
B \ = \ \sum_{i \in \textbf{A} \cap \textbf{b}(0)} E(i;1) + 
\sum_{i \in \textbf{A} \cap \textbf{b}(0)'}  \sum_{i' \in \textbf{A}_{i}}  E(i';0) + 
\sum_{i \in \textbf{A} \cap \textbf{b}(1)} E(i;0) \qquad \mbox{(mod 2)}
\end{align}
since
\begin{align}
\ell(B)_{x} \ = \ \prod_{\{ i \in \textbf{A} : i \in \textbf{b}(0) \} } 
\begin{bmatrix}
V_{i}\\
V_{i}
\end{bmatrix}
\prod_{\{ i \in \textbf{A} : i \in \textbf{b}(0)' \} }
\begin{bmatrix}
V_{i}'\\
V_{i}'
\end{bmatrix}
\prod_{\{ i \in \textbf{A} : i \in \textbf{b}(1) \} } 
\begin{bmatrix}
V_{i}\\
V_{i}
\end{bmatrix} \ = \ \begin{bmatrix}
I\\
I
\end{bmatrix}.
\end{align}
Since $\textbf{A} \not\subseteq \textbf{b}(0)$, $\textbf{A}$ has some element in $\textbf{b}(0)' \cup \textbf{b}(1)$. Let the largest integer in $\textbf{A}$ be $i_{max}$. Then, we have $Par(B)_{i_{max}}=1$, and thus, $B$ is odd. This leads to a contradiction. Again, permutations of columns do not affect this discussion. This completes the proof of lemma~\ref{lemma_existence} for $m=1$.

\subsection{Characteristic vectors}

Let us proceed to the proof for the cases where $m>1$. When $m=1$, we assigned characteristic values $0$ and $1$ to each column operator according to its symmetry. For $m>1$, we will assign a ``binary vector'' with $m$ components to each column operator, which we will call a \emph{characteristic vector}. We encode ``symmetries'' of a column operator on these characteristic vectors. 

\textbf{Characteristic vector:}
We first define two maps $f_{0}$ and $f_{1}$ from $\mathcal{P}^{m'}_{col}$ to $\mathcal{P}^{m'-1}_{col}$ for $m' >0$ as follows. For a given column operator $U \in \mathcal{P}^{m'}_{col}$ which is represented as
\begin{align}
U = \begin{bmatrix} 
U_{1}\\
\vdots \\
U_{2^{m'}}
\end{bmatrix},
\end{align}
we define $f_{0}(U), f_{1}(U) \in \mathcal{P}^{m'-1}_{col}$ as follows:
\begin{align}
f_{0}(U)_j \equiv U_{j}U_{j+ 2^{m'-1}}, \qquad f_{1}(U)_j \equiv U_{j} \qquad (j=1,\cdots,2^{m'-1}).
\end{align}
One may represent $f_{0}(U)$ and $f_{1}(U)$ more explicitly:
\begin{align}
f_{0}(U) = \begin{bmatrix} 
U_{1}U_{2^{m'-1}+1} \\
\vdots \\
U_{2^{m'-1}}U_{2^{m'}} 
\end{bmatrix}, \qquad 
f_{1}(U) = \begin{bmatrix} 
U_{1}\\
\vdots \\
U_{2^{m'-1}} 
\end{bmatrix}
\end{align}
Note that $f_{0}$ and $f_{1}$ decrease the length of the column by half.

Let us denote a set of all the $m$ component binary vectors as $\textbf{B}_{vec}^{m}$. Now, for a column operator $U \in \mathcal{P}^{m}_{col}$, we assign an $m$ component binary vector $\vec{b} \in \textbf{B}^{m}_{vec}$ through the following rule.
\begin{itemize}
\item If $f_{0}(U)=I$, take $b_{m}=1$ and define $U^{(1)} \equiv f_{1}(U)$.
\item If $f_{0}(U)\not= I$, take $b_{m}=0$ and define $U^{(1)} \equiv f_{0}(U)$.
\end{itemize}
and, iterate this procedure:
\begin{itemize}
\item If $f_{0}(U^{(j)})=I$, $b_{m-j}=1$ and define $U^{(j+1)}\equiv f_{1}(U^{(j)})$.
\item If $f_{0}(U^{(j)})\not= I$, $b_{m-j}=0$ and define $U^{(j+1)} \equiv f_{0}(U^{(j)})$.
\end{itemize}
for $1 \leq j \leq m-1$. We define a characteristic operator $V$ of $U$ as follows:
\begin{align}
V \ \equiv \ U^{(m)} \ = \ f_{b_{1}} f_{b_{2}} \cdots f_{b_{m}} (U) \ \in \ \mathcal{P}^{0}_{col}.
\end{align}
Note that $V \not= I$ when $U \not=I$. It is worth presenting examples of characteristic vectors here ($m=2$):
\begin{align}
\begin{bmatrix} 
V \\
V\\
V\\
V
\end{bmatrix} \ \rightarrow \ (1,1), \ \
\begin{bmatrix} 
V \\
I \\
V \\
I
\end{bmatrix} \ \rightarrow \ (0,1), \ \
\begin{bmatrix} 
V \\
V \\
I \\
I
\end{bmatrix} \ \rightarrow \ (1,0), \ \
\begin{bmatrix} 
V \\
I \\
I \\
I
\end{bmatrix} \ \rightarrow \ (0,0).
\end{align}
Thus, \emph{symmetries of column operators are encoded in characteristic vectors.}

Next, we introduce an order between binary vectors in $\textbf{B}^{m}_{vec}$. We define 
\begin{align}
g(\vec{b}) \ \equiv \ \sum_{j=1}^{m} b_{j}2^{j-1}
\end{align}
where $\vec{b}$ is like a binary representation of an integer $g(\vec{b})$. For a given pair of $m$ component binary vectors $\vec{b}$ and $\vec{b'}$, we denote
\begin{align}
\vec{b} \ < \ \vec{b'}
\end{align}
if and only if 
\begin{align}
g(\vec{b}) \ < \ g(\vec{b'}).
\end{align}
For example, for $m=3$, we have the following relations between binary vectors:
\begin{align}
(0,0,0) \ <\ (1,0,0)\ <\ (0,1,0)\ <\ (1,1,0)\ <\ (0,0,1)\ <\ (1,0,1)\ <\ (0,1,1)\ <\ (1,1,1).
\end{align}
Below, we shall see that a column operator with a larger characteristic vector is ``more symmetric'' than a column vector with a smaller characteristic vector.

\textbf{Property of characteristic vectors:}
Let us briefly recall the proof of lemma~\ref{lemma_existence} for $m=1$. In the proof, we constructed a column operator with $b=1$ from a column operator with $b=0$. In particular, if $U$ is a column operator with a characteristic value $b=0$ and a characteristic operator $V$:
\begin{align}
U \ = \ \begin{bmatrix}
VV'\\
V'
\end{bmatrix}
\end{align}
where $V'$ is some Pauli operator, we have
\begin{align}
UT_{2}(U) \ = \ \begin{bmatrix}
V \\
V
\end{bmatrix}
\end{align}
which is a column operator with a characteristic value $b=1$ and a characteristic operator $V$. Thus, \emph{we can create a column operator with a larger characteristic value from a column operator with a smaller characteristic value}. 

In a way similar to this, one can construct a column operator with $\vec{b'}$ from a column operator with $\vec{b}$ as long as $\vec{b'}>\vec{b}$. Let us represent a binary column $B$ as follows:
\begin{align}
B \ = \ \begin{bmatrix}
B_{1,1}\\
\vdots \\
B_{1,2^{m}}
\end{bmatrix},
\end{align}
and denote a set of all the binary columns as $\textbf{B}^{m}_{col}$. Here, we define a parity of $B$ as 
\begin{align}
Par(B) \ \equiv \ Par(B)_{1}
\end{align}
by viewing a binary column $B$ as a binary matrix. For a column operator $U\in \mathcal{P}_{col}^{m}$, we define $U(B)$ as follows:
\begin{align}
U(B)\ \equiv \ \prod_{j=1}^{2^{m}}T_{2}^{j-1}(U^{B_{1,j}})\  \in \ \mathcal{P}^{m}_{col}
\end{align}
just like $\ell(B)$. Note that $U(B)$ is a product of translations of $U$ taken according to a binary column $B$. Then, the following lemma holds.

\begin{lemma}\label{lemma_construction}
Consider an arbitrary pair of $m$ component binary vectors $\vec{b}, \vec{b'} \in \textbf{B}^{m}_{vec}$ such that $\vec{b}<\vec{b'}$. For a given column operator $U$ which has a characteristic vector $\vec{b}$ and a characteristic operator $V$, there always exists some binary column $B \in \textbf{B}^{m}_{col}$ with an even parity $Par(B)= 0$ such that $U(B)$ has a characteristic vector $\vec{b'}$ and a characteristic operator $V$.
\end{lemma}

In other words, from a column operator $U$ with a characteristic vector $\vec{b}$, one can always create a column operator $U'$ with larger characteristic vector $\vec{b'}$ by taking a product of translations of $U$. On the other hand, it is impossible to create a column operator with a smaller characteristic vector from an operator with a larger characteristic vector. Therefore, \emph{one can create a column operator with higher symmetries (a larger characteristic vector), but cannot create a column operator with lower symmetries (a smaller characteristic vector).} 

\subsection{Proof of lemma~\ref{lemma_construction}}
 
Now, we prove lemma~\ref{lemma_construction} by explicitly finding a binary column $B \in \textbf{B}^{m}_{col}$ for creating $U(B)$ for every pair of $\vec{b}$ and $\vec{b'}$. In order to derive such binary matrices, we introduce a certain binary column $B(\vec{b})$, called a \emph{characteristic column}, which can be used to change the characteristic vector of a column operator.

\textbf{Characteristic column:}
Given an integer $p \in \mathbb{Z}_{2^{m}}$, one may have its binary representation by considering the inverse of $g$ denoted as $g^{-1}$:
\begin{align}
\vec{p} \ \equiv \ (p_{1},\cdots,p_{m}) \ \equiv \ g^{-1}(p) 
\end{align}
where $p = g(\vec{p}) = \sum_{j=1}^{m-1} p_{j}2^{j-1}$. Now, we define the following sets:
\begin{align}
\textbf{J}_{\vec{b}} \ = \ \{ \ \vec{a}\ \in \ \textbf{B}^{m}_{vec}  \ : \ a_{j} \ \leq \ b_{j} \ \mbox{for all} \ j \ \}.
\end{align}
For example, $\textbf{J}_{(1,0)} = \{(0,0),(1,0)\}$ and $\textbf{J}_{(0,0,1)} = \{(0,0,0),(0,0,1)\}$. 

Based on $\textbf{J}_{\vec{b}}$, we define the \emph{characteristic binary column} $B(\vec{b}) \in \textbf{B}^{m}_{col}$ as follows:
\begin{align}
B(\vec{b})_{1,p+1} \ &= \ 1 \qquad \vec{p}\ \in \ \textbf{J}_{\vec{b}}  \\
B(\vec{b})_{1,p+1} \ &= \ 0 \qquad \mbox{otherwise}.
\end{align}
Here, we give some examples:
\begin{equation}
\begin{split}
&\textbf{J}_{(0,0)} \ = \ \{ (0,0) \}, \quad \textbf{J}_{(1,0)} \ = \ \{ (0,0), (1,0) \},  \quad \textbf{J}_{(0,1)} \ = \ \{ (0,0),(0,1) \}\\ &\textbf{J}_{(1,1)} \ = \ \{ (0,0),(1,0),(0,1),(1,1) \}
\end{split}
\end{equation}
and
\begin{align}
B(0,0) \ = \ \begin{bmatrix}
1 \\
0 \\
0 \\
0
\end{bmatrix}, \quad
B(1,0) \ = \ \begin{bmatrix}
1 \\
1 \\
0 \\
0
\end{bmatrix}, \quad
B(0,1) \ = \ \begin{bmatrix}
1 \\
0 \\
1 \\
0
\end{bmatrix}, \quad
B(1,1) \ = \ \begin{bmatrix}
1 \\
1 \\
1 \\
1
\end{bmatrix}.
\end{align}
One may see the relation between a characteristic column and a characteristic vector: 
\begin{align}
\begin{bmatrix} 
V \\
I \\
I \\
I
\end{bmatrix} \ \rightarrow \ (0,0), \quad
\begin{bmatrix} 
V \\
V \\
I \\
I
\end{bmatrix} \ \rightarrow \ (1,0), \quad
\begin{bmatrix} 
V \\
I \\
V \\
I
\end{bmatrix} \ \rightarrow \ (0,1), \quad
\begin{bmatrix} 
V \\
V\\
V\\
V
\end{bmatrix} \ \rightarrow \ (1,1).
\end{align}
Thus, if we replace $1$ entries in $B(\vec{b})$ with $V$ and create a column operator, it has a characteristic vector $\vec{b}$. 

In order to discuss changes of characteristic vectors, let us introduce the summation rule between binary vectors. We denote a summation of binary vectors $\vec{a}, \vec{b}\in \textbf{B}^{m}_{vec}$ as $\vec{a} + \vec{b}\in \textbf{B}^{m}_{vec}$ and define it as follows:
\begin{align}
g(\vec{a}) + g(\vec{b}) \ = \ g(\vec{a} + \vec{b})
\end{align}
when $g(\vec{a} + \vec{b}) \ \leq 2^{m}-1$. Therefore, $\vec{a} + \vec{b}$ is just like a summation of two binary ``numbers'' $\vec{a}$ and $\vec{b}$.

The characteristic columns defined above can be used to change a characteristic vector of a column operator, as summarized in the following lemma.

\begin{lemma}\label{lemma_property}
Let $\vec{b} < \vec{b'}$. When $U$ has a characteristic value $\vec{b}$ and a characteristic operator $V$, 
\begin{align}
U' \ = \ U(B(\Delta\vec{b})) \quad \mbox{where} \quad \vec{b} + \Delta\vec{b} \ = \ \vec{b'}
\end{align}
has a characteristic value $\vec{b'}$ and a characteristic operator $V$.
\end{lemma}

Below, we present a proof of this lemma by finding some property of characteristic columns and a certain rule on multiplications of column operators.

\textbf{Property of characteristic columns:}
There is a useful relation between a summation of vectors and characteristic columns, as summarized in the following lemma.

\begin{lemma}\label{lemma_summation}
Let 
\begin{align}
B(\vec{a})*B(\vec{b}) \ \equiv \ \sum_{j=1}^{2^{m}} T_{2}(B(\vec{a})^{B(\vec{b})_{1,j}})^{j-1} \qquad \mbox{(mod 2)}.
\end{align}
Then, 
\begin{align}
B(\vec{a} + \vec{b}) \ = \ B(\vec{a})*B(\vec{b}).
\end{align}
\end{lemma}

The proof involves some exercises on elementary math. 

\begin{proof}
We begin by defining a summation of sets of binary vectors. For $\textbf{B}_{1},\textbf{B}_{2},\cdots,\textbf{B}_{\alpha} \ \subseteq \ \textbf{B}^{m}_{vec}$ where $\alpha$ is some positive integer, and for $\vec{b} \in \textbf{B}^{m}_{vec}$, consider decompositions:
\begin{align}
\vec{b} \ = \ \vec{b_{1}} + \cdots + \vec{b_{\alpha}}, \qquad \vec{b_{i}} \ \in \ \textbf{B}_{i} \quad \mbox{for all} \ i
\end{align}
and denote the number of different decompositions of $\vec{b}$ as $N(\vec{b};\textbf{B}_{1},\textbf{B}_{2},\cdots,\textbf{B}_{\alpha})$. Then, we define the following summation:
\begin{align}
\textbf{B}_{1} + \textbf{B}_{2} + \cdots + \textbf{B}_{\alpha} \ \equiv \ \{ \ \vec{b} \ \in \ \textbf{B}^{m}_{vec} \ : \ N(\vec{b};\textbf{B}_{1},\textbf{B}_{2},\cdots, \textbf{B}_{\alpha}) \ = \ \mbox{odd}\ \}.
\end{align}

With the summation defined above, the claim of the lemma can be written as follows. By setting $B = B(\vec{a})*B(\vec{b})$, we may notice that
\begin{align}
B_{1, p+1} \ &= \ 1 \qquad \vec{p} \ \in \ \textbf{J}_{\vec{b}} + \textbf{J}_{\vec{b'}}\\
B_{1, p+1} \ &= \ 0 \qquad \mbox{otherwise}
\end{align}
from a direct calculation. Therefore, we need to show that
\begin{align}
\textbf{J}_{\vec{a}} + \textbf{J}_{\vec{b}} \ = \ \textbf{J}_{\vec{a} + \vec{b}}.
\end{align}

The proof relies on the following sublemma.

\begin{sublemma}
Consider
\begin{align}
{}_\alpha C_{\beta} \ \equiv \ \frac{\alpha!}{\beta!(\alpha-\beta)!}
\end{align}
for ($2^{m} > \alpha \geq \beta \geq 1$). Let the binary representations of $\alpha$ and $\beta$ be
\begin{align}
\vec{\alpha} \ = \ (\alpha_{1},\cdots, \alpha_{m}), \qquad  \vec{\beta} \ = \ (\beta_{1},\cdots, \beta_{m})
\end{align}
where $\alpha = \sum_{i=1}^{m}\alpha_{i}2^{i-1}$ and $\beta = \sum_{i=1}^{m}\beta_{i}2^{i-1}$. Then, ${}_\alpha C_{\beta}$ is odd if and only if
\begin{align}
\beta_{i} \ \leq \ \alpha_{i} \qquad \mbox{for all} \ i. 
\end{align}
\end{sublemma}

We suspect that the sublemma above has been proven somewhere else as it seems elementary. Yet, we could not find a reference, and thus, we present a proof here.

\begin{proof}
For a given integer $p \in \mathbb{Z}_{2^{m}}$, let $h(p)$ be the largest integer such that $\frac{p!}{2^{h(p)}}$ is an integer. Then, with some speculations, one may notice that
\begin{align}
h(p) \ = \ \sum_{i=1}^{m} (2^{i-1}-1)p_{i}
\end{align}
where $\vec{p} = (p_{1},\cdots,p_{m})$. Here,  
\begin{align}
{}_\alpha C_{\beta} \ = \ \frac{\alpha!}{\beta!(\alpha-\beta)!},
\end{align}
and 
\begin{align}
h({}_\alpha C_{\beta}) \ = \ h(\alpha) - h(\beta)  -  h(\alpha - \beta) \ \geq \ 0.
\end{align}
Then, ${}_\alpha C_{\beta}$ is odd if and only if
\begin{align}
h(\beta) + h(\alpha-\beta) \ = \ h(\alpha).
\end{align}
Let $\alpha - \beta \equiv \gamma$. Then, we have 
\begin{align}
\sum_{i=1}^{m} (2^{i-1}-1)\beta_{i} + \sum_{i=1}^{m} (2^{i-1}-1)\gamma_{i} \ = \ \sum_{i=1}^{m}  (2^{i-1}-1)\alpha_{i}
\end{align}
and $\beta_{i} + \gamma_{i}=\alpha_{i}$ for all $i$. This is true if and only if
\begin{align}
\beta_{i} \ \leq \ \alpha_{i} \qquad \mbox{for all} \quad i.
\end{align}
This completes the proof of the sublemma.
\end{proof}

Now, let us return to the proof of the lemma. Below, we prove $\textbf{J}_{\vec{a}} + \textbf{J}_{\vec{b}} \ = \ \textbf{J}_{\vec{a} + \vec{b}}$. Let $\vec{e_{1}} \equiv (1,0,\cdots,0)$. First, we show that 
\begin{align}
\underbrace{\textbf{J}_{\vec{e_{1}}} +\cdots + \textbf{J}_{\vec{e_{1}}}}_{\alpha} \ = \ \textbf{J}_{\vec{\alpha}}
\end{align}
where
\begin{align}
\vec{\alpha} \ \equiv \ \underbrace{ \vec{e_{1}} +\cdots + \vec{e_{1}} }_{\alpha}.
\end{align}
Here, notice that 
\begin{align}
\vec{\beta} \ \in \  \underbrace{ \textbf{J}_{\vec{e_{1}}} +\cdots + \textbf{J}_{\vec{e_{1}}} }_{\alpha}
\end{align}
if and only if ${}_{\alpha} C_{\beta}$ is odd since $\textbf{J}_{ \vec{e_{1}} }$ has two elements: $(1,0,\cdots,0)$ and $(0,\cdots,0)$. Thus, $\vec{\beta} \in \textbf{J}_{\vec{\alpha}}$ from the sublemma, and we have
\begin{align}
\underbrace{\textbf{J}_{\vec{e_{1}}} +\cdots + \textbf{J}_{\vec{e_{1}}}}_{\alpha} \ = \ \textbf{J}_{\vec{\alpha}}.
\end{align}

Now, let us show that
\begin{align}
\textbf{J}_{\vec{a}} + \textbf{J}_{\vec{b}} \ = \ \textbf{J}_{\vec{a} + \vec{b}}.
\end{align}
Note that $\vec{c} \in \textbf{J}_{\vec{a} + \vec{b}}$ if and only if ${}_{a+b} C_c$ is odd. Here, we have
\begin{align}
{}_{a+b} C_c \ = \ \sum_{i=0}^{c} {}_{a} C_i \cdot  {}_{b} C_{c-i}
\end{align}
where ${}_{x} C_{y}=0$ when $y>x$. Let us assume that ${}_{a} C_i \cdot  {}_{b} C_{c-i}$ is odd for $i = \alpha_{1},\cdots,\alpha_{p}$. Notice that ${}_{a} C_i \cdot  {}_{b} C_{c-i}$ is odd if and only if both ${}_{a} C_i$ and ${}_{b} C_{c-i}$ are odd. Then, ${}_{a+b} C_c$ is odd if and only if $p$ is odd. Notice that $p$ is the number of decompositions of $\vec{c}$ such that
\begin{align}
\vec{c} \ = \ \vec{c_{1}} + \vec{c_{2}}\qquad (\vec{c_{1}} \ \in \ \textbf{J}_{\vec{a}} \quad \mbox{and} \quad \vec{c_{2}} \ \in \ \textbf{J}_{\vec{b}}).
\end{align}
Therefore, $p$ is odd if and only if $\vec{c} \in \textbf{J}_{\vec{a}} + \textbf{J}_{\vec{b}}$. This completes the proof of the lemma.
\end{proof} 

\textbf{Multiplication of column operators:} Next, let us consider how characteristic vectors change under multiplications of column operators.

\begin{lemma}\label{lemma_multiplication}
Consider the following column operators:
\begin{align}
U '' \ = \ UU'\ \not= \ I
\end{align}
where
\begin{align}
\begin{array}{cccc}
U   & \rightarrow &  V , &  \vec{b}\\
U'  & \rightarrow &  V' , & \vec{b'}\\
U'' & \rightarrow &  V'' , &  \vec{b''}
\end{array}.
\end{align}
Then,
\begin{itemize}
\item If $\vec{b} > \vec{b'}$, $\vec{b''}=\vec{b'}$ and $V'' = V'$.
\item If $\vec{b'} = \vec{b}$ and $V \not= V'$, $\vec{b''}=\vec{b}$ and $V'' = VV'$.
\item If $\vec{b'} = \vec{b}$ and $V = V'$, $\vec{b''} > \vec{b}$.
\end{itemize}
\end{lemma}

Below, we present the proof.


\begin{proof}
We start with the first claim. When $b_{m}'=1$, $b_{m}=1$ since $\vec{b}>\vec{b'}$. Then, $f_{0}(U)=f_{0}(U')=I$, and $f_{0}(U'')=I$. Thus, $b_{m}''=1$. When $b_{m}'=0$ and $b_{m}=0$, we have $f_{0}(U)\not=I$ and $f_{0}(U')\not=I$. Suppose $f_{0}(U'')=I$. Then, $f_{0}(U)=f_{0}(U')$, and $U^{(1)}=U'^{(1)}$. This means $\vec{b}=\vec{b'}$ which contradicts with $\vec{b}>\vec{b'}$. Thus, $f_{0}(U'')\not=I$, and $b_{m}''=0$. When $b_{m}'=0$ and $b_{m}=1$, we have $f_{0}(U)=I$ and $f_{0}(U')\not=I$, and $f_{0}(U'')\not=I$, and $b_{m}''=0$. In summary, $b_{m}'=b_{m}''$. One can repeat the same discussion and show $\vec{b}' = \vec{b}''$ and $V''=V$.

The second claim is easy to prove, so we shall skip the proof. Let us move to the third claim. Since 
\begin{align}
f_{b_{1}}\cdots f_{b_{m}}(U) \ = \ f_{b_{1}}\cdots f_{b_{m}}(U') \ = \ V,
\end{align}
we have $f_{b_{1}}\cdots f_{b_{m}}(U'')=  I$. Let the largest integer $i$ such that $f_{b_{i}}\cdots f_{b_{m}}(U'') = I$ be $i_{max}$. Then, $f_{b_{i_{max}+1}}\cdots f_{b_{m}}(U'')  \not= I$ and $b_{i_{max}}''=1$. If $b_{i_{max}}=1$, $f_{b_{i_{max}+1}}\cdots f_{b_{m}}(U'')  =  I$ which leads to a contradiction. Thus, $b_{i_{max}}=0$. 

Since $f_{b_{i_{max}+1}}\cdots f_{b_{m}}(U'')  \not= I$, we have 
\begin{align}
b_{i} \ &= \ b_{i}'' \qquad (i_{max}+1 \ \leq \ i)\\
b_{i_{max}} \ &= \ 0, \qquad b_{i_{max}}'' \ = \ 1. 
\end{align}
Thus, $\vec{b''}>\vec{b}$.
\end{proof}

\textbf{Proof of lemma~\ref{lemma_property}:}
Finally, let us finish the proof of lemma~\ref{lemma_property}. Consider the following column operator $V(\vec{b})$ which replaces $1$ entries in $B(\vec{b})$ with $V$ and $0$ entries with $I$:
\begin{align}
V(\vec{b})_{j} \ \equiv \ V^{B(\vec{b})_{1,j}}. 
\end{align}
Then, one may easily notice that $V(\vec{b})$ has a characteristic vector $\vec{b}$ with a characteristic operator $V$. Therefore, from lemma~\ref{lemma_summation}, we have
\begin{align}
V(\vec{b})(B(\vec{a})) \ = \ V(\vec{a}+\vec{b}).
\end{align}

For $U$ with $\vec{b}$ and $V$, we decompose $U$ with $V(\vec{b})$ as follows:
\begin{align}
U \ = \ V(\vec{b})U'.
\end{align}
When $U\not=V(\vec{b})$, $U'$ has a characteristic vector $\vec{b'}$ with $\vec{b'}>\vec{b}$ from lemma~\ref{lemma_multiplication}. One can repeat the same decomposition and obtain:
\begin{align}
U \ = \ V(\vec{b})V'(\vec{b'})V''(\vec{b''})\cdots
\end{align}
where $\vec{b}<\vec{b'}<\vec{b''}$. Then, from lemma~\ref{lemma_summation}, we have
\begin{align}
U(B(\vec{a})) \ = \ V(\vec{a}+\vec{b})V'(\vec{a}+\vec{b'})V''(\vec{a}+\vec{b''})\cdots.
\end{align}
Note that $V(\vec{c})(B(\vec{a}))=I$ if $a + c >2^{m}-1$. From lemma~\ref{lemma_multiplication}, $U(B(\vec{a}))$ is a column operator with a characteristic vector $\vec{a}+\vec{b}$ and a characteristic operator $V$. Note that $B(\vec{a})$ is an even column when $\vec{a}\not=(0,\cdots,0)$. This completes the proof of lemma~\ref{lemma_property}.

\subsection{The existence of an odd matrix for $m>1$}

Finally, let us proceed to the proof of lemma~\ref{lemma_existence}, to complete the proof of theorem~\ref{theorem_decomposition}. 

\textbf{Procedure:}
Consider a logical operator $\ell$ defined inside $P(x,2^{m},1)$ with $x >2v$. Let us represent characteristic values and characteristic operators for each column of $\ell$ as follows:
\begin{align}
\begin{array}{ccc}
U_{1}\ \rightarrow\ & \vec{b_{1}}, & V_{1} \\
U_{2}\ \rightarrow\ & \vec{b_{2}},& V_{2} \\
\vdots & \vdots & \vdots \\
U_{x}\ \rightarrow\  & \vec{b_{x}},& V_{x}.
\end{array}
\end{align}
Without loss of generality, we may assume that $\vec{b_{i}} \leq \vec{b_{i+1}}$ for all $i$ since permutations of columns do not affect parities of binary matrices. We define a set of integers such that $\vec{b_{i}}=\vec{b}$ as $\textbf{b}(\vec{b})$:
\begin{align}
\textbf{b}(\vec{b}) \ \equiv \ \{ \ i \ : \ \vec{b_{i}} \ = \ \vec{b} \ \}.
\end{align}
We denote a group generated by $\{V_{i}\}$ for $i \in \textbf{b}(\vec{b})$ as $\mathcal{V}_{\vec{b}}$:
\begin{align}
\mathcal{V}_{\vec{b}} \ = \ \left\langle \ \{ \ V_{i} \ : \ i \ \in \ \textbf{b}(\vec{b}) \ \} \ \right\rangle.
\end{align}

Since $x > 2v$, there always exists a set of integers $\textbf{A}$ such that 
\begin{align}
\prod_{i \in \textbf{A}} V_{i}\ = \ I. 
\end{align}
We denote the largest vector in $\{\vec{b_{i}}\}_{i\in \textbf{A}}$ as $\vec{b_{\alpha}}$ and the largest integer $i$ with $\vec{b_{i}} = \vec{b_{\alpha}}$ as $i = \alpha$. Here, we define the following binary matrix $E(i')$
\begin{align}
E(i') \ \equiv \ \begin{bmatrix}
E_{1,1} ,& \cdots ,& E_{x,1} \\
 \vdots  & \vdots  & \vdots \\
E_{1,2^{m}} ,& \cdots ,& E_{x,2^{m}}
\end{bmatrix}
\end{align}
such that
\begin{align}
E_{i,1} \ &= \ 1 \qquad (i \ = \ i' )\\
E_{i,j} \ &= \ 0 \qquad \mbox{otherwise}.
\end{align}
For example, 
\begin{align}
E(1) \ = \ \begin{bmatrix}
1 ,& 0 ,& \cdots ,& 0 \\
 0 ,& 0 ,& \cdots ,& 0 \\
 \vdots & \vdots & \vdots  & \vdots \\
0  ,& 0,& \cdots ,& 0 
\end{bmatrix},\qquad E(2) \ = \ \begin{bmatrix}
0 ,& 1 ,& 0,& \cdots ,& 0 \\
 0 ,& 0 ,& 0,& \cdots ,& 0 \\
 \vdots & \vdots & \vdots & \vdots  & \vdots \\
0  ,& 0,& 0 ,& \cdots ,& 0 
\end{bmatrix}.
\end{align}
Notice that
\begin{align}
Par(E(i))_{j} \ = \ \delta_{i,j}. 
\end{align}

Now, we define the following operation between a binary column $A \in \textbf{B}^{m}_{col}$ and a binary matrix $B$:
\begin{align}
B * A \ \equiv \ \sum_{j=1}^{2^{m}} T_{2}(B^{A_{1,j}})^{j-1}  \qquad ( \mbox{mod 2}).
\end{align}
Then, consider the following binary matrix:
\begin{align}
E \ = \ \sum_{i \in \textbf{A}} E(i) * B(\Delta \vec{b_{i}}) 
\end{align}
where $\vec{b_{i}} + \Delta \vec{b_{i}} = \vec{b_{\alpha}}$. \emph{Note that $E(i) * B(\Delta \vec{b_{i}}) $ is odd if and only if $\Delta \vec{b_{i}} = (0,\cdots,0)$}. This matrix can generate the following column operator:
\begin{equation}
\begin{split}
\ell(E)_{x} \ &= \ \ell\left(\sum_{i \in \textbf{A}} E(i) * B(\Delta \vec{b_{i}}) \right)_{x}\\ 
          &= \ \prod_{i \in \textbf{A} } U_{i} \left( B(\Delta \vec{b_{i}})  \right).
\end{split}
\end{equation}
Note that 
\begin{align}
U_{i} \left( B(\Delta \vec{b_{i}})  \right) 
\end{align}
has a characteristic vector $\vec{b_{\alpha}}$ and a characteristic operator $V_{i}$ from lemma~\ref{lemma_summation}. Here, we notice that 
\begin{align}
Par(E)_{\alpha} \ = \ 1
\end{align}
and $E$ is an odd matrix since $\Delta \vec{b_{\alpha}} = (0,\cdots,0)$. Then, $\ell(E)_{x}\not=I$ since there is no odd identity generating matrix. Now, we notice that $\ell(E)_{x}$ is a column operator with a characteristic vector $\vec{b_{\alpha}'} > \vec{b_{\alpha}}$ from lemma~\ref{lemma_multiplication}. 

We summarize the discussion so far as follows.
\begin{itemize}
\item From $\textbf{A}$ such that $\prod_{i \in \textbf{A}}V_{i}=I$, one can form a column operator $\ell(E)_{x}$ which has a characteristic vector $\vec{b_{\alpha}'} > \vec{b_{\alpha}}$ and a characteristic operator $V'_{\alpha}$ where $E$ is an odd matrix which satisfies
\begin{align}
Par(E)_{\alpha} \ = \ 1 \quad \mbox{and} \quad
Par(E)_{j} \ = \ 0 \quad (j \ > \ \alpha).
\end{align}
\end{itemize}

\textbf{Update:}
Next, we ``update'' $U_{\alpha}$, $\vec{b_{\alpha}}$, $V_{\alpha}$ and $E(\alpha)$ to $\ell(E)_{x}$, $\vec{b_{\alpha}'}$, $V_{\alpha}'$ and $E$:
\begin{align}
\begin{array}{ccc}
U_{\alpha}         & \rightarrow & \ell(E)_{x}           \\
\vec{b_{\alpha}}   & \rightarrow & \vec{b_{\alpha}'}\\
V_{\alpha}         & \rightarrow & V_{\alpha}' \\
E(\alpha)          & \rightarrow & E
\end{array}
\end{align}
In other words, we replace $U_{\alpha}$, $\vec{b_{\alpha}}$, $V_{\alpha}$ and $E(\alpha)$ with $E(\alpha)$, $\ell(E)_{x}$, $\vec{b_{\alpha}'}$, $V_{\alpha}'$ and $E$, and rename them as $U_{\alpha}$, $\vec{b_{\alpha}}$, $V_{\alpha}$ and $E(\alpha)$. Note that these ``updated'' $E(i)$ satisfy
\begin{align}
Par(E(i))_{i} \ = \ 1 \quad \mbox{and} \quad Par(E(i))_{j} \ = \ 0\quad (j \ > \ i).
\end{align}

Then, one may repeat the discussion above. Since $x> 2v$, there is a set $\textbf{A}'$ such that
\begin{align}
\prod_{i \in \textbf{A}'} V_{i}\ = \ I.
\end{align}
We denote the largest vector in $\{\vec{b_{i}}\}_{i\in \textbf{A}'}$ as $\vec{b_{\beta}}$ and the largest integer $i$ with $\vec{b_{i}} = \vec{b_{\beta}}$ as $i = \beta$. Then, consider the following binary matrix:
\begin{align}
E' \ = \ \sum_{i \in \textbf{A}} E(i) * B(\Delta \vec{b_{i}}) \qquad \mbox{(mod 2)}
\end{align}
where $\vec{b_{i}} + \Delta \vec{b_{i}} = \vec{b_{\beta}}$. This matrix $E'$ can generate the following column operator:
\begin{equation}
\begin{split}
\ell(E')_{x} \ &= \ \ell\left(\sum_{i \in \textbf{A}'} E(i) * B(\Delta \vec{b_{i}}) \right)_{x} \\
&= \ \prod_{i \in \textbf{A}' } U_{i} \left( B(\Delta \vec{b_{i}})  \right).
\end{split}
\end{equation}
Here, we notice that 
\begin{align}
Par(E')_{\beta} \ = \ 1
\end{align}
since $\Delta\vec{b_{\beta}} = (0,\cdots,0)$ and $E'$ is an odd matrix. Then, $\ell(E')_{x}\not=I$. Note that 
\begin{align}
U_{i} \left( B(\Delta \vec{b_{i}}) \right) 
\end{align}
has a characteristic vector $\vec{b_{\beta}}$ and a characteristic operator $V_{i}$ from lemma~\ref{lemma_summation}. Then, we notice that $\ell(E')_{x}$ is a column vector with a characteristic vector $\vec{b_{\beta}'} > \vec{b_{\beta}}$. Note that 
\begin{align}
Par(E')_{\beta} \ = \ 0 \quad (j>\beta).
\end{align}

Then, we obtain the following observation.
\begin{itemize}
\item From $\textbf{A}'$ such that $\prod_{i \in \textbf{A'}'}V_{i}=I$ for ``updated'' $V_{i}$, one can form a column operator $\ell(E')_{x}$ which has a characteristic vector $\vec{b_{\beta}'} > \vec{b_{\beta}}$ and a characteristic operator $V'_{\beta}$. $E'$ is an odd matrix which satisfies
\begin{align}
Par(E')_{\beta} \ = \ 1 \quad \mbox{and} \quad
Par(E')_{j} \ = \ 0 \quad (j \ > \ \beta).
\end{align}
\end{itemize}

Here, we again ``update'' $U_{\beta}$, $\vec{b_{\beta}}$, $V_{\beta}$ and $E(\beta)$ to $\ell(E')_{x}$, $\vec{b_{\beta}'}$, $V_{\beta}'$ and $E'$. Then, since updated $E(i)$ always satisfy
\begin{align}
Par(E(i))_{i} \ = \ 1 \quad \mbox{and} \quad Par(E(i))_{j} \ = \ 0\quad (j>i)
\end{align}
one can repeat the same discussion again. In each update, characteristic vectors $\vec{b_{i}}$ increase, and at the end, one ends up with the following column operators
\begin{align}
\begin{array}{cccc}
U_{1}\ \rightarrow\ & \vec{1}, & V_{1} ,& E(1) \\
U_{2}\ \rightarrow\ & \vec{1},& V_{2} ,& E(2) \\
\vdots & \vdots & \vdots & \vdots \\
U_{x}\ \rightarrow\  & \vec{1},& V_{x} ,& E(x)
\end{array}
\end{align}
where $\vec{1} \equiv (1,\cdots,1)$ and 
\begin{align}
Par(E(i))_{i} \ = \ 1 \quad \mbox{and} \quad Par(E(i))_{j} \ = \ 0\quad (j>i)
\end{align}
Then, there exists a set $\textbf{A}$ such that 
\begin{align}
\prod_{i \in \textbf{A}}V_{i} \ = \ I
\end{align}
and, the following matrix is an identity generating matrix
\begin{align}
E \ = \  \sum_{i \in \textbf{A}} E(i) \qquad \mbox{(mod 2)}.
\end{align}
Let the largest integer in $\textbf{A}$ be $i_{max}$. Then, $E$ is odd since $Par(E)_{i_{max}}=1$. However, this contradicts with our original assumption that there is no odd identity generating matrix. This completes the proof of lemma~\ref{lemma_existence}, lemma~\ref{lemma_decomposition_final} and theorem~\ref{theorem_decomposition}.

\section{Derivation of logical operators}\label{sec:construction}

Having showed that two-dimensional logical operators can be decomposed as a product of two-dimensional and one-dimensional centralizer operators, let us proceed to the proof of theorem~\ref{theorem_3dim}. The proof owes a lot to arguments presented in~\cite{Beni10b}. For simplicity of presentation and in order to avoid making the paper unnecessarily long, we shall skip some parts of the derivation. However, we believe that interested readers can easily construct rigorous proofs.

\textbf{Preliminaries:} We begin by providing some corollaries and lemma which are useful in the derivations of logical operators. Let us first generalize theorem~\ref{theorem_decomposition} for any $n_{2}$. 

\begin{corollary}\label{corollary_decomposition}
Consider a three-dimensional STS model with the system size $n_{1}= 2\cdot2^{2n_{2}v}!$, arbitrary $n_{2}$ and $n_{3}>1$. 
For a given logical operator $\ell$ supported inside $P(n_{1},n_{2},1)$, one can decompose $\ell$ as a product of the following centralizer operators
\begin{align}
\ell \ \sim \ \ell_{a}\ell_{b}, \qquad \ell_{a},\ell_{b}\ \in \ \mathcal{C}_{P(n_{1},n_{2},1)}
\end{align}
where
\begin{align}
T_{1}^{\beta}(\ell_{b})\ = \ \ell_{b}, \qquad \mbox{where}\quad \beta \ \leq \ 2^{2n_{2}v}
\end{align}
and $\ell_{a}$ is defined inside $P(2v, n_{2},1)$.
\end{corollary}

The proof relies on the fact that one can make a logical operator $\ell_{a}$ ``quasi-periodic''.

\begin{proof}
Let us represent $n_{2}$ as $n_{2}=2^{m}\cdot n_{2}'$ where $n_{2}'$ is some odd integer. Then, for a logical operator $\ell$ defined inside $P(n_{1},n_{2},1)$, one can see that the following logical operator
\begin{align}
\ell' \ = \ \prod_{j=1}^{n_{2}'}T_{2}^{(j-1)2^{m}}(\ell) \ \sim \ \ell
\end{align}
is equivalent to $\ell$ since $\ell'$ is a product of an odd number of translations of $\ell$. Notice that $\ell'$ is periodic in the $\hat{2}$ direction:
\begin{align}
T_{2}^{2^{m}}(\ell') \ = \ \ell'
\end{align}
with the periodicity $2^{m}$. Because of this periodicity, one can form identity generating matrices in a way similar to the cases when $n_{2}=2^{m}$. Therefore, one can see that theorem~\ref{theorem_decomposition} holds for any $n_{2}$.
\end{proof}

We have seen that a two-dimensional logical operator can be decomposed as a product of a one-dimensional centralizer operator and a two-dimensional centralizer operator, as summarized in theorem~\ref{theorem_decomposition}. One can further decompose a one-dimensional logical operator as a product of a one-dimensional centralizer operator and a zero-dimensional centralizer operator, as summarized in the following lemma.

\begin{lemma}\label{lemma_decomposition_1dim}
Consider a three-dimensional STS model with the system size $n_{1}= 2\cdot2^{2n_{2}v}!$, $n_{2}=2\cdot 2^{(2v)^{2}}!$ and $n_{3}>1$ where $m$ is an arbitrary positive integer. For a given logical operator $\ell$ supported inside $P(2v,n_{2},1)$, one can decompose $\ell$ as a product of the following centralizer operators
\begin{align}
\ell \ \sim \ \ell_{a}\ell_{b}, \qquad \ell_{a},\ell_{b}\ \in \ \mathcal{C}_{P(2v,n_{2},1)}
\end{align}
where
\begin{align}
T_{1}^{\beta}(\ell_{b})\ = \ \ell_{b}, \qquad \mbox{where}\quad \beta \ \leq \ 2^{(2v)^{2}}
\end{align}
and $\ell_{a}$ is defined inside $P(2v, (2v)^{2},1)$.
\end{lemma}

We show the claim of lemma~\ref{lemma_decomposition_1dim} graphically in Fig.~\ref{fig_decomposition1D}. One can prove the lemma through discussion similar to the one used in the proof of lemma~\ref{lemma_decomposition}. So, we shall skip the proof.

\begin{figure}[htb!]
\centering
\includegraphics[width=0.45\linewidth]{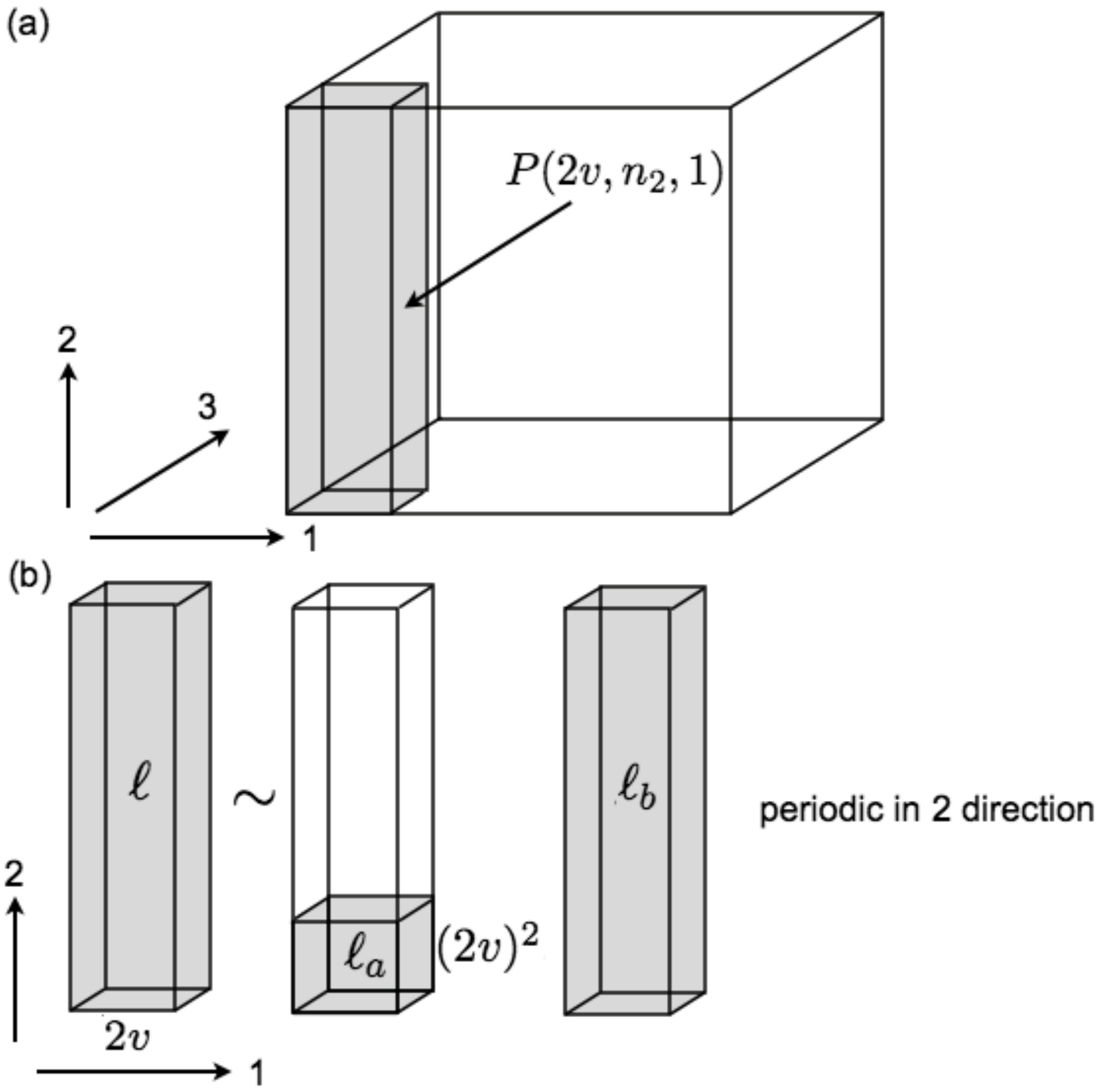}
\caption{The claim of lemma~\ref{lemma_decomposition_1dim}. 
} 
\label{fig_decomposition1D}
\end{figure}

Finally, let us extend the claim of theorem~\ref{theorem_2dim} slightly. 

\begin{corollary}\label{corollary_2dim}
Consider a two-dimensional STS model with even $n_{1}$ and $n_{2}$.
\begin{itemize}
\item Let $\ell$ be a logical operator which is periodic in the $\hat{1}$ and $\hat{2}$ directions:
\begin{align}
T_{1}(\ell) \ = \ T_{2}(\ell) \ = \ \ell.
\end{align}
Then, $\ell$ is a two-dimensional logical operator, and there exists a zero-dimensional logical operator $r$ which is defined inside $P(1,2v)$ and anti-commutes with $\ell$.
\item Let $\ell$ be a logical operator which is defined inside $P(1,n_{2})$ and periodic in the $\hat{2}$ direction:
\begin{align}
T_{2}(\ell) \ = \ \ell.
\end{align}
Then, $\ell$ is a one-dimensional logical operator, and there exists another one-dimensional logical operator $r$ which is defined inside $P(n_{1},1)$, anti-commutes with $\ell$: $\{\ell,r\}=0$ and periodic in the $\hat{1}$ direction:
\begin{align}
T_{1}(r) \ = \ r.
\end{align}
\end{itemize}
\end{corollary}

In other words, if we find a logical operator which is periodic in both $\hat{1}$ and $\hat{2}$ directions, we readily know that it is a two-dimensional logical operator. Similarly, if we find a logical operator which is defined inside $P(1,n_{2})$ and periodic in the $\hat{2}$ direction, we readily know that it is a one-dimensional logical operator. Note that the corollary holds only for the cases where both $n_{1}$ and $n_{2}$ are even. 

\begin{proof}
Since $\ell$ is periodic and system sizes are even, $\ell$ commutes with all the one-dimensional logical operators and all the two-dimensional logical operators. Therefore, $\ell$ can anti-commute only with zero-dimensional logical operators. This means that $\ell$ is a two-dimensional logical operator. The second claim can be proven in a similar way.
\end{proof}

Now, we derive all the logical operators in a three-dimensional STS models, as described in theorem~\ref{theorem_3dim}. Consider the system size analyzed in lemma~\ref{lemma_decomposition_1dim}: $(n_{1},n_{2},n_{3})= ( 2\cdot2^{2n_{2}v}!, 2\cdot 2^{(2v)^{2}}!, n_{3})$ where $n_{3}>1$ is some fixed integer. Here, we view the entire system as a two-dimensional system by considering $P(1, n_{2},1)$ as a single composite particle (Fig.~\ref{fig_2Dview}). (So, the entire system is viewed as a two-dimensional lattice of one-dimensional tubes). 

For a two-dimensional STS model, we already know the geometric shapes of all the logical operators, as summarized in theorem~\ref{theorem_2dim}. When viewed as a two-dimensional system, ``zero-dimensional'' logical operators are defined inside $P(2vn_{2},n_{2},1)$. From theorem~\ref{theorem_decomposition}, we notice that these ``zero-dimensional'' logical operators can be actually defined inside $P(2v,n_{2},1)$. Then, we have the following logical operators (Fig.~\ref{fig_2Dview}). 

\begin{figure}[htb!]
\centering
\includegraphics[width=0.50\linewidth]{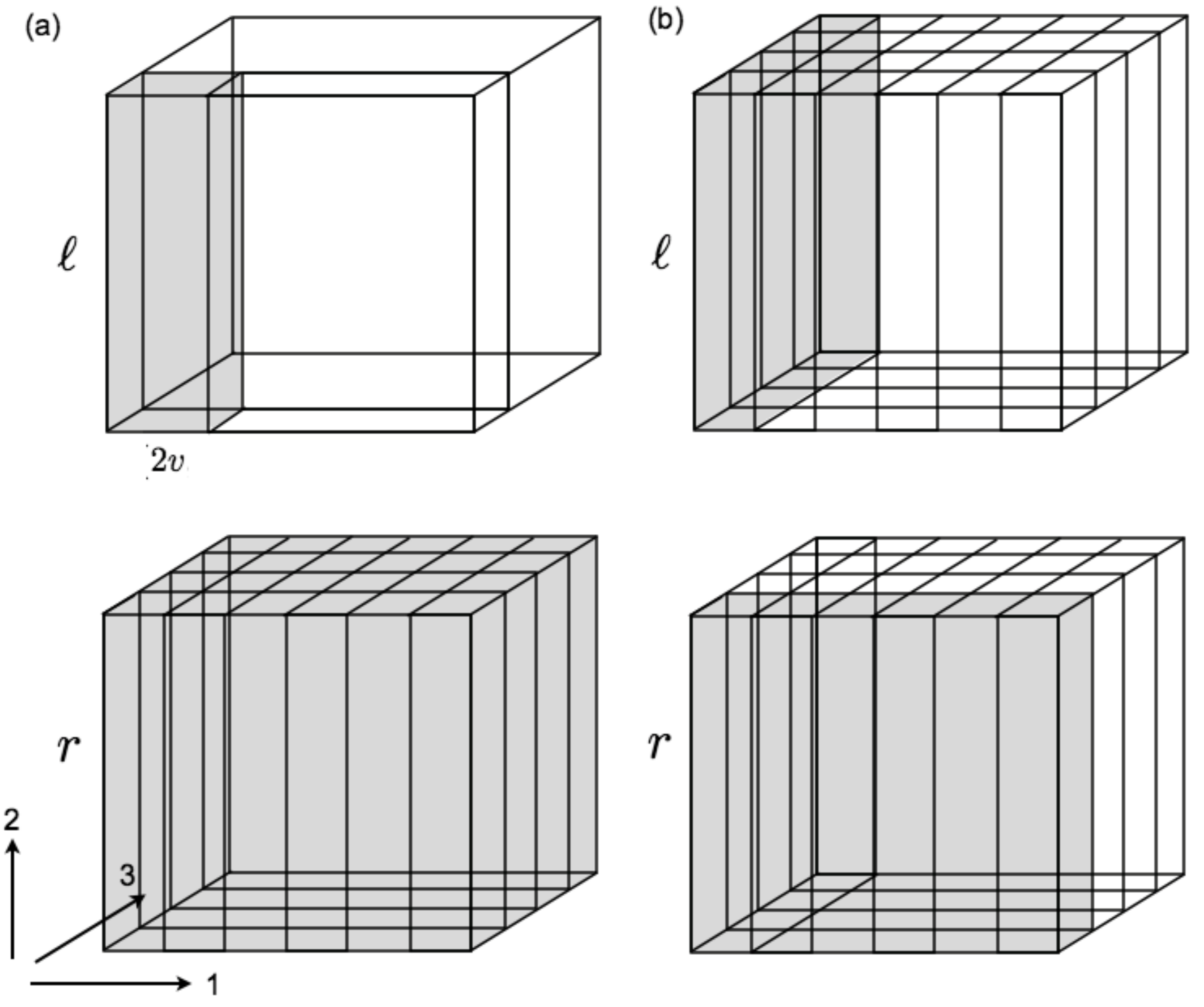}
\caption{Viewed as a two-dimensional system.
} 
\label{fig_2Dview}
\end{figure}

\begin{itemize}
\item An anti-commuting pair of logical operators in Fig.~\ref{fig_2Dview}(a) where $\ell$ is defined inside $P(2v,n_{2},1)$ and $r$ is periodic:
\begin{align}
T_{1}(r) \ = \ T_{3}(r) \ = \ r. 
\end{align}
\item  An anti-commuting pair of logical operators in Fig.~\ref{fig_2Dview}(b) where $\ell$ is defined inside $P(1,n_{2},n_{3})$ and periodic:
\begin{align}
T_{3}(\ell) \ = \ \ell
\end{align}
and $r$ is defined inside $P(n_{1},n_{2},1)$ and periodic:
\begin{align}
T_{1}(r) \ = \ r.
\end{align}
\end{itemize}

Below, we analyze anti-commuting pairs of logical operators in Fig.~\ref{fig_2Dview}(a) and Fig.~\ref{fig_2Dview}(b), and derive logical operators. 


\textbf{Pairs in (a):}
Below, we analyze properties of logical operators described above. We start with anti-commuting pairs described in Fig.~\ref{fig_2Dview}$(a)$. We stop viewing the system as a two-dimensional system for the moment. Logical operators defined inside $P(2v,n_{2},1)$ consists of periodic one-dimensional logical operators and zero-dimensional logical operators defined inside $P(2v, (2v)^{2},1)$ due to lemma~\ref{lemma_decomposition_1dim}. Let us first analyze a zero-dimensional logical operator $\ell$ defined inside $P(2v, (2v)^{2},1)$. 

From theorem~\ref{theorem_2dim}, $\ell$ is also a logical operators for arbitrary $n_{1}$ and $n_{3}$. Consider the case when $n_{3}=1$ (Fig.~\ref{fig_proof_aid1}). Then, by viewing the system as a two-dimensional system which extends only in the $\hat{1}$ and $\hat{2}$ directions, one notices that there exists a two-dimensional logical operator $r$ which is periodic in both $\hat{1}$ and $\hat{2}$ directions:
\begin{align}
T_{1}(r) \ = \ T_{2}(r) \ = \ r
\end{align}
and anti-commutes with $\ell$: $\{ \ell,r\}=0$. Next, let us consider the case when $n_{3}>1$. Then, one may extend the construction of $r$ as follows (Fig.~\ref{fig_proof_aid1}):
\begin{align}
r \ \rightarrow \ r' \ =  \ \prod_{x=1}^{n_{3}}T_{3}^{x}(r).
\end{align}
In other words, we put $r$ in a periodic way in the $\hat{3}$ direction to form $r'$. We shall call such an extension the \emph{periodic extension}. The three-dimensional logical operator $r'$ obtained after the periodic extension of $r$ is periodic in all the directions:
\begin{align}
T_{1}(r') \ = \ T_{2}(r') \ = \ T_{3}(r') \ = \ r'
\end{align}
and anti-commutes with $\ell$ : $\{r',\ell \}=0$. Then, one may notice that $r'$ and $\ell$ form a pair of anti-commuting logical operators for any system size $\vec{n}$. From this discussion, we obtain the following observation (Fig.~\ref{fig_proof_aid1}).

\begin{itemize}
\item For zero-dimensional logical operators defined inside $P(2v, (2v)^{2},1)$, there always exists a three-dimensional logical operator $r$ which is periodic:
\begin{align}
T_{1}(r) \ = \ T_{2}(r) \ = \ T_{3}(r) \ = \ r
\end{align}
and anti-commutes with $\ell$: $\{\ell,r\}=0$. $\ell$ and $r$ are logical operators for any system size $\vec{n}$.
\end{itemize}

\begin{figure}[htb!]
\centering
\includegraphics[width=0.55\linewidth]{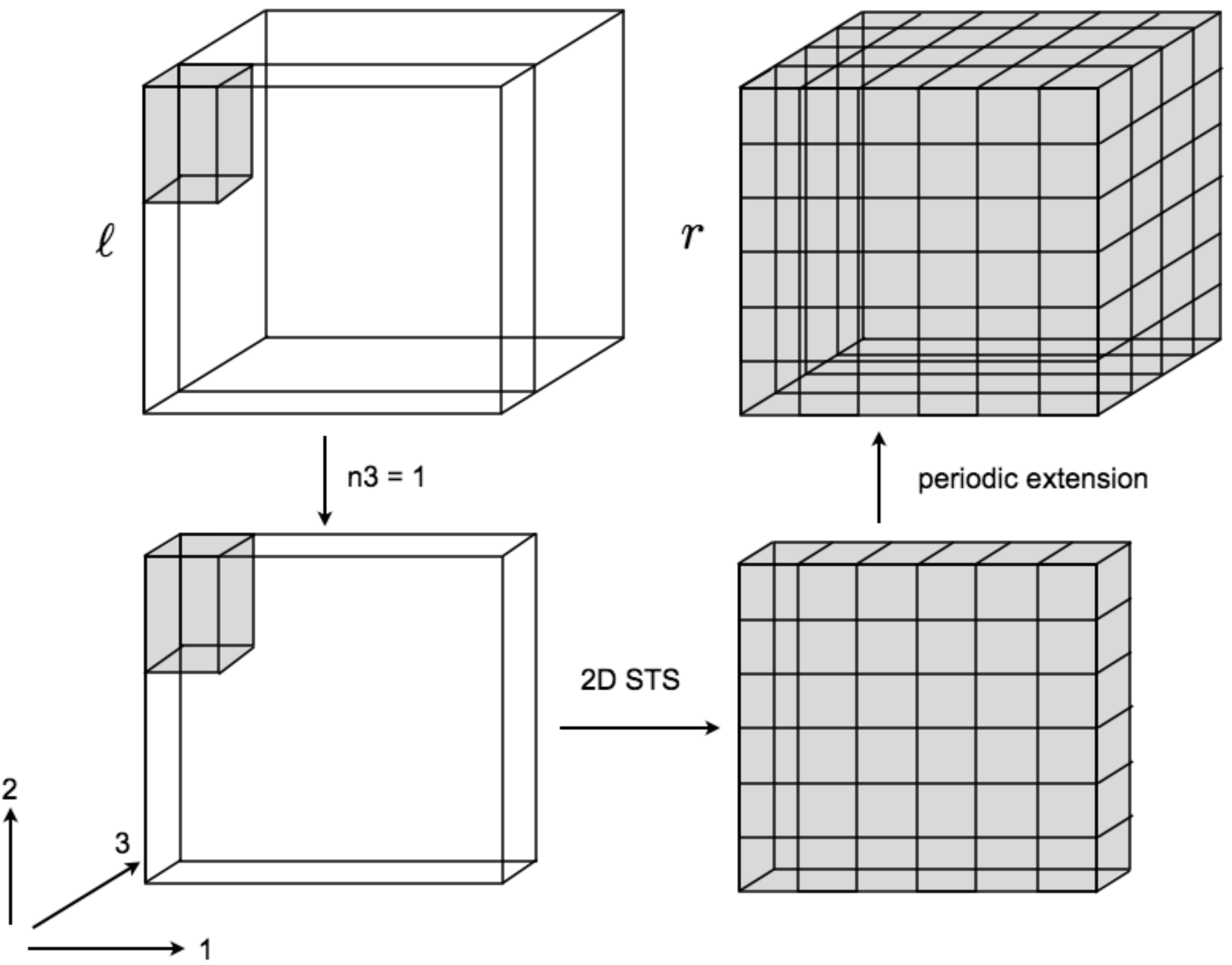}
\caption{Constructions of zero-dimensional and three-dimensional logical operators.
} 
\label{fig_proof_aid1}
\end{figure}

Next, let us consider a one-dimensional logical operator $\ell$ defined inside $P(2v,n_{2},1)$ which is periodic in the $\hat{2}$ direction:
\begin{align}
T_{2}^{\beta}(\ell) \ = \ \ell
\end{align}
for $\beta \leq 2^{(2v)^{2}}$. Recall that $r$ is periodic:
\begin{align}
T_{1}(r) \ = \ T_{3}(r) \ = \ r
\end{align}
and anti-commutes with $\ell$ : $\{r,\ell\}=0$. Then, one can decompose $r$ as a product of two centralizer operators from lemma~\ref{lemma_decomposition}:
\begin{align}
r \ \sim \ r_{a}r_{b}
\end{align}
where 
\begin{align}
T_{1}(r_{a}) \ = \ T_{3}(r_{a}) \ = \ r_{a}, \qquad T_{1}(r_{b}) \ = \ T_{3}(r_{b}) \ = \ r_{b}
\end{align}
and $r_{a}$ is defined inside $P(n_{1},2v,n_{3})$, and 
\begin{align}
T_{2}^{\beta'}(r_{b}) \ = \ r_{b}
\end{align}
for $\beta' \leq 2^{2v}$. Then, we notice that
\begin{align}
[r_{b},\ell] \ = \ 0
\end{align}
since $T_{2}^{\beta}(\ell) = \ell$ and $T_{2}^{\beta'}(r_{b}) = r_{b}$, and $n_{2}/\beta\beta'$ is an even integer for $n_{2}=2\cdot 2^{(2v)^{2}}!$. Thus, we have
\begin{align}
\{r_{a},\ell\} \ = \ 0.
\end{align}

Since $r_{a}$ is periodic in the $\hat{1}$ and $\hat{3}$ directions, one can periodically extend its construction for arbitrary $n_{1}$ and $n_{3}$. Now, let us consider the system size such that $n_{1}$ and $n_{3}$ are odd. Here, note that $r_{a}$ has some equivalent logical operator $r_{a}'$ defined inside $P(n_{1},1,n_{3})$~\cite{Beni10b}. Then, the following operator
\begin{align}
r_{a}'' \ = \ \prod_{i,j}T_{1}^{i}T_{3}^{j}(r_{a}') \ \sim \ r_{a} 
\end{align}
is equivalent to $r_{a}$ (Fig.~\ref{fig_proof_aid2}). Note that $r_{a}''$ is defined inside $P(n_{1},1,n_{3})$ and periodic in the $\hat{1}$ and $\hat{3}$ directions:
\begin{align}
T_{1}(r_{a}'') \ = \ T_{3}(r_{a}'') \ = \ r_{a}''.
\end{align}
Finally, we show that $\ell$ can be periodic in the $\hat{2}$ direction. Consider the case when $n_{1}$ and $n_{3}$ are even, and $n_{2}=1$. Then, $r_{a}''$ is also a logical operators. Now, there always exists some logical operator $\ell'$ defined inside $P(1,2v,1)$ which anti-commutes with $r_{a}''$ from corollary~\ref{corollary_2dim}. One can periodically extend its construction to arbitrary $n_{2}$. We denote it $\ell''$. Then, $\ell''$ and $r_{a}''$ are logical operators for any system size. From this discussion, we obtain the following observation (Fig.~\ref{fig_proof_aid2}).

\begin{itemize}
\item For one-dimensional logical operators defined inside $P(2v, n_{2},1)$, there exists a two-dimensional logical operator $r$ defined inside $P(n_{1},1,n_{3})$ which is periodic:
\begin{align}
T_{1}(r) \ = \ T_{3}(r) \ = \ r
\end{align}
and anti-commutes with $\ell$: $\{\ell,r\}=0$. $\ell$ can be also periodic:
\begin{align}
T_{2}(\ell) \ = \ \ell,
\end{align}
and, $\ell$ and $r$ are logical operators regardless of the system size $\vec{n}$.
\end{itemize}

\begin{figure}[htb!]
\centering
\includegraphics[width=0.55\linewidth]{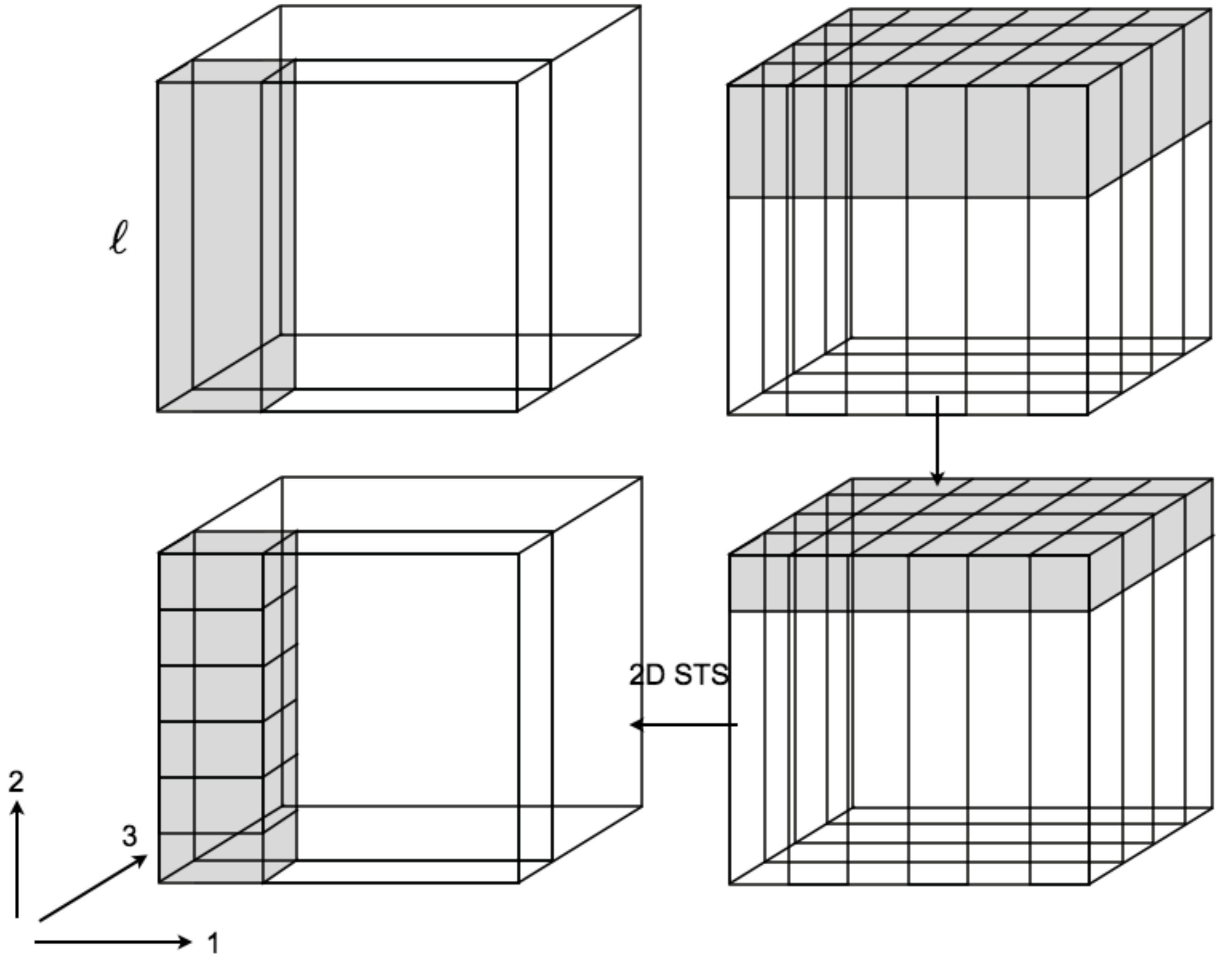}
\caption{Constructions of one-dimensional and two-dimensional logical operators.
} 
\label{fig_proof_aid2}
\end{figure}

\textbf{Pairs in (b):}
Let us proceed to the analysis on pairs of logical operators in Fig.~\ref{fig_2Dview}(b). We consider the following anti-commuting logical operators $\ell$ and $r$.

\begin{itemize}
\item  $\ell$ is defined inside $P(1,n_{2},n_{3})$ and periodic: $T_{3}(\ell) = \ell$.
\item  $r$ is defined inside $P(n_{1},n_{2},1)$ and periodic: $T_{1}(r) = r$.
\end{itemize}

Since $\ell$ is periodic in the $\hat{3}$ direction, one can decompose $\ell$ as follows:
\begin{align}
\ell \ \sim \ \ell_{a} \ell_{b}, \qquad T_{3}(\ell_{a}) \ = \ \ell_{a} \quad \mbox{and} \quad  T_{3}(\ell_{b}) \ = \ \ell_{b}
\end{align}
where $\ell_{a}$ is defined inside $P(1,2v,n_{3})$, and $\ell_{b}$ is defined inside $P(1,n_{2},n_{3})$ and periodic:
\begin{align}
T_{2}^{\beta}(\ell_{b}) \ = \ \ell_{b}
\end{align}
where $\beta \leq 2^{2v}$. Thus, logical operators defined inside $P(1,n_{2},n_{3})$ consist of two-dimensional logical operators and one-dimensional logical operators defined inside $P(1,2v,n_{3})$.

Let us analyze a one-dimensional logical operator $\ell$ defined inside $P(1,2v,n_{3})$ first. We decompose $r$ defined inside $P(n_{1},n_{2},1)$ as follows:
\begin{align} 
r \ \sim \ r_{a} r_{b}
\end{align}
where $r_{a}$ is defined inside $P(n_{1},2v,1)$, and $r_{b}$ is defined inside $P(n_{1},n_{2},1)$ and periodic:
\begin{align}
T_{2}^{\beta'}(r_{b}) \ = \ r_{b}
\end{align}
where $\beta'\leq 2^{2v}$. Then, we notice that
\begin{align}
[\ell,r_{a}] \ = \ 0
\end{align}
since there exists a translation of $r_{a}$ which does not overlap with $\ell$. Thus, we have
\begin{align}
\{\ell,r_{b}\} \ = \ 0.
\end{align}

Let us consider the case where $n_{1}=1$, and $n_{3}$ is even. Note that $\ell$ and $r_{b}$ are both logical operators. Then, from corollary~\ref{corollary_2dim}, there exists a one-dimensional logical operator $r'$ which is defined inside $P(1,n_{2},1)$, anti-commutes with $\ell$ and is periodic in the $\hat{2}$ direction:
\begin{align}
T_{2}(r') \ = \ r'.
\end{align}
Then, we periodically extend $r'$ in the $\hat{1}$ direction and define $r''$. Then, we notice that $r''$ is defined inside  $P(n_{1},n_{2},1)$ and periodic in the $\hat{1}$ and $\hat{2}$ directions. From this discussion, we obtain the following observation.

\begin{itemize}
\item For a one-dimensional logical operator $\ell$ defined inside $P(1,2v,n_{3})$, there exists a two-dimensional logical operator $r$ defined inside $P(n_{1},n_{2},1)$ which is periodic:
\begin{align}
T_{1}(r) \ = \ T_{2}(r) \ = \ r
\end{align}
and anti-commutes with $\ell$: $\{\ell,r\}=0$. $\ell$ can be also periodic:
\begin{align}
T_{3}(\ell) \ = \ \ell,
\end{align}
and, $\ell$ and $r$ are logical operators regardless of the system size $\vec{n}$.
\end{itemize}

Finally, let us analyze a two-dimensional logical operator $\ell$ defined inside $P(1,n_{2},n_{3})$ with
\begin{align}
T_{2}^{\beta}(\ell) \ = \ \ell. 
\end{align}
Since $r$ is periodic in the $\hat{1}$ direction, one can decompose it as follows:
\begin{align}
r \ \sim \ r_{a} r_{b}, \qquad T_{1}(r_{a}) \ = \ r_{a} \quad \mbox{and} \quad  T_{1}(r_{b}) \ = \ r_{b}
\end{align}
where $r_{a}$ is defined inside $P(n_{1},2v,1)$ and $r_{b}$ are defined inside $P(n_{1},n_{2},1)$ and periodic:
\begin{align}
T_{2}^{\beta'}(r_{b}) \ = \ r_{b}
\end{align}
where $\beta' \leq 2^{2v}$. Then, one may notice that
\begin{align}
\{ \ell, r_{b} \} \ = \ 0.
\end{align}
Then, the rest is immediate, and we obtain the following observation.

\begin{itemize}
\item For a two-dimensional logical operator $\ell$ defined inside $P(1,n_{2},n_{3})$, there exists a one-dimensional logical operator $r$ defined inside $P(n_{1},2v,1)$ which is periodic:
\begin{align}
T_{1}(r) \ = \ r
\end{align}
and anti-commutes with $\ell$: $\{\ell,r\}=0$. $\ell$ can be also periodic:
\begin{align}
T_{2}(\ell) \ = \ T_{3}(\ell) \ = \ \ell,
\end{align}
and, $\ell$ and $r$ are logical operators regardless of the system size $\vec{n}$.
\end{itemize}

Let us recall the discussion so far. We started our analysis with the system size considered in lemma~\ref{lemma_decomposition_1dim}. Then, for each pair of anti-commuting logical operators, we found logical operators whose geometric shapes are the same as the ones described as in theorem~\ref{theorem_3dim}. These logical operators are also logical operators for other system sizes because of their periodic structures and scale symmetries of the system. 

It remains to sort these logical operators in a canonical form by analyzing their commutation relations. However, we shall skip this process since it is straightforward. See~\cite{Beni10b} for a similar discussion. This completes the proof of theorem~\ref{theorem_3dim}. 

\section*{Acknowledgments}

I thank Eddie Farhi and Peter Shor for support at MIT. I thank Sergey Bravyi, Jeongwan Haah, Isaac Kim, Zhenghang Wang and Sam Ocko for stimulating and fruitful discussion. 
I thank Masahito Ueda and Spiros Michalakis for valuable comments. 
I gratefully acknowledge the hospitality of Caltech IQI where I enjoyed interesting discussion with group members, and particularly thank John Preskill. 
This work is supported by the U.S. Department of Energy under cooperative research agreement Contract Number DE-FG02-05ER41360, and by  the Nakajima Foundation.

\end{document}